\documentclass{article}

\PassOptionsToPackage{numbers, compress}{natbib}

\usepackage[preprint]{neurips_2026}

\usepackage[utf8]{inputenc} 
\usepackage[T1]{fontenc}    
\usepackage{hyperref}       
\usepackage{url}            
\usepackage{booktabs}       
\usepackage{amsfonts}       
\usepackage{nicefrac}       
\usepackage{microtype}      
\usepackage[table,dvipsnames]{xcolor} 
\usepackage{amsmath}
\usepackage{amssymb}
\usepackage{mathtools}
\usepackage{algorithm}
\usepackage{algorithmic}
\newcommand{\gray}[1]{\textcolor{gray}{#1}}
\renewcommand{\algorithmiccomment}[1]{\hfill$\triangleright$ #1}
\newcommand{\Comment}[1]{\algorithmiccomment{#1}}

\usepackage{enumitem}

\usepackage{float}          
\usepackage{graphicx}       
\usepackage{wrapfig}        
\usepackage{subfig}         
\usepackage{multirow}       
\usepackage{makecell}       
\usepackage{amsthm}         
\usepackage{mathrsfs}       
\usepackage{pifont}         
\usepackage{etoc}           

\theoremstyle{plain}
\newtheorem{theorem}{Theorem}[section]
\newtheorem{proposition}[theorem]{Proposition}

\theoremstyle{definition}

\newtheorem{assumption}[theorem]{Assumption}

\theoremstyle{remark}

\newcommand{\lr}{\eta_{\mathrm{lr}}}
\newcommand{\venue}[1]{{\small\textcolor{gray!70!black}{[#1]}}}
\newcommand{\up}[1]{{\scriptsize\textcolor{ForestGreen}{($\uparrow$#1)}}}
\newcommand{\down}[1]{{\scriptsize\textcolor{red}{($\downarrow$#1)}}}
\newcommand{\gap}[1]{{\scriptsize\textcolor{gray!70!black}{(\phantom{$\uparrow$}#1)}}}
\newcommand{\best}[1]{\textcolor{red}{\textbf{\underline{#1}}}}
\newcommand{\second}[1]{\textcolor{blue}{#1}}

\title{\texttt{EASE}: Federated Multimodal Unlearning via Entanglement-Aware Anchor Closure}

\author{
Zihao Ding \quad Beining Wu \quad Jun Huang \\
Department of Electrical Engineering and Computer Science \\
South Dakota State University \\
Brookings, SD 57007, USA \\
\texttt{\{Zihao.Ding, Wu.Beining\}@jacks.sdstate.edu} \\
\texttt{Jun.Huang@sdstate.edu}
}

\begin{document}

\maketitle


\begin{abstract}
Federated Multimodal Learning (FML) trains multimodal models across decentralized clients while keeping their image-text pairs private. However, joint embedding training entangles forgotten knowledge across both modalities and client gradient subspaces, hindering federated unlearning. Previous federated unlearning approaches neither sever the cross-modal reconstruction channel mediated by bilinear coupling nor separate forget-exclusive update directions from those shared with retained clients. We identify an Anchor Principle for federated multimodal contrastive unlearning: forgotten alignments persist through three residual anchors arising from bilinear cross-modal coupling, principal-angle subspace entanglement, and continued federated updates. At the modality level, we show that bilateral displacement of both visual and language branches closes the cross-modal reconstruction channel. Correspondingly, our method addresses subspace entanglement through Cosine--Sine decomposition of client-update subspaces, isolating forget-exclusive directions from retain support. Moreover, we propose a direction-selective Forget Lock that bounds residual drift across rounds. Combining these strategies, we present \underline{\texttt{EASE}}, an \textbf{E}ntanglement-\textbf{A}ware \textbf{S}ubspace \textbf{E}xcision framework that closes all three anchor channels under a unified design. \underline{\texttt{EASE}} demonstrates consistent superiority across multiple datasets and unlearning scenarios, for instance, matching the retrain reference to within $0.2$ and $4.2$ R@1 points on the forget and retain sides under client unlearning on Flickr30K with CLIP-B/32.
\end{abstract}

\etocdepthtag.toc{mainbody}
\section{Introduction}
\label{sec:introduction}

Federated Learning (FL)~\cite{McMahan2017AISTATS, Li2020SPM,Wu2026COMST,Wu2025ToN,DingICNC2025,Wu2026ARXIV,Pudasaini2026HPSR} allows multiple clients to jointly train a shared model while keeping their raw data private, establishing itself as a key paradigm in distributed machine learning~\cite{Kairouz2021FNT, Yang2019TIST,Wu2026TNSE}. Recent advances integrate FL with multimodal models through parameter-efficient Low-Rank Adaptation (LoRA)~\cite{Hu2022ICLR} adapters, forming Federated Multimodal Learning (FML)~\cite{Chen2024AAAI, Yu2023ICLR,Wu2025MNET,Ding2025IPCCC}. The server never observes raw image-text pairs; it only coordinates client updates and later continues FedAvg with the remaining clients. When deployed in regulated domains, FML inherits the right to be forgotten~\cite{Voigt2017GDPR}: once a client withdraws, a target concept becomes obsolete, or harmful content surfaces, the system must approximate the retain-only retrain model without centralizing the removed pairs, a task referred to as Federated Unlearning~\cite{Liu2021IWQOS, Halimi2022ARXIV, Wang2022WWW,Ding2026ICDCS,Xing2026ACR,Wu2025WASA,Wu2026ARXIV1}. We review related work in Appendix~\ref{app:related_work}.

As illustrated in Figure~\ref{fig:intro_overview}, while federated unlearning has drawn growing interest~\cite{Liu2021IWQOS,Wu2023MPE, Halimi2022ARXIV, Wu2022ARXIV, Gao2024TDSC,Wu2026ICDCS,Huang2025TMC,Pan2023SCIS,DingICNC2025,Wu2023ACCESS}, extending it to multimodal models remains open: client-update-only access, the bilinear coupling of the alignment loss, and continued FedAvg jointly create failure modes absent from single-modality unlearning. We use anchor to denote a residual mechanism that keeps a forgotten image-text alignment reconstructible after an unlearning request, and identify three such mechanisms in this regime, a failure pattern we call an Anchor Principle for federated multimodal contrastive unlearning. The first arises because joint embedding training entangles forgotten knowledge across both modalities through the bilinear similarity, so the gradient with respect to one modality depends on the embeddings produced by the other. Excising parameters from one branch alone leaves the other as a cross-modal anchor that drives the alignment gradient back toward the original pairing, raising a first question: \textbf{\textit{I) How can we close this Modality Anchor while keeping each branch's representation faithful enough to support retrieval over the retained image-text pairs?}} Compounding this, federated training entangles client gradient subspaces along a continuous spectrum, mixing directions exclusive to the forget set with directions shared with retained clients. Existing methods treat the forget set's contribution as an indivisible block, conflating removable anchors with retain support, and as we confirm in Section~\ref{sec:preliminaries} they consistently suffer a forget--retain trade-off, motivating a second question: \textbf{\textit{II) How can we identify and excise only the removable Unique-Subspace Anchors while sparing the retain-support directions whose deletion would degrade retained knowledge across clients?}}

\begin{wrapfigure}[24]{r}{0.55\linewidth}
    \vspace{-14pt}
    \centering
    \includegraphics[width=\linewidth]{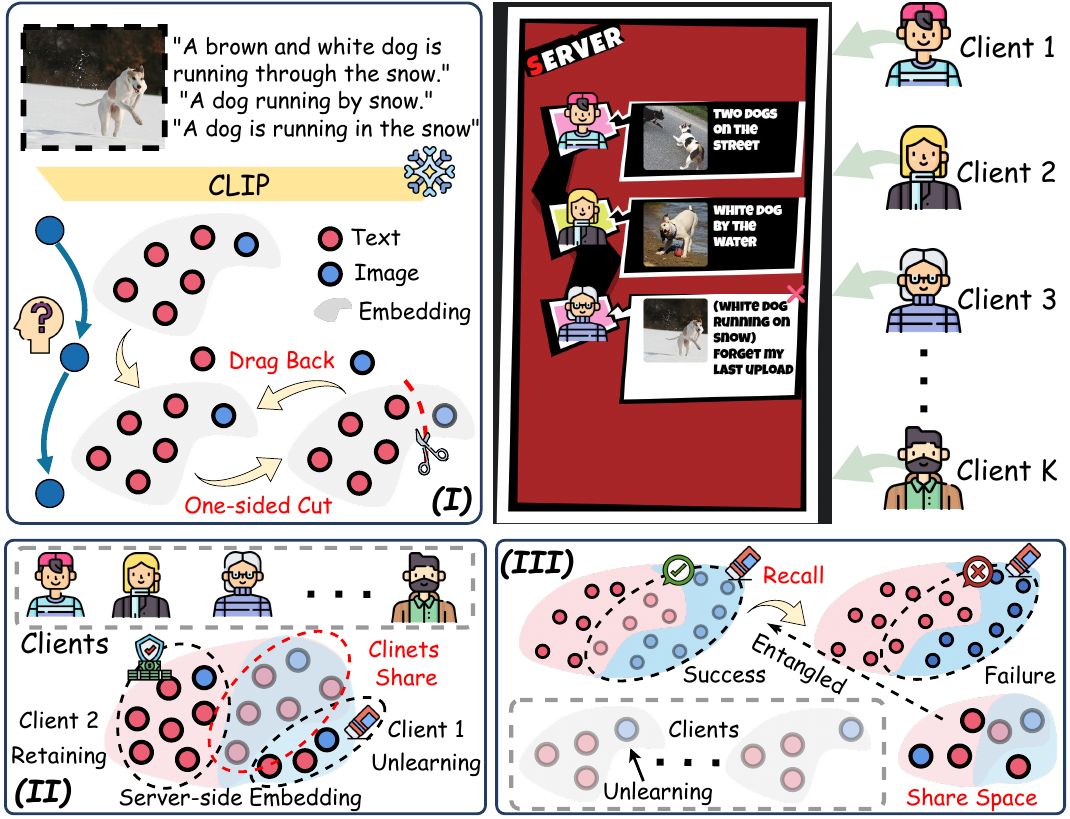}
    \caption{\textbf{Problem Illustration.} Three residual anchors in federated multimodal unlearning. \textbf{I)} Modality Anchor: the intact branch preserves forgotten alignments after one-sided excision. \textbf{II)} Unique-Subspace Anchor: forget-exclusive directions mix with retain support across client subspaces. \textbf{III)} Temporal Re-anchoring: continued aggregation rewrites the removed directions during subsequent FedAvg training rounds.}
    \label{fig:intro_overview}
    \vspace{-12pt}
\end{wrapfigure}

Even after the first two anchors are closed at the moment of unlearning, each client's local SGD follows an alignment gradient with non-zero components along the removed directions. Without a persistent safeguard, every subsequent round of federated training rewrites a fraction of the cut anchors back and gradually undoes the unlearning, raising a third question: \textbf{\textit{III) How can we prevent Temporal Re-anchoring across continued rounds while leaving directions outside the forget channel free to recover retained knowledge?}}

We address all three challenges with \underline{\texttt{EASE}}: an \textbf{E}ntanglement-\textbf{A}ware \textbf{S}ubspace \textbf{E}xcision framework for federated multimodal unlearning, built on lightweight LoRA adapters and projection operators over a frozen multimodal backbone.
For challenge \textbf{I)}, we establish \textbf{Bilateral Knowledge Excision (BKE)} to close the Modality Anchor by displacing the image and text branches simultaneously, so the bilinear similarity can no longer reconstruct the removed pairing through the untouched modality. This bilateral principle has no counterpart in single-modality unlearning.
For challenge \textbf{II)}, we develop \textbf{Gradient Subspace Decomposition (GSD)}. Reusing client updates already transmitted during FedAvg, the server extracts per-client subspace bases via SVD and measures their pairwise principal angles; directions nearly orthogonal to the retain subspace are treated as forget-exclusive anchors, while directions shared with retained clients are preserved as retain support. This replaces the binary erase-or-preserve decision with a per-direction operation adapted to the actual client-subspace entanglement.
For challenge \textbf{III)}, we devise \textbf{Projection with Forget Lock (PFL)}. At each unlearning round the server projects the parameter displacement onto the complement of the unique subspace, while a client-side Forget Lock penalizes drift back along the identified unique directions during local SGD, leaving the orthogonal complement unconstrained for retained knowledge recovery. The projection closes the anchor at the unlearning round, and the lock keeps it closed across continued training, with drift bounded inversely by the regularization strength. Our main contributions are:

\begin{itemize}[leftmargin=*]
    \item[\ding{182}] \textbf{\textit{Anchor Principle.}} We identify and formalize an Anchor Principle for federated multimodal contrastive unlearning: forgotten knowledge persists through three principal anchor locations in this regime, the Modality Anchor (bilinear cross-modal coupling), the Unique-Subspace Anchor (principal-angle subspace entanglement), and Temporal Re-anchoring (residual alignment gradient field); empirically, closing any strict subset of the three leaves a measurable reconstruction path. The principle is verified across three practical scenarios: client, class, and sample unlearning.

    \item[\ding{183}] \textbf{\textit{Anchor Closure Framework.}} We introduce \underline{\texttt{EASE}}, realizing the Anchor Principle through three coordinated mechanisms: BKE shuts the Modality Anchor by displacing both branches simultaneously, GSD separates removable anchors from retain support via the Cosine--Sine decomposition of the forget--retain subspace pair, and PFL prevents Temporal Re-anchoring through a server-side projection and a direction-selective client-side penalty. We back the framework with conditional geometric guarantees on unique-subspace excision and retain-direction preservation.

    \item[\ding{184}] \textbf{\textit{Empirical Verification.}} Experiments on Flickr30K, COCO, and TextCaps with three multimodal backbones across three unlearning scenarios show that \underline{\texttt{EASE}} most consistently narrows the gap to the retrain reference on both forget and retain sides among nine federated and centralized baselines, with retrieval, alignment-residual, and shadow/LiRA membership-inference metrics confirming closure at the embedding and inference levels. An anchor-by-anchor ablation confirms that disabling any single closure reopens a distinct failure mode predicted by the principle.
\end{itemize}

\section{Preliminaries}
\label{sec:preliminaries}

\paragraph{Federated multimodal learning.}
We consider a federated system with one central server and $K$ clients. Each client $k$ holds a private multimodal dataset $\mathcal{D}_k$ of $N_k$ image-text pairs $(x_v, x_t)$, where subscripts $v$ and $t$ denote the visual and textual modalities throughout, and data never leaves the client. The model consists of frozen pretrained encoders and lightweight trainable modules, namely Low-Rank Adaptation (LoRA) adapters and projectors. We denote the visual and language trainable parameters by $w_v$ and $w_t$, and write $w = [w_v;\, w_t]$ for the joint parameter vector. A complete notation table covering all symbols used in this paper is provided in Appendix~\ref{sec:app_notation}.

\paragraph{Contrastive alignment objective.}
The visual and language encoders map each input into a shared embedding space, producing $m$-dimensional representations $z_v$ and $z_t$. The system is trained via Federated Averaging (FedAvg) with the symmetric InfoNCE alignment loss:
\begin{equation}
    \mathcal{L}_a(w;\, \mathcal{D}) = -\frac{1}{2N}\sum_{i=1}^{N} \Biggl[\log \frac{\exp(z_{i,v}^\top z_{i,t} / \gamma)}{\sum_{j=1}^{N} \exp(z_{i,v}^\top z_{j,t} / \gamma)} + \log \frac{\exp(z_{i,v}^\top z_{i,t} / \gamma)}{\sum_{j=1}^{N} \exp(z_{j,v}^\top z_{i,t} / \gamma)}\Biggr],
    \label{eq:infonce}
\end{equation}
where $\gamma$ is the temperature and $\mathcal{D} = \bigcup_{k=0}^{K-1}\mathcal{D}_k$; the two terms correspond to the visual-to-text and text-to-visual retrieval directions. A key structural property of this loss is that the gradient with respect to one modality depends on the embeddings of the other: $\nabla_{w_v}\mathcal{L}_a$ depends on all textual embeddings $z_{t,j}$, and vice versa. As we show in Section~\ref{sec:bilateral}, this cross-modal coupling is the mechanism behind the first anchor channel of our principle, the Modality Anchor.

\paragraph{Unlearning problem.}
Given an unlearning request, we partition $\mathcal{D}$ into a forget set $\mathcal{D}_f$ and a retain set $\mathcal{D}_r = \mathcal{D} \setminus \mathcal{D}_f$. We consider three practical scenarios: client unlearning, where one or more clients request complete data withdrawal; class unlearning, where all samples matching a target class must be forgotten across the federation; and sample unlearning, where designated samples across all clients must be removed. Since the caption-retrieval datasets we use do not provide ground-truth semantic class labels, we instantiate class unlearning via KMeans pseudo-class clusters in the pretrained image-text embedding space (Appendix~\ref{sec:app_datasets}). In all three scenarios, the request is issued by a participating client and the server executes the unlearning on the global model. We let $w_n$ denote the model after standard federated training and $\tilde{w}$ the retrain model obtained by rerunning the full training pipeline on $\mathcal{D}_r$ from scratch.

\paragraph{Desiderata.}
The goal is to produce an unlearned model $w^*$ satisfying three criteria: (i) unlearning completeness, the retrieval behavior of $w^*$ on $\mathcal{D}_f$ matches that of $\tilde{w}$; (ii) retention integrity, the performance of $w^*$ on $\mathcal{D}_r$ matches that of $\tilde{w}$; and (iii) efficiency, the cost is substantially lower than full retraining. Existing methods attempt to satisfy these criteria by erasing the forget data's influence from $w$ as a whole, treating the forget gradient subspace as an indivisible block. This blind erasure cannot distinguish forget-exclusive directions, the removable anchors, from directions shared with retain, the retain support; nor does it address the bilinear cross-modal coupling that makes the untouched modality act as a Modality Anchor, as we analyze in Section~\ref{sec:bilateral}. These challenges motivate the entanglement-aware excision framework in Section~\ref{sec:method}.

\textbf{Observation 1: Existing federated unlearning methods exhibit a forget--retain trade-off, indicating that at least one anchor channel remains open in every baseline.}

\begin{wrapfigure}{l}{0.55\linewidth}
    \vspace{-10pt}
    \centering
    \subfloat[Unlearning completeness.\label{fig:obs_forget}]{%
        \includegraphics[width=0.48\linewidth]{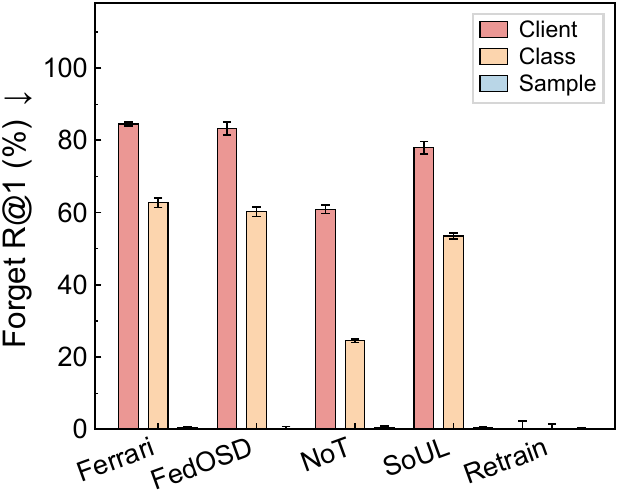}}
    \hfill
    \subfloat[Retention integrity.\label{fig:obs_retain}]{%
        \includegraphics[width=0.48\linewidth]{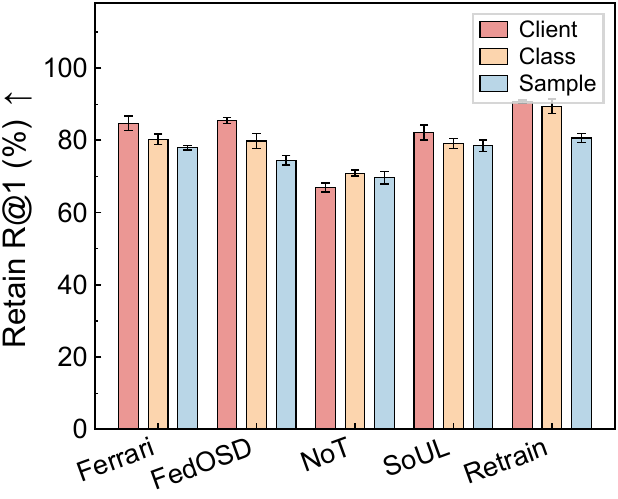}}
    \caption{Forget and retain trade-off of federated unlearning baselines on CLIP-B32 with Flickr30K.}
    \label{fig:obs_baselines}
    \vspace{-15pt}
\end{wrapfigure}

We evaluate four federated unlearning baselines, Ferrari~\cite{Liu2024ICML}, FedOSD~\cite{Pan2025AAAI}, NoT~\cite{Khalil2025CVPR}, and SoUL~\cite{Jia2024EMNLP}, on CLIP-B32 fine-tuned on Flickr30K under three unlearning scenarios. We report Forget R@1, where lower indicates more thorough erasure on the forget set, and Retain R@1, where higher indicates better preservation of retained knowledge. Figure~\ref{fig:obs_retain} shows that most baselines preserve Retain R@1 close to the Retrain reference, suggesting retained knowledge survives the unlearning procedure. Figure~\ref{fig:obs_forget}, however, reveals a sharp failure mode: only sample unlearning is effective, while client and class requests leave the forget set largely intact. NoT is the exception that erases more aggressively, but at the cost of a lower Retain, while the remaining baselines under-erase. These observations expose a tension in current federated multimodal unlearning, motivating a finer treatment of how forget-side anchors and retain support coexist within client gradient updates.

\textbf{Observation 2: Client gradient subspaces in federated multimodal training are partially entangled, and their overlap admits a continuous principal-angle spectrum from retain-aligned support to forget-exclusive anchors.}

\begin{wrapfigure}{r}{0.55\linewidth}
    \vspace{-18pt}
    \centering
    \subfloat[Subspace similarity.\label{fig:obs_heatmap}]{%
        \includegraphics[width=0.48\linewidth]{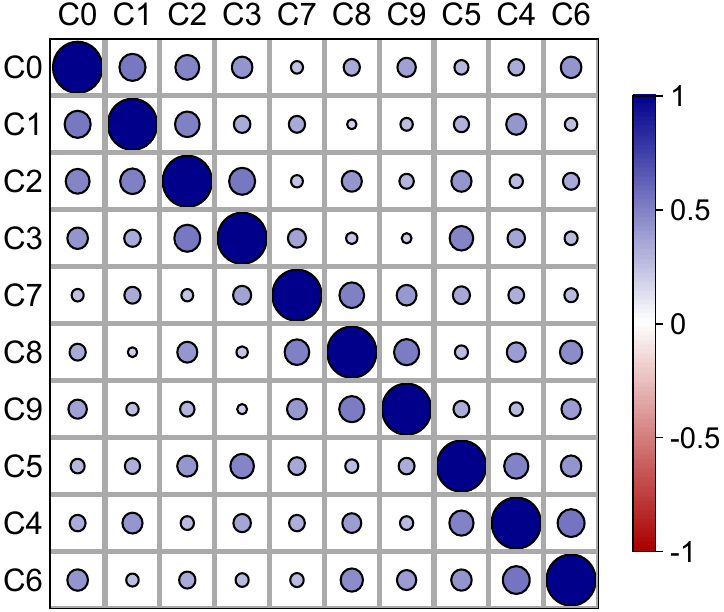}}
    \hfill
    \subfloat[Principal angle spectrum.\label{fig:obs_pairs}]{%
        \includegraphics[width=0.48\linewidth]{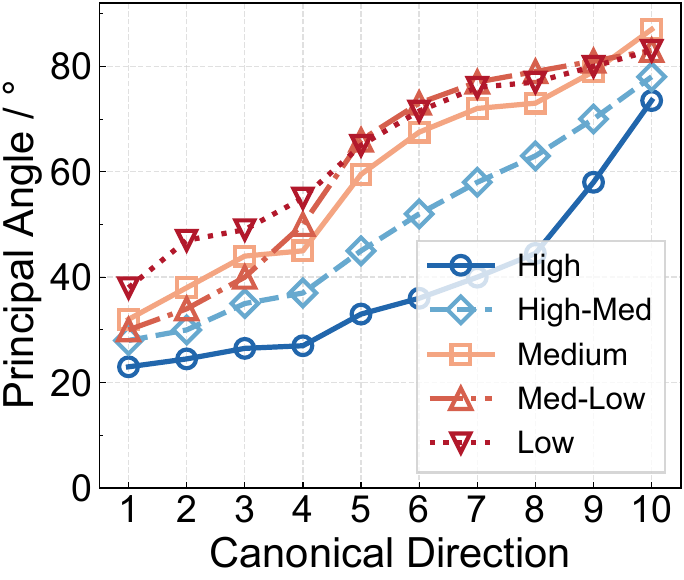}}
    \caption{Gradient subspace entanglement in federated multimodal training on Flickr30K with 10 clients.}
    \label{fig:observation}
    \vspace{-15pt}
\end{wrapfigure}

The trade-off in Observation 1 stems from the entangled structure of client gradient updates: the forget gradient subspace mixes removable anchors (forget-exclusive directions) with retain support (directions shared with retained clients), and blind erasure conflates the two. To validate this, we train a CLIP-B/32 projector on Flickr30K distributed across 10 clients via FedAvg for 30 rounds, collect per-round parameter updates, and extract each client's top-10 update directions via SVD. As shown in Figure~\ref{fig:obs_heatmap}, pairwise subspace similarity varies widely with an off-diagonal mean of 0.49, confirming that client gradient subspaces are neither fully shared nor fully orthogonal. Figure~\ref{fig:obs_pairs} decomposes five representative pairs by their principal angles: in all cases, the angles increase monotonically from small values for shared directions to large values for unique directions. These findings motivate an entanglement-aware decomposition that erases only the forget-exclusive Unique-Subspace Anchors while leaving the retain support intact.

\section{Methodology}
\label{sec:method}

\subsection{Framework Overview}
\label{sec:overview}

\begin{figure*}
    \centering
    \includegraphics[width=1\linewidth]{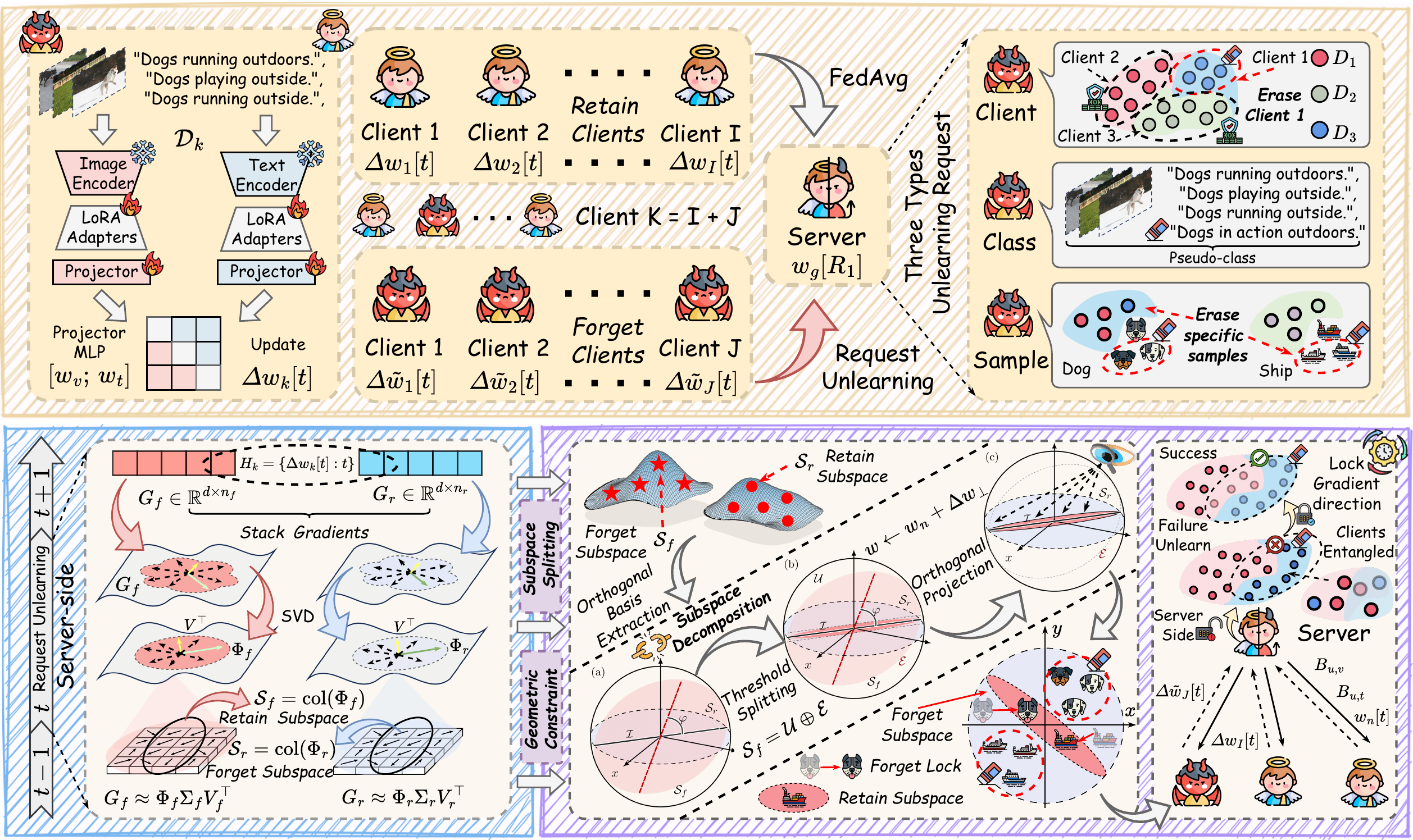}
    \caption{\textbf{Overview of \underline{\texttt{EASE}}.} BKE displaces both visual and language branches simultaneously to close the cross-modal reconstruction channel. GSD separates forget-exclusive directions from retain support via principal-angle decomposition. PFL projects the server displacement off the unique subspace and locks each client's drift along it during continued training.}
    \label{fig:framework}
\end{figure*}

Motivated by Observations 1--2, our framework closes the three anchor channels in turn (Figure~\ref{fig:framework}): \textbf{I)} Bilateral Knowledge Excision (BKE, Section~\ref{sec:bilateral}) shuts the Modality Anchor by displacing both visual and language branches simultaneously; \textbf{II)} Gradient Subspace Decomposition (GSD, Section~\ref{sec:subspace}) separates removable Unique-Subspace Anchors from retain support via Cosine--Sine decomposition of the forget--retain subspace pair; \textbf{III)} Projection with Forget Lock (PFL, Section~\ref{sec:projection}) prevents Temporal Re-anchoring through a server-side projection and a direction-selective client-side penalty. Disabling any single component reopens a distinct failure mode (Table~\ref{tab:ablation}).

\subsection{Closing the Modality Anchor (BKE)}
\label{sec:bilateral}

\paragraph{The Modality Anchor.}
In multimodal models with a bilinear similarity (e.g., InfoNCE), each modality's gradient depends on the other modality's embedding. Excising parameter directions on one branch alone leaves the unmodified branch as a memory of the original alignment, and the alignment gradient exploits this memory as an anchor that reconstructs the erased pairing. We term this residual encoding the Modality Anchor of the forget set. BKE closes it by displacing both branches simultaneously.

\paragraph{Anchor mechanism.}
The cross-modal alignment depends on both the visual and language parameters through the bilinear similarity in Eq.~\eqref{eq:infonce}. To see why unilateral excision fails, consider a forget-set pair $(x_{v,i}, x_{t,i})$. The InfoNCE gradient with respect to $w_v$ decomposes as:
\begin{equation}
    \nabla_{w_v} \mathcal{L}_a = \frac{1}{\gamma} J_v(w_v)^\top \!\left(\sum_{j=1}^{N} \frac{\exp\!\bigl(z_{v,i}^\top z_{t,j}/\gamma\bigr)}{\sum_{l}\exp\!\bigl(z_{v,i}^\top z_{t,l}/\gamma\bigr)}\, z_{t,j}(w_t) \;-\; z_{t,i}(w_t)\right),
    \label{eq:infonce_grad}
\end{equation}
where $J_v$ denotes the Jacobian $\partial z_v / \partial w_v$. The gradient direction with respect to $w_v$ is determined by the textual embeddings $z_{t,j}$, and vice versa. This cross-modal dependence is the physical seat of the Modality Anchor: if we remove forget-unique directions from $w_v$ only, leaving $w_t$ at the pre-unlearning value $w_{n,t}$, the visual representation is displaced, yet the textual embedding $z_{t,i}(w_{n,t})$ still encodes the original pairing. Since this intact embedding dominates the softmax weights in the contrastive loss, the gradient evaluated at the excised point reduces to:
\begin{equation}
    \nabla_{w_v} \mathcal{L}_a\big|_{w_v^*,\, w_{n,t}} \;\propto\; J_v(w_v^*)^\top\, z_{t,i}(w_{n,t}).
    \label{eq:anchor_grad}
\end{equation}
The unchanged textual embedding $z_{t,i}(w_{n,t})$ is the Modality Anchor: it encodes the semantic content of the erased pair and continuously drives the alignment gradient to re-align $w_v$ with the original pairing during retraining.

\paragraph{Alignment residual.}
We quantify the open Modality Anchor by the normalized gap between the unlearned model $w^*$ and the retrain reference $\tilde{w}$:
\begin{equation}
    \rho = \frac{\bigl|z_v(w_v^*)^\top z_t(w_t^*) \;-\; z_v(\tilde{w}_v)^\top z_t(\tilde{w}_t)\bigr|}{\bigl|z_v(w_{n,v})^\top z_t(w_{n,t}) \;-\; z_v(\tilde{w}_v)^\top z_t(\tilde{w}_t)\bigr| + \epsilon},
    \label{eq:residual}
\end{equation}
with $w_n$ the pre-unlearning model and $\epsilon$ a stability constant. $\rho \approx 0$ indicates the anchor is closed; $\rho \approx 1$ that it remains fully open.

\paragraph{Comparison of excision strategies.}
\textbf{(i)}~Visual-only excision on $w_v$ leaves $z_t(w_{n,t})$ as the Modality Anchor that re-pulls $z_v$ toward the forgotten alignment, yielding $\rho > 0$; \textbf{(ii)}~by symmetry, language-only excision also yields $\rho > 0$; \textbf{(iii)}~bilateral excision (ours) displaces both embeddings, and with GSD and PFL holding the unique directions closed across rounds, $\rho \approx 0$. Appendix~\ref{app:proof_bilateral} formalizes this contrast under Assumptions A1--A4 (Appendix~\ref{app:assumptions}).

\subsection{Closing the Unique-Subspace Anchor (GSD)}
\label{sec:subspace}

\paragraph{Removable anchors versus retain support.}
BKE tells us to operate on both modalities, but how much should we erase from each? In federated non-IID training, the forget client's parameter updates share directions with retained clients along a continuous principal-angle spectrum. We need to decompose the forget subspace into removable Unique-Subspace Anchors (forget-exclusive directions, safe to erase) and retain support (directions shared with retain, must preserve).

\paragraph{Gradient history collection.}
Over the initial $R_1$ FedAvg rounds, each sampled client's per-round update $\Delta w_k[t]$ (its post-local-SGD parameter delta relative to the broadcast global model) is already transmitted to the server, so storing it as a per-client history $H_k$ adds no extra communication for client unlearning. Sample and class unlearning (Appendix~\ref{sec:app_datasets}) require one additional request-specific round for clients holding matching data to upload forget-specific gradients. The global model $w_n$ after $R_1$ rounds serves as the unlearning reference.

\paragraph{Forget and retain subspace extraction.}
We construct a forget gradient matrix $G_f \in \mathbb{R}^{d \times n_f}$ by column-stacking the updates associated with the forget set $\mathcal{D}_f$. The column space of $G_f$ captures the principal directions along which the forget data has shaped the model parameters. The construction of $G_f$ depends on the unlearning scenario. For client unlearning, $G_f$ collects the full history of updates from the requesting client:
\begin{equation}
    G_f = \bigl[\Delta w_{k^*}[t_1],\; \dots,\; \Delta w_{k^*}[t_{|H_{k^*}|}]\bigr],
    \label{eq:Gf_client}
\end{equation}
where $t_1, \dots, t_{|H_{k^*}|}$ are the rounds in which $k^*$ participated. For sample and class unlearning, the forget targets are generally distributed across multiple clients. The server broadcasts a target descriptor (a set of sample identifiers for sample unlearning, or a textual descriptor for class unlearning), and every client that holds matching data performs one epoch of training on its local forget subset and uploads the resulting gradient; $G_f$ is column-stacked from these per-client contributions. A retain gradient matrix $G_r$ is constructed analogously from each client's remaining non-forget data (see Appendix~\ref{sec:app_impl} for details).

We apply Singular Value Decomposition (SVD) to $G_f$ and retain the top-$p$ left singular vectors that capture at least a fraction $\tau_e$ of the total energy:
\begin{equation}
    p = \min\!\Bigl\{j \in \{1, \dots, \mathrm{rank}(G_f)\} \;:\; \frac{\sum_{i=1}^{j} \sigma_i^2}{\|G_f\|_F^2} \;\geq\; \tau_e \Bigr\},
    \label{eq:energy}
\end{equation}
This energy-based criterion adapts $p$ to the spectral decay of $G_f$. The resulting forget subspace basis is $\Phi_f \in \mathbb{R}^{d \times p}$, whose columns are the top-$p$ left singular vectors of $G_f$. The same procedure applied to $G_r$ yields the retain subspace basis $\Phi_r \in \mathbb{R}^{d \times q}$. We denote the column spaces of $\Phi_f$ and $\Phi_r$ by $\mathcal{S}_f$ and $\mathcal{S}_r$, respectively.

\paragraph{Entanglement spectrum.}
Not all directions inside $\mathcal{S}_f$ behave the same way: a direction lying entirely within $\mathcal{S}_r$ is fully retain-aligned and serves as retain support, since erasing it destroys retained knowledge; a direction orthogonal to $\mathcal{S}_r$ is forget-exclusive and constitutes a removable anchor. To quantify this overlap for every forget-side canonical direction, we introduce the entanglement spectrum of $\mathcal{S}_f$ with respect to $\mathcal{S}_r$. We compute the cross-correlation matrix $M = \Phi_f^\top \Phi_r$ and take its full SVD:
\begin{equation}
    M = U_f\,\Sigma\,V_f^\top,\quad U_f \in \mathbb{R}^{p \times p},\quad V_f \in \mathbb{R}^{q \times q},\quad r=\min(p,q),
    \label{eq:entanglement_spectrum}
\end{equation}
where the rectangular diagonal of $\Sigma$ has entries $\cos\varphi_1,\ldots,\cos\varphi_r$, the cosines of the principal angles. If $p>q$, the remaining $p-q$ columns of $U_f$ span directions in $\mathcal{S}_f$ that are orthogonal to $\mathcal{S}_r$; we assign them entanglement coefficient $0$. Thus for every $i \in \{1,\ldots,p\}$ we define $\kappa_i=\cos\varphi_i$ for $i\le r$ and $\kappa_i=0$ for $i>r$. The rotated basis vectors $c_i = (\Phi_f U_f)_{\cdot,i}$ are the canonical directions of $\mathcal{S}_f$ relative to $\mathcal{S}_r$: a coefficient near 1 indicates retain-aligned (retain support); near 0 indicates forget-exclusive (removable anchor). We partition using a threshold $\delta \in [0,1]$ into unique (anchor) and entangled (support) basis matrices, where a larger $\delta$ classifies more directions as removable anchors; we analyze the sensitivity to this threshold in Section~\ref{sec:experiments}:
\begin{equation}
    B_u = \bigl[c_i\bigr]_{\kappa_i \leq \delta} \in \mathbb{R}^{d \times |\mathcal{U}|}, \quad
    B_e = \bigl[c_i\bigr]_{\kappa_i > \delta} \in \mathbb{R}^{d \times |\mathcal{E}|}.
    \label{eq:Bu_Be}
\end{equation}

The unique subspace $\mathcal{U} = \mathrm{col}(B_u)$ spans forget-exclusive removable anchors (collectively Unique-Subspace Anchor) and the entangled subspace $\mathcal{E} = \mathrm{col}(B_e)$ spans retain support, with $\mathcal{S}_f = \mathcal{U} \oplus \mathcal{E}$. The decomposition is intrinsic to the subspace pair $(\mathcal{S}_f, \mathcal{S}_r)$, which distinguishes anchor identification from single-subspace projection methods that erase $\mathcal{S}_f$ as a block. Applying it independently to each LoRA layer and projector within both modalities yields per-group bases $B_{u,\mu}^{l}$ and $B_{e,\mu}^{l}$; the theorems in Section~\ref{sec:projection} apply per block, and global guarantees follow by summing across blocks.

\subsection{Preventing Temporal Re-anchoring (PFL)}
\label{sec:projection}

\paragraph{Scope of theoretical guarantees.}
Our guarantees are geometric, not a certified data-deletion proof: given a faithful $\mathcal{U}_\mu$, projection removes its displacement (Theorem~\ref{thm:completeness}), collateral energy on $\mathcal{S}_r$ is bounded by $\delta^2\|\Delta_r\|_2^2$ (Theorem~\ref{thm:retention}), and the Forget Lock contracts unique-subspace drift to $G_u/(2\alpha)$ (Theorem~\ref{thm:durability}). Operational claims about membership leakage rely on the empirical proxies $\rho$, raincloud distributions, and LiRA TPR, evaluated in Section~\ref{sec:experiments} and Appendices~\ref{sec:app_raincloud},~\ref{sec:app_lira}.

\paragraph{Temporal Re-anchoring.}
A one-time projection excises the Unique-Subspace Anchor at $t{=}0$, but retain-data SGD still has nonzero components $g_u(t) := \Pi_{u,\mu}\nabla_{w_\mu}\mathcal{L}_a(w(t))$ along $\mathcal{U}_\mu$, since the unique subspace is only partially orthogonal to the retain gradient field. Letting $d(t) := \Pi_{u,\mu}(w_\mu(t) - w_{n,\mu})$, a single SGD step under the combined objective gives the recurrence
\begin{equation}
    d(t{+}1) \;=\; (1 - 2\lr\alpha)\,d(t) \;-\; \lr\, g_u(t),
    \label{eq:reanchor_recurrence}
\end{equation}
where $g_u(t)$ is the residual alignment gradient surviving along $\mathcal{U}_\mu$ after BKE has suppressed the Modality Anchor's leading component. Without persistent constraint ($\alpha = 0$), $g_u(t)$ accumulates across rounds and rewrites the cut anchors, the third anchor channel.

\paragraph{Direction-selective lock.}
PFL constrains drift along $\mathcal{U}_\mu$ while leaving $\mathcal{U}_\mu^\perp$ unconstrained, an asymmetry that neither per-round server projection alone nor axis-aligned importance regularizers can enforce (Appendix~\ref{app:lock_alternatives}). The construction pairs a server-side projection that pulls each global modality $\mu \in \{v,t\}$ back to $\mathcal{M}_\mu$ with a client-side quadratic penalty supported on $\mathcal{U}_\mu$:
\begin{align}
    w_{g,\mu} &\;\leftarrow\; w_{g,\mu} - \Pi_{u,\mu}(w_{g,\mu} - w_{n,\mu}),\quad \Pi_{u,\mu} := B_{u,\mu} B_{u,\mu}^\top, \label{eq:proj_step} \\
    \mathcal{L}_f(w) &= \sum_{\mu \in \{v,t\}} \bigl\|B_{u,\mu}^\top(w_\mu - w_{n,\mu})\bigr\|_2^2. \label{eq:forget_lock}
\end{align}
Each client trains on the retained data $\mathcal{D}_k' = \mathcal{D}_k \setminus \mathcal{D}_f$ with the combined objective
\begin{equation}
    \mathcal{L}(w;\, \mathcal{D}_k') = \mathcal{L}_a(w;\, \mathcal{D}_k') + \alpha \cdot \mathcal{L}_f(w),
    \label{eq:local_obj}
\end{equation}
where $\alpha > 0$ is the lock strength. The recurrence Eq.~\eqref{eq:reanchor_recurrence} then becomes a contraction with factor $1 - 2\lr\alpha < 1$ (Theorem~\ref{thm:durability}), bounding unique-subspace drift by $G_u/(2\alpha)$ regardless of round count, while $\mathcal{U}_\mu^\perp$ (including the entangled retain support from GSD) stays unconstrained so retained knowledge recovers. The bases $B_{u,v}$, $B_{u,t}$ and reference $w_n$ are broadcast with the global model, so $\mathcal{L}_f$ is local with no extra communication. After $R_2$ excision rounds the server runs $R_3$ stabilization rounds without projection; the lock remains active.

The two thresholds therefore control complementary failure modes: smaller $\delta$ leaves anchors partially open while larger $\delta$ misclassifies retain support, and larger $\alpha$ tightens re-anchoring durability while leaving $\mathcal{U}_\mu^\perp$ free for recovery (Appendices~\ref{app:proof_completeness}--\ref{app:proof_durability}; $\alpha{=}0$ ablation in Section~\ref{sec:experiments}).

\section{Experiments}
\label{sec:experiments}

\subsection{Experimental Setup}
\label{sec:setup}

\textbf{Datasets.} We evaluate on three benchmark image-text datasets: Flickr30K~\cite{Young2014TACL}, MS COCO~\cite{Lin2014ECCV}, and TextCaps~\cite{Sidorov2020ECCV}, covering standard caption retrieval and scene-text retrieval. We use three backbones: CLIP-B/32~\cite{Radford2021ICML}, CLIP-L/14, and SigLIP~\cite{Zhai2023ICCV}. Dataset splits are detailed in Appendix~\ref{sec:app_datasets}; implementation details and hyperparameters are in Appendix~\ref{sec:app_impl}. Since these datasets lack ground-truth class labels, class unlearning uses pseudo-class clusters in the pretrained backbone embedding space.

\textbf{Counterparts.} We compare our method against nine baselines from two categories. Federated unlearning methods: (1) \textbf{FedEraser}~\cite{Liu2021IWQOS}, (2) \textbf{FedRecover}~\cite{Cao2023SP}, (3) \textbf{Ferrari}~\cite{Liu2024ICML}, (4) \textbf{FedOSD}~\cite{Pan2025AAAI}, (5) \textbf{NoT}~\cite{Khalil2025CVPR}, (6) \textbf{SoUL}~\cite{Jia2024EMNLP}, (7) \textbf{FFMU}~\cite{Che2023ICML}, (8) \textbf{FUSED}~\cite{Zhong2025CVPR}; centralized unlearning adapted to the federated setting: (9) \textbf{GradAscent}~\cite{Thudi2022USENIX}; and the retrain reference $\tilde{w}$. Detailed descriptions of all counterparts are provided in Appendix~\ref{sec:app_baselines}.

\textbf{Evaluation metrics.} We evaluate four metrics: Recall@$k$~\cite{Karpathy2015CVPR} with $k \in \{1, 5, 10\}$ for cross-modal retrieval, a shadow-model membership inference attack (MIA)~\cite{Shokri2017SP} score for aggregate privacy verification, the likelihood-ratio attack true-positive rate at $1\%$ false-positive rate ($\mathrm{LiRA}~\mathrm{TPR}@\mathrm{FPR}{=}1\%$)~\cite{Carlini2022SP} for low-FPR privacy auditing, and communication cost~\cite{Kairouz2021FNT} in megabytes. Full metric definitions are in Appendix~\ref{sec:app_metrics}. With this setup, three research questions structure the remainder of the section: \textbf{Q1} (Superiority), \textbf{Q2} (Effectiveness), and \textbf{Q3} (Sensitivity).

\subsection{Superiority (Q1)}
\label{sec:q1}

For \textbf{Q1}, we test whether $w^*$ matches the retrain reference $\tilde{w}$ on the forget set $\bar{\mathcal{D}}_f$ and the retain set $\bar{\mathcal{D}}_r$. Table~\ref{tab:q1_flickr30k_clipb32} gives the comparison on Flickr30K with CLIP-B/32; eight dataset--backbone combinations are in Appendix~\ref{sec:app_baselines_full}. Figure~\ref{fig:circular} summarizes per-metric performance across scenarios.

\begin{table*}[!ht]
\centering
\setlength{\tabcolsep}{3pt}
\setlength{\aboverulesep}{0.3pt}
\setlength{\belowrulesep}{0.3pt}
\renewcommand{\arraystretch}{0.95}
\caption{Main comparison on Flickr30K with CLIP-B/32 across three unlearning scenarios. The top sub-table reports forget-side metrics ($F\text{-}R@k$) and membership inference (MIA); the bottom sub-table reports retain-side metrics ($R\text{-}R@k$) and communication cost ($Comm.$, MB). For $F\text{-}R@k$ and $R\text{-}R@k$, parenthetical values give the absolute gap to the retrain reference with \up{green} closer to ideal and \down{red} farther; for MIA and $Comm.$, \gap{val} is the absolute difference to retrain. Best in \best{red}, runner-up in \second{blue} among baselines and our method; \textit{Retrain} does not participate in ranking.}
\label{tab:q1_flickr30k_clipb32}
\resizebox{\linewidth}{!}{
\begin{tabular}{l|cccc|cccc|cccc}
\Xhline{1.2pt}
\rowcolor{CadetBlue!20}
\textbf{Method} & \multicolumn{4}{c|}{\textbf{Client Unlearning}} & \multicolumn{4}{c|}{\textbf{Class Unlearning}} & \multicolumn{4}{c}{\textbf{Sample Unlearning}} \\
\cmidrule(lr){2-5} \cmidrule(lr){6-9} \cmidrule(lr){10-13}
\rowcolor{CadetBlue!20}
 & $F\text{-}R@1\;\downarrow$ & $F\text{-}R@5\;\downarrow$ & $F\text{-}R@10\;\downarrow$ & $MIA$ & $F\text{-}R@1\;\downarrow$ & $F\text{-}R@5\;\downarrow$ & $F\text{-}R@10\;\downarrow$ & $MIA$ & $F\text{-}R@1\;\downarrow$ & $F\text{-}R@5\;\downarrow$ & $F\text{-}R@10\;\downarrow$ & $MIA$ \\
\Xhline{1.2pt}
\rowcolor{gray!10}
FedEraser & 69.8 \down{69.7} & 90.9 \down{89.4} & 95.2 \down{92.9} & 48.9 \gap{32.2} & 38.8 \down{38.6} & 68.2 \down{67.5} & 79.2 \down{77.8} & \best{64.0} \gap{39.8} & 0.3 \down{0.2} & 2.4 \down{1.3} & 5.8 \down{3.4} & 16.7 \\
FedRecover & 57.9 \down{57.8} & 81.2 \down{79.7} & 87.7 \down{85.4} & 47.2 \gap{30.5} & 41.6 \down{41.4} & 72.0 \down{71.3} & 83.5 \down{82.1} & 41.7 \gap{17.5} & 0.5 \down{0.4} & 2.8 \down{1.7} & 6.2 \down{3.8} & \second{16.9} \gap{0.2} \\
\rowcolor{gray!10}
Ferrari & 84.5 \down{84.4} & 98.7 \down{97.2} & 99.6 \down{97.3} & 52.7 \gap{36.0} & 62.7 \down{62.5} & 92.7 \down{92.0} & 97.1 \down{95.7} & \second{60.3} \gap{36.1} & 0.4 \down{0.3} & 2.5 \down{1.4} & 5.9 \down{3.5} & 16.7 \\
FedOSD & 83.3 \down{83.2} & 97.7 \down{96.2} & 99.4 \down{97.1} & \second{54.0} \gap{37.3} & 60.2 \down{60.0} & 89.6 \down{88.9} & 95.3 \down{93.9} & 58.7 \gap{34.5} & 0.3 \down{0.2} & 2.4 \down{1.3} & 5.7 \down{3.3} & 16.7 \\
\rowcolor{gray!10}
NoT & 60.9 \down{60.8} & 88.1 \down{86.6} & 93.6 \down{91.3} & 50.1 \gap{33.4} & 24.5 \down{24.3} & 60.4 \down{59.7} & 74.1 \down{72.7} & 48.9 \gap{24.7} & 0.6 \down{0.5} & 3.0 \down{1.9} & 6.5 \down{4.1} & 16.9 \gap{0.2} \\
SoUL & 77.9 \down{77.8} & 97.1 \down{95.6} & 98.9 \down{96.6} & 50.4 \gap{33.7} & 53.5 \down{53.3} & 84.1 \down{83.4} & 91.4 \down{90.0} & 51.5 \gap{27.3} & 0.5 \down{0.4} & 2.9 \down{1.8} & 6.3 \down{3.9} & 16.7 \\
\rowcolor{gray!10}
FFMU & 39.9 \down{39.8} & 69.8 \down{68.3} & 80.9 \down{78.6} & 52.1 \gap{35.4} & \second{14.7} \down{14.5} & \second{35.9} \down{35.2} & \second{48.4} \down{47.0} & 43.9 \gap{19.7} & 0.3 \down{0.2} & 2.5 \down{1.4} & 5.8 \down{3.4} & \best{17.3} \gap{0.6} \\
FUSED & 82.1 \down{82.0} & 97.5 \down{96.0} & 99.4 \down{97.1} & 52.1 \gap{35.4} & 56.5 \down{56.3} & 87.3 \down{86.6} & 93.9 \down{92.5} & 58.9 \gap{34.7} & 0.4 \down{0.3} & 2.7 \down{1.6} & 6.0 \down{3.6} & 16.7 \\
\rowcolor{gray!10}
GradAscent & \second{1.8} \down{1.7} & \second{7.2} \down{5.7} & \second{14.1} \down{11.8} & \best{54.5} \gap{37.8} & 15.9 \down{15.7} & 39.4 \down{38.7} & 53.7 \down{52.3} & 46.5 \gap{22.3} & \second{0.3} \down{0.2} & \second{2.4} \down{1.3} & \second{5.7} \down{3.3} & 16.8 \gap{0.1} \\
\rowcolor{gray!10}
\textbf{EASE} & \best{0.3} \down{0.2} & \best{1.0} \up{0.5} & \best{2.2} \up{0.1} & 16.9 \gap{0.2} & \best{0.3} \down{0.1} & \best{1.5} \down{0.8} & \best{2.6} \down{1.2} & 25.4 \gap{1.2} & \best{0.1} & \best{2.2} \down{1.1} & \best{5.5} \down{3.1} & 16.7 \\
\textit{Retrain} & 0.1 & 1.5 & 2.3 & 16.7 & 0.2 & 0.7 & 1.4 & 24.2 & 0.1 & 1.1 & 2.4 & 16.7 \\
\Xhline{1.2pt}
\end{tabular}}
\vspace{-2pt}
\resizebox{\linewidth}{!}{
\begin{tabular}{l|cccc|cccc|cccc}
\Xhline{1.2pt}
\rowcolor{CadetBlue!20}
\textbf{Method} & $R\text{-}R@1\;\uparrow$ & $R\text{-}R@5\;\uparrow$ & $R\text{-}R@10\;\uparrow$ & $Comm.\;\downarrow$ & $R\text{-}R@1\;\uparrow$ & $R\text{-}R@5\;\uparrow$ & $R\text{-}R@10\;\uparrow$ & $Comm.\;\downarrow$ & $R\text{-}R@1\;\uparrow$ & $R\text{-}R@5\;\uparrow$ & $R\text{-}R@10\;\uparrow$ & $Comm.\;\downarrow$ \\
\Xhline{1.2pt}
\rowcolor{gray!10}
FedEraser & 85.1 \down{5.6} & 98.7 \down{1.2} & 99.5 \down{0.5} & 262.5 & 78.8 \down{10.6} & 96.0 \down{3.9} & 98.9 \down{1.1} & 196.9 & 78.3 \down{2.3} & 96.9 \down{0.8} & 99.0 \down{0.4} & 196.9 \\
FedRecover & 73.6 \down{17.1} & 94.9 \down{5.0} & 97.2 \down{2.8} & \best{47.2} \gap{215.3} & 76.8 \down{12.6} & 93.0 \down{6.9} & 94.6 \down{5.4} & \best{35.4} \gap{161.5} & 58.8 \down{21.8} & 81.6 \down{16.1} & 87.4 \down{12.0} & \best{35.4} \gap{161.5} \\
\rowcolor{gray!10}
Ferrari & 84.7 \down{6.0} & 98.9 \down{1.0} & 99.5 \down{0.5} & 72.2 \gap{190.3} & 80.3 \down{9.1} & 96.4 \down{3.5} & 99.1 \down{0.9} & 59.1 \gap{137.8} & 77.9 \down{2.7} & 96.8 \down{0.9} & 99.0 \down{0.4} & 59.1 \gap{137.8} \\
FedOSD & 85.5 \down{5.2} & 99.0 \down{0.9} & 99.6 \down{0.4} & 105.0 \gap{157.5} & 79.8 \down{9.6} & 96.0 \down{3.9} & 98.9 \down{1.1} & 78.8 \gap{118.1} & 74.5 \down{6.1} & 95.9 \down{1.8} & 98.4 \down{1.0} & 78.8 \gap{118.1} \\
\rowcolor{gray!10}
NoT & 66.9 \down{23.8} & 92.3 \down{7.6} & 96.8 \down{3.2} & 52.5 \gap{210.0} & 70.9 \down{18.5} & 93.9 \down{6.0} & 98.1 \down{1.9} & \second{39.4} \gap{157.5} & 69.6 \down{11.0} & 93.9 \down{3.8} & 97.8 \down{1.6} & 39.4 \gap{157.5} \\
SoUL & 82.1 \down{8.6} & 98.3 \down{1.6} & 99.1 \down{0.9} & 52.5 \gap{210.0} & 79.1 \down{10.3} & 95.8 \down{4.1} & 98.8 \down{1.2} & 39.4 \gap{157.5} & 78.5 \down{2.1} & 97.3 \down{0.4} & 99.0 \down{0.4} & 39.4 \gap{157.5} \\
\rowcolor{gray!10}
FFMU & 44.5 \down{46.2} & 75.0 \down{24.9} & 83.5 \down{16.5} & 105.0 \gap{157.5} & 48.9 \down{40.5} & 80.7 \down{19.2} & 89.1 \down{10.9} & 78.8 \gap{118.1} & 40.5 \down{40.1} & 74.1 \down{23.6} & 82.5 \down{16.9} & 78.8 \gap{118.1} \\
FUSED & \second{86.2} \down{4.5} & \second{99.1} \down{0.8} & \second{99.7} \down{0.3} & 105.0 \gap{157.5} & \second{80.5} \down{8.9} & \second{96.4} \down{3.5} & \second{99.2} \down{0.8} & 78.8 \gap{118.1} & \second{79.0} \down{1.6} & \second{97.5} \down{0.2} & \second{99.1} \down{0.3} & 78.8 \gap{118.1} \\
\rowcolor{gray!10}
GradAscent & 1.6 \down{89.1} & 7.8 \down{92.1} & 13.1 \down{86.9} & 72.2 \gap{190.3} & 36.7 \down{52.7} & 66.2 \down{33.7} & 79.6 \down{20.4} & 59.1 \gap{137.8} & 76.6 \down{4.0} & 96.0 \down{1.7} & 98.7 \down{0.7} & 59.1 \gap{137.8} \\
\rowcolor{gray!10}
\textbf{EASE} & \best{86.5} \down{4.2} & \best{99.2} \down{0.7} & \best{99.8} \down{0.2} & \second{52.5} \gap{210.0} & \best{80.8} \down{8.6} & \best{96.6} \down{3.3} & \best{99.3} \down{0.7} & 47.2 \gap{149.7} & \best{79.4} \down{1.2} & \best{97.7} & \best{99.2} \down{0.2} & \second{39.4} \gap{157.5} \\
\textit{Retrain} & 90.7 & 99.9 & 100.0 & 262.5 & 89.4 & 99.9 & 100.0 & 196.9 & 80.6 & 97.7 & 99.4 & 196.9 \\
\Xhline{1.2pt}
\end{tabular}}
\end{table*}

\begin{figure*}[t]
    \centering
    \subfloat[Client Unlearning]{\includegraphics[width=0.32\linewidth]{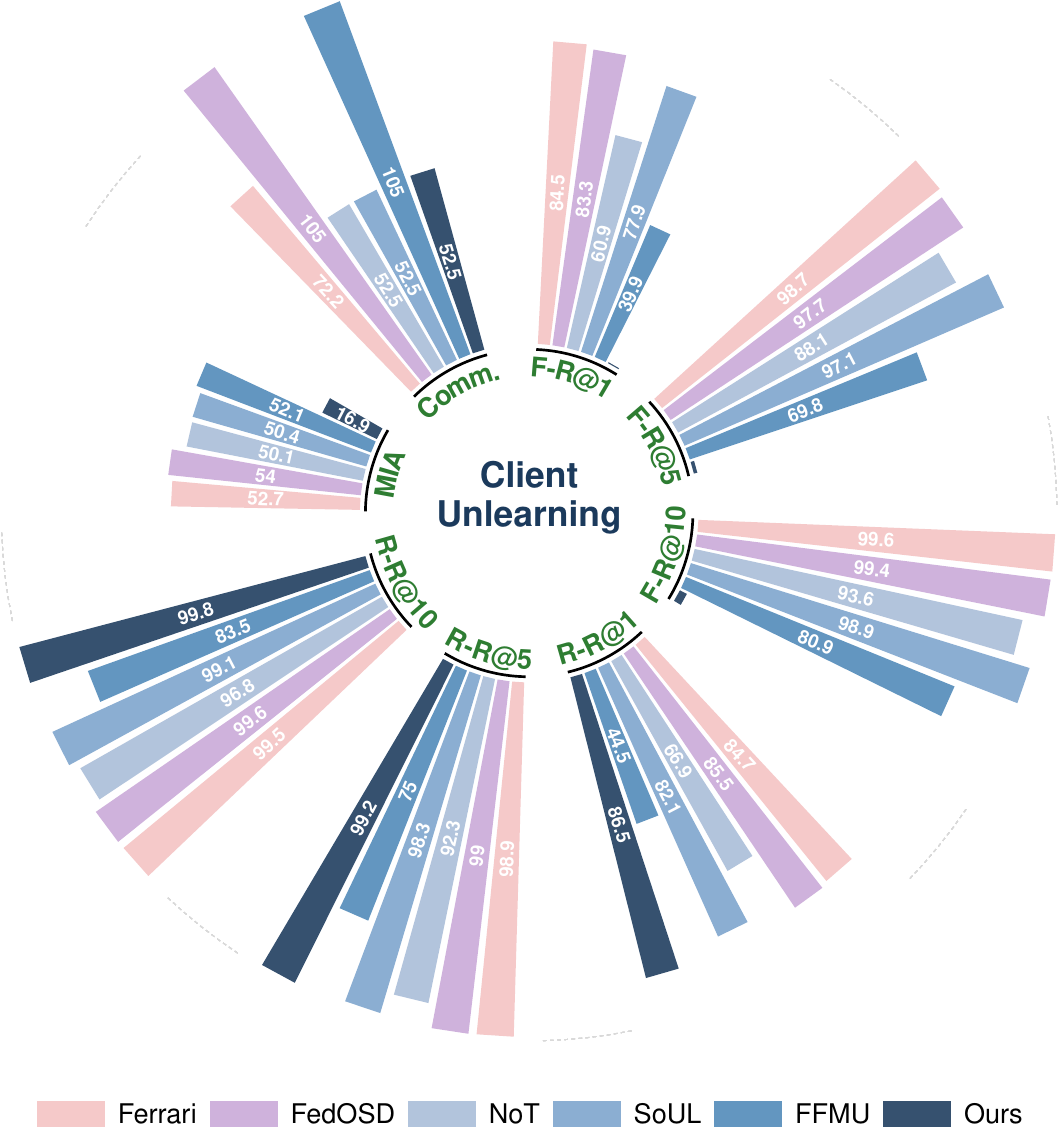}\label{fig:circ_client}}
    \hfill
    \subfloat[Class Unlearning]{\includegraphics[width=0.32\linewidth]{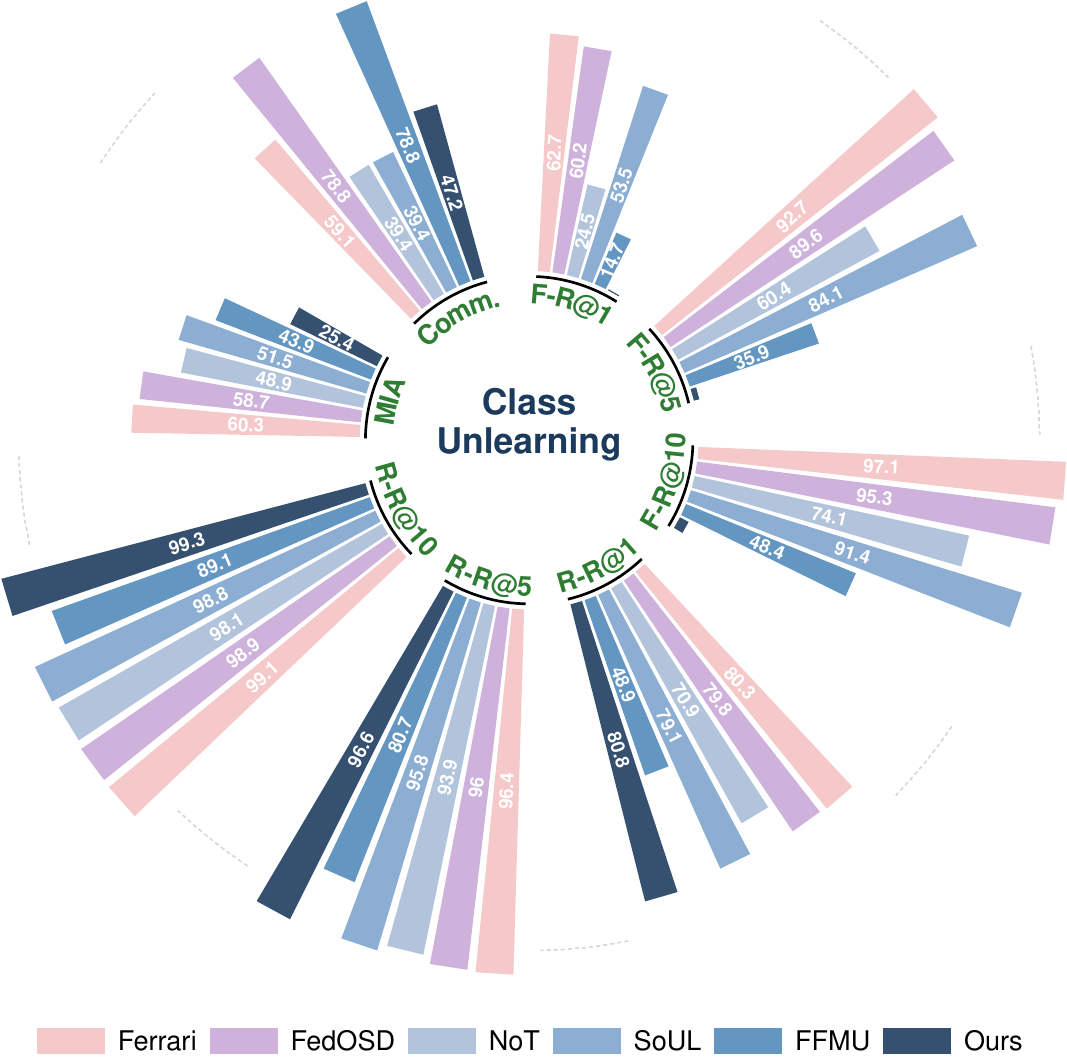}\label{fig:circ_class}}
    \hfill
    \subfloat[Sample Unlearning]{\includegraphics[width=0.32\linewidth]{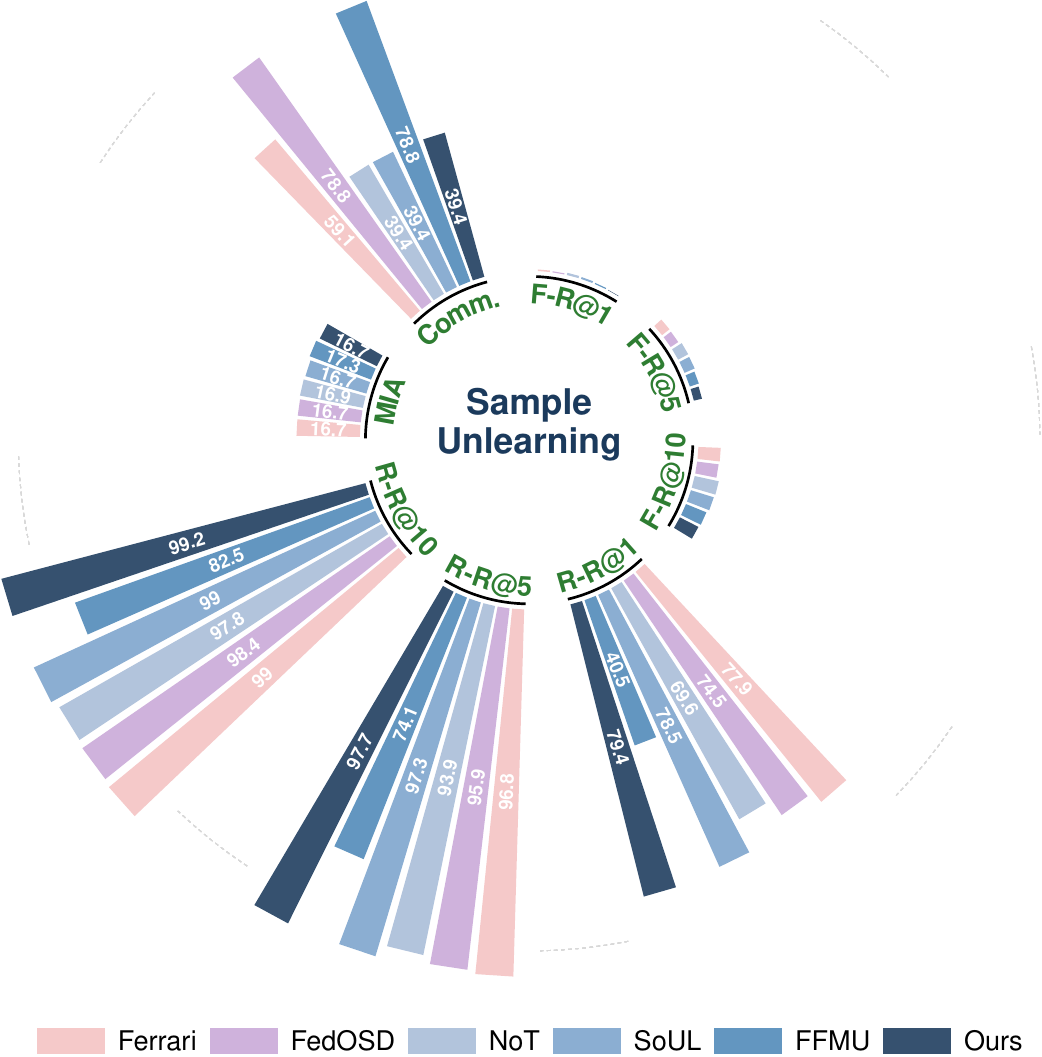}\label{fig:circ_sample}}
    \caption{Per-metric comparison on CLIP-B/32 / Flickr30K under three unlearning scenarios. Each sector represents one metric; bars extend outward proportional to the metric value. Our method achieves the lowest Forget R@k while maintaining competitive Retain R@k across all scenarios.}
    \label{fig:circular}
\end{figure*}

Across nine dataset/backbone combinations in Table~\ref{tab:q1_flickr30k_clipb32} and Appendix~\ref{sec:app_baselines_full}, EASE most consistently narrows both $F\text{-}R@k$ and $R\text{-}R@k$ to retrain. GradAscent erases thoroughly but collapses retain by conflating anchors with support; Ferrari and FUSED preserve retain but under-erase by leaving the Modality Anchor open. Raincloud plots in Appendix~\ref{sec:app_raincloud} confirm this at the embedding level.

\subsection{Effectiveness (Q2)}
\label{sec:q2}

For \textbf{Q2}, each row of Table~\ref{tab:ablation} removes exactly one component from the full method on CLIP-B/32 with Flickr30K across all three unlearning scenarios. BKE and GSD are load-bearing: removing BKE reopens the Modality Anchor (Eq.~\ref{eq:anchor_grad}, with text-side excision worse than visual-side due to a denser forget fingerprint in the language branch), and removing GSD conflates Unique-Subspace Anchors with retain support so over-erasure hits retain (Theorem~\ref{thm:retention}). Removing the Forget Lock admits a milder Temporal Re-anchoring drift since BKE and GSD still close the first two channels at $t{=}0$; only the full method closes the gap on both axes.

\begin{table*}[t]
\centering
\setlength{\tabcolsep}{3pt}
\setlength{\aboverulesep}{1pt}
\setlength{\belowrulesep}{1pt}
\renewcommand{\arraystretch}{1.0}
\caption{Ablation on CLIP-B/32 / Flickr30K across three unlearning scenarios (mean over three seeds). Top: forget-side metrics and MIA; bottom: retain-side metrics and alignment residual $\rho$ (Section~\ref{sec:bilateral}; $\rho{=}0$ matches retrain). For $F\text{-}R@k$, $R\text{-}R@k$, $\rho$, parentheticals give absolute gap to retrain with \up{green} closer and \down{red} farther; for MIA, \gap{val} is the gap to retrain. Best in \best{red}, runner-up in \second{blue} among variants. $^\dagger$visual-only, $^\ddagger$text-only.}
\label{tab:ablation}
\resizebox{\linewidth}{!}{
\begin{tabular}{l|cccc|cccc|cccc}
\Xhline{1.2pt}
\rowcolor{CadetBlue!20}
\textbf{Variant} & \multicolumn{4}{c|}{\textbf{Client Unlearning}} & \multicolumn{4}{c|}{\textbf{Class Unlearning}} & \multicolumn{4}{c}{\textbf{Sample Unlearning}} \\
\cmidrule(lr){2-5}\cmidrule(lr){6-9}\cmidrule(lr){10-13}
\rowcolor{CadetBlue!20}
 & $F\text{-}R@1\;\downarrow$ & $F\text{-}R@5\;\downarrow$ & $F\text{-}R@10\;\downarrow$ & $MIA$ & $F\text{-}R@1\;\downarrow$ & $F\text{-}R@5\;\downarrow$ & $F\text{-}R@10\;\downarrow$ & $MIA$ & $F\text{-}R@1\;\downarrow$ & $F\text{-}R@5\;\downarrow$ & $F\text{-}R@10\;\downarrow$ & $MIA$ \\
\Xhline{1.2pt}
\underline{\texttt{EASE}} & \best{0.3} \down{0.2} & \best{1.0} \up{0.5} & \best{2.2} \up{0.1} & 16.9 \gap{0.2} & \best{0.3} \down{0.1} & \best{1.5} \down{0.8} & \best{2.6} \down{1.2} & 25.4 \gap{1.2} & \best{0.1} & \best{2.2} \down{1.1} & \best{5.5} \down{3.1} & 16.7 \\
\rowcolor{gray!10}
w/o BKE$^\dagger$ & 28.5 \down{28.4} & 52.4 \down{50.9} & 65.8 \down{63.5} & \second{36.8} \gap{20.1} & 24.8 \down{24.6} & 46.2 \down{45.5} & 58.7 \down{57.3} & \second{43.2} \gap{19.0} & 22.6 \down{22.5} & 44.3 \down{43.2} & 57.2 \down{54.8} & \second{32.5} \gap{15.8} \\
w/o BKE$^\ddagger$ & 35.7 \down{35.6} & 62.5 \down{61.0} & 74.8 \down{72.5} & \best{38.5} \gap{21.8} & 30.5 \down{30.3} & 55.8 \down{55.1} & 68.5 \down{67.1} & \best{45.6} \gap{21.4} & 28.4 \down{28.3} & 53.6 \down{52.5} & 66.4 \down{64.0} & \best{34.1} \gap{17.4} \\
\rowcolor{gray!10}
w/o GSD & 18.4 \down{18.3} & 42.5 \down{41.0} & 56.3 \down{54.0} & 32.4 \gap{15.7} & 15.6 \down{15.4} & 35.8 \down{35.1} & 49.2 \down{47.8} & 40.8 \gap{16.6} & 14.2 \down{14.1} & 32.5 \down{31.4} & 45.8 \down{43.4} & 30.6 \gap{13.9} \\
w/o Lock & \second{4.2} \down{4.1} & \second{9.8} \down{8.3} & \second{15.6} \down{13.3} & 19.5 \gap{2.8} & \second{3.5} \down{3.3} & \second{7.8} \down{7.1} & \second{12.4} \down{11.0} & 27.8 \gap{3.6} & \second{2.8} \down{2.7} & \second{6.5} \down{5.4} & \second{10.8} \down{8.4} & 18.6 \gap{1.9} \\
\hline
\rowcolor{gray!10}
\textit{Retrain} & 0.1 & 1.5 & 2.3 & 16.7 & 0.2 & 0.7 & 1.4 & 24.2 & 0.1 & 1.1 & 2.4 & 16.7 \\
\Xhline{1.2pt}
\end{tabular}}

\vspace{1pt}

\resizebox{\linewidth}{!}{
\begin{tabular}{l|cccc|cccc|cccc}
\Xhline{1.2pt}
\rowcolor{CadetBlue!20}
\textbf{Variant} & $R\text{-}R@1\;\uparrow$ & $R\text{-}R@5\;\uparrow$ & $R\text{-}R@10\;\uparrow$ & $\rho\;\downarrow$ & $R\text{-}R@1\;\uparrow$ & $R\text{-}R@5\;\uparrow$ & $R\text{-}R@10\;\uparrow$ & $\rho\;\downarrow$ & $R\text{-}R@1\;\uparrow$ & $R\text{-}R@5\;\uparrow$ & $R\text{-}R@10\;\uparrow$ & $\rho\;\downarrow$ \\
\Xhline{1.2pt}
\underline{\texttt{EASE}} & \best{86.5} \down{4.2} & \best{99.2} \down{0.7} & \best{99.8} \down{0.2} & \best{0.08} \down{0.08} & \best{80.8} \down{8.6} & \best{96.6} \down{3.3} & \best{99.3} \down{0.7} & \best{0.07} \down{0.07} & \best{79.4} \down{1.2} & \best{97.7} & \best{99.2} \down{0.2} & \best{0.06} \down{0.06} \\
\rowcolor{gray!10}
w/o BKE$^\dagger$ & 65.2 \down{25.5} & 89.4 \down{10.5} & 94.1 \down{5.9} & 0.65 \down{0.65} & 74.5 \down{14.9} & 93.8 \down{6.1} & 97.6 \down{2.4} & 0.62 \down{0.62} & 71.5 \down{9.1} & 94.8 \down{2.9} & 98.0 \down{1.4} & 0.58 \down{0.58} \\
w/o BKE$^\ddagger$ & 63.8 \down{26.9} & 88.1 \down{11.8} & 93.2 \down{6.8} & 0.72 \down{0.72} & 73.1 \down{16.3} & 92.7 \down{7.2} & 96.9 \down{3.1} & 0.69 \down{0.69} & 70.2 \down{10.4} & 93.9 \down{3.8} & 97.3 \down{2.1} & 0.65 \down{0.65} \\
\rowcolor{gray!10}
w/o GSD & 24.6 \down{66.1} & 52.3 \down{47.6} & 68.9 \down{31.1} & 0.68 \down{0.68} & 32.4 \down{57.0} & 62.5 \down{37.4} & 78.3 \down{21.7} & 0.65 \down{0.65} & 38.6 \down{42.0} & 68.4 \down{29.3} & 82.5 \down{16.9} & 0.61 \down{0.61} \\
w/o Lock & \second{69.8} \down{20.9} & \second{93.3} \down{6.6} & \second{96.5} \down{3.5} & \second{0.22} \down{0.22} & \second{80.0} \down{9.4} & \second{96.3} \down{3.6} & \second{99.0} \down{1.0} & \second{0.20} \down{0.20} & \second{78.6} \down{2.0} & \second{97.2} \down{0.5} & \second{98.8} \down{0.6} & \second{0.18} \down{0.18} \\
\hline
\rowcolor{gray!10}
\textit{Retrain} & 90.7 & 99.9 & 100.0 & 0 & 89.4 & 99.9 & 100.0 & 0 & 80.6 & 97.7 & 99.4 & 0 \\
\Xhline{1.2pt}
\end{tabular}}
\end{table*}

\subsection{Sensitivity (Q3)}
\label{sec:q3}

We sweep $\delta$, $\alpha$, $K$, and $\beta$ on Flickr30K/CLIP-B/32 (full curves: Figure~\ref{fig:sensitivity}). The $\delta$ plateau gives a stable anchor/support boundary, the monotone $\alpha$ response confirms re-anchoring closure without aggressive lock tuning, and $K$/$\beta$ sweeps verify robustness to federation scale and heterogeneity.

\section{Conclusion}
\label{sec:conclusion}

We studied federated multimodal unlearning, identifying three reconstruction anchors (cross-modal coupling, forget-exclusive directions, post-unlearning FedAvg drift) and closing them through \underline{\texttt{EASE}}'s bilateral excision, subspace decomposition, and direction-selective Forget Lock; EASE moves forget/retain closer to retrain than evaluated baselines. The Anchor Principle is articulated for image-text dual-encoder retrieval (CLIP/SigLIP-style); late-fusion, generative, and multimodal-LLM extensions remain future work.

\begin{ack}

\end{ack}


\clearpage

\appendix

\etocdepthtag.toc{appendix}
\etocsettagdepth{mainbody}{none}
\etocsettagdepth{appendix}{subsection}
\etocsettocstyle{\section*{Appendix Contents}\vspace{-6pt}\noindent\hrulefill\par\vspace{4pt}}{\par\noindent\hrulefill\vspace{10pt}}
\tableofcontents

\clearpage
\section{Related Work}
\label{app:related_work}

\paragraph{Federated Unlearning.}
Federated unlearning reconstructs the influence of a withdrawn client without retraining. Historical reconstruction calibrates remaining clients' updates~\cite{Liu2021IWQOS, Cao2023SP, Liu2024ICML}; curvature- and direction-based approaches localize the forget contribution via second-order signals (SoUL~\cite{Jia2024EMNLP}) or orthogonal subspace projection (FedOSD~\cite{Pan2025AAAI}); task-vector methods (NoT~\cite{Khalil2025CVPR}) subtract a directional vector of the forget client's accumulated update; recent extensions reduce cost through nonlinear functional refinement (FFMU~\cite{Che2023ICML}) or sparse-adapter overwriting (FUSED~\cite{Zhong2025CVPR}); and verification-oriented studies~\cite{Gao2024TDSC, Halimi2022ARXIV, Wang2022WWW, Wu2022ARXIV} provide complementary auditing tools. However, these methods assume a single-modality loss and a one-shot view of unlearning, and thus do not address the cross-modal reconstruction channel or continued-training drift in federated multimodal settings.

\paragraph{Centralized Machine Unlearning.}
Centralized methods assume a trusted curator with full data access. Exact approaches retrain affected data shards~\cite{Bourtoule2021SP}; approximate approaches edit a trained model via gradient ascent on the forget loss~\cite{Thudi2022USENIX, Golatkar2020CVPR}, influence-function subtraction~\cite{Koh2017ICML}, distillation from a retain-only teacher~\cite{Kurmanji2023NeurIPS}, or importance-weighted regularization on parameter coordinates~\cite{Kirkpatrick2017PNAS, Zenke2017ICML, Aljundi2018ECCV}. However, none of these methods is designed for unlearning that must remain valid under continued training over heterogeneous clients.

\paragraph{Multimodal Unlearning.}
A small but growing body of work targets unlearning in multimodal contrastive models. CLIP-specific erasers~\cite{Poppi2024ECCVa, Poppi2024ECCVb} remove undesired image-text associations via fine-tuning on a designated forget set, and recent multimodal unlearners~\cite{Li2026AAAI} acknowledge cross-modal coupling by issuing dual-branch updates. However, all these efforts are centralized and assume unrestricted access to raw pairs, so they do not transfer to the federated regime where the server only sees gradient summaries.

\paragraph{Gradient Subspace and Importance-weighted Regularization.}
Two long-standing families in continual learning are closely related to our subspace and lock components. Gradient subspace methods (GPM~\cite{Saha2021ICLR}, Adam-NSCL~\cite{Wang2021CVPR}, OGD~\cite{Farajtabar2020AISTATS}) confine new updates to the orthogonal complement of a single past-task subspace, while importance-weighted regularizers (EWC~\cite{Kirkpatrick2017PNAS}, SI~\cite{Zenke2017ICML}) weight a quadratic drift penalty by per-parameter importance scores. However, both families were conceived for centralized continual learning rather than unlearning, and neither separates forget-exclusive directions from retain-shared directions inside a forget subspace nor remains valid under continued federated training.

\definecolor{headerblue}{HTML}{DBD9E4}
\begin{table}[h]
\centering
\caption{Notation and Definitions.}
\label{tab:notations}
\definecolor{grayrow}{HTML}{F2F2F2}
\rowcolors{2}{grayrow}{white}
\setlength{\aboverulesep}{0pt}
\setlength{\belowrulesep}{0pt}
\setlength{\doublerulesep}{3pt}
\doublerulesepcolor{white}
\begin{tabular}{c||l}
\toprule
\rowcolor{headerblue}
\textbf{Notation} & \textbf{Definition} \\
\midrule
$K$ & The number of clients. \\
$\mathcal{D}_k$ & The private multimodal dataset for the $k$-th client. \\
$N_k$ & The number of image-text pairs for the $k$-th client. \\
$(x_v, x_t)$ & An image-text pair in $\mathcal{D}_k$. \\
$w_v,\, w_t$ & The trainable parameters for the visual and language branches. \\
$w$ & The joint parameter vector $[w_v;\, w_t]$. \\
$z_v,\, z_t$ & The visual and textual embeddings in the shared space. \\
$m$ & The dimension of the embedding space. \\
$\gamma$ & The temperature in the InfoNCE loss. \\
$\mathcal{L}_a$ & The symmetric InfoNCE alignment loss. \\
$J_v,\, J_t$ & The Jacobians $\partial z_v / \partial w_v$ and $\partial z_t / \partial w_t$. \\
$\mathcal{D}_f$ & The forget set whose influence is to be removed. \\
$\mathcal{D}_r$ & The retain set $\mathcal{D} \setminus \mathcal{D}_f$. \\
$w_n$ & The model after standard federated training. \\
$\tilde{w}$ & The retrain model obtained by training from scratch on $\mathcal{D}_r$. \\
$w^*$ & The unlearned model. \\
$\rho$ & Alignment residual, computed per forget pair and reported as the median over $\mathcal{D}_f$. \\
$\Delta w_k[t]$ & The parameter update of client $k$ at round $t$. \\
$H_k$ & The gradient update history of client $k$. \\
$G_f,\, G_r$ & The forget and retain gradient matrices. \\
$\Phi_f,\, \Phi_r$ & The orthonormal bases of $\mathcal{S}_f$ and $\mathcal{S}_r$. \\
$\mathcal{S}_f,\, \mathcal{S}_r$ & The forget and retain subspaces. \\
$p,\, q$ & The dimensions of $\mathcal{S}_f$ and $\mathcal{S}_r$. \\
$\tau_e$ & The energy threshold for SVD truncation. \\
$M$ & The cross-correlation matrix $\Phi_f^\top \Phi_r$. \\
$\cos\varphi_i$ & The entanglement coefficient of the $i$-th canonical direction. \\
$c_i$ & The $i$-th canonical direction of $\mathcal{S}_f$. \\
$\delta$ & The threshold separating removable anchor directions from retain support directions. \\
$B_u,\, B_e$ & The unique (anchor) and entangled (support) basis matrices. \\
$\mathcal{U},\, \mathcal{E}$ & The subspaces spanned by removable anchor directions and retain support directions. \\
$\Pi_{u,\mu}$ & The orthogonal projector onto the unique subspace $\mathcal{U}_\mu$, defined as $B_{u,\mu}(B_{u,\mu})^\top$. \\
$\mathcal{L}_f$ & The Forget Lock regularizer. \\
$\alpha$ & The Forget Lock regularization strength. \\
$\lr$ & The learning rate. \\
$G_u$ & The uniform upper bound on the projected alignment gradient $\|\Pi_{u,\mu}\nabla_{w_\mu}\mathcal{L}_a\|_2$. \\
\bottomrule
\end{tabular}
\end{table}

\section{Technical Appendix}
\label{sec:appendix}

\subsection{Notation Table}
\label{sec:app_notation}

We summarize the commonly used notation throughout this paper in Table~\ref{tab:notations}.

\subsection{Dataset Details}
\label{sec:app_datasets}

We evaluate on three standard image-text retrieval benchmarks.

\paragraph{Flickr30K \venue{TACL'14}~\cite{Young2014TACL}.}
Flickr30K contains $31{,}783$ images collected from Flickr, each annotated with five human-written English captions, for a total of $158{,}915$ image-text pairs. Following the Karpathy split~\cite{Karpathy2015CVPR}, we use $29{,}000$ images for training, $1{,}014$ for validation, and $1{,}000$ for the test set used to compute Recall@$k$; the remaining $769$ images are not assigned to any split and are omitted from our evaluation, consistent with prior retrieval work using this split. The captions cover diverse everyday scenes, making Flickr30K a standard benchmark for general-purpose cross-modal retrieval.

\paragraph{MS COCO Captions \venue{ECCV'14}~\cite{Lin2014ECCV,Chen2015ARXIV}.}
MS COCO is a large-scale dataset comprising $123{,}287$ images of common objects in context, each paired with five human-written captions. We again adopt the Karpathy split, using $113{,}287$ training images, $5{,}000$ validation images, and the $5{,}000$-image test split for retrieval evaluation. Compared with Flickr30K, MSCOCO covers a broader object taxonomy and a wider caption style, and stresses retrieval on a larger candidate pool.

\paragraph{TextCaps \venue{ECCV'20}~\cite{Sidorov2020ECCV}.}
TextCaps extends the MSCOCO-style caption format with scene-text understanding: each of the $28{,}408$ images contains visible textual content (signs, logos, product labels), and captions are written to reference this text where appropriate. We use the official splits of $21{,}953$ training, $3{,}166$ validation, and $3{,}289$ test images. TextCaps complements the other two datasets by stressing fine-grained alignment between language and visual text, a regime where forget-set leakage is particularly sensitive to bilateral closure of the Modality Anchor.

\paragraph{Federated partition.}
For every dataset, training images (together with their captions) are distributed across $K$ clients via a Dirichlet partition with concentration parameter $\beta$ over pseudo-cluster labels obtained by KMeans on concatenated visual-textual embeddings of the pretrained backbone (Section~\ref{par:scenario_setup}). The default configuration is $K{=}10$ and $\beta{=}0.5$, yielding moderately non-IID client shards; we sweep both parameters in Figure~\ref{fig:sensitivity_fed}. Forget sets are constructed on top of this partition according to the unlearning scenario, as described in Appendix~\ref{sec:app_impl}.

\subsection{Baseline Descriptions}
\label{sec:app_baselines}

This section provides a comprehensive overview of the nine counterparts employed in our experiments. We group them by methodological category.

\paragraph{Federated unlearning methods.}

\begin{itemize}[leftmargin=*]
    \item \textbf{FedEraser} \venue{IWQOS'21}~\cite{Liu2021IWQOS}. FedEraser is one of the earliest federated unlearning methods. It reconstructs an approximation of the retrained model by combining cached historical parameter updates from the remaining clients, replaying them through a calibrated aggregation procedure. This avoids the cost of full retraining while producing a model that behaves similarly to one never trained on the forget client's data.

    \item \textbf{FedRecover} \venue{S\&P'23}~\cite{Cao2023SP}. FedRecover originally targets recovery from poisoning attacks but is naturally adapted to federated unlearning. It stores per-round gradient and model information, then uses this history to estimate what the global model would have looked like if the forget client had been absent, applying a low-cost correction instead of rerunning training from scratch.

    \item \textbf{Ferrari} \venue{ICML'24}~\cite{Liu2024ICML}. Ferrari performs rapid approximate retraining guided by cached client gradients. The server constructs a compact representation of each client's contribution and selectively reverses the forget client's influence, achieving significantly lower wall-clock time than retrain while preserving retrieval accuracy on the retain set.

    \item \textbf{FedOSD} \venue{AAAI'25}~\cite{Pan2025AAAI}. FedOSD applies orthogonal subspace decomposition to client update directions. It identifies the subspace spanned by the forget client's updates and projects the global model away from that subspace, selectively removing the forget component. The method is close in spirit to ours, but operates on single-modality models and does not separate removable anchor directions from retain support directions inside the forget subspace.

    \item \textbf{NoT} \venue{CVPR'25}~\cite{Khalil2025CVPR}. NoT treats each client's historical contribution as a task vector and performs unlearning by subtracting the forget client's task vector from the global model. This yields aggressive erasure on the forget set but frequently damages retained knowledge, since the task vector subtraction is applied uniformly without separating shared and unique components.

    \item \textbf{SoUL} \venue{EMNLP'24}~\cite{Jia2024EMNLP}. SoUL employs second-order curvature information to localize the forget contribution. Using approximate Hessian or Fisher-information-based signals, it estimates which parameters are most relevant to the forget data and perturbs only those. While more targeted than gradient-based baselines, it incurs nontrivial computational cost due to second-order estimation.

    \item \textbf{FFMU} \venue{ICML'23}~\cite{Che2023ICML}. FFMU accelerates federated unlearning by modeling local unlearning models as output functions of a Nemytskii operator and leveraging nonlinear functional analysis to refine them. It achieves competitive unlearning quality with reduced computation, but operates on single-modality models and does not address the Modality Anchor present in multimodal architectures.

    \item \textbf{FUSED} \venue{CVPR'25}~\cite{Zhong2025CVPR}. FUSED performs reversible federated unlearning by training independent sparse adapters that overwrite forget-specific knowledge. It identifies the most affected layers through sensitivity analysis and retrains only sparse adapters on those layers, reducing unlearning cost while keeping the process reversible. Like other federated unlearning methods, it operates on single-modality models.
\end{itemize}

\paragraph{Centralized unlearning methods adapted to the federated setting.}

\begin{itemize}[leftmargin=*]
    \item \textbf{GradAscent} \venue{USENIX Security'22}~\cite{Thudi2022USENIX}. Gradient ascent is a simple and widely used unlearning baseline that performs gradient updates on the forget set using the negative of the original training loss. We adapt it to the federated setting by running localized gradient ascent on the forget client followed by standard federated aggregation on the remaining clients. Although straightforward to implement, it often overshoots and degrades retain-set performance.
\end{itemize}

\subsection{Evaluation Metrics}
\label{sec:app_metrics}

We report four groups of metrics to capture unlearning completeness, retention integrity, privacy leakage, and efficiency.

\paragraph{Retrieval accuracy (Recall@$k$) \venue{CVPR'15}~\cite{Karpathy2015CVPR}.}
Recall@$k$ is the standard metric for cross-modal retrieval. Given a query from one modality, it measures whether the ground-truth match in the other modality is ranked within the top-$k$ retrieved items. We report Recall@$k$ for $k \in \{1, 5, 10\}$, averaged over both image-to-text and text-to-image directions. Three splits are evaluated: (i) the forget set $\bar{\mathcal{D}}_f$ for unlearning completeness, where lower values indicate that the model has forgotten the target pairs; (ii) the retain set $\bar{\mathcal{D}}_r$ for retention integrity, where higher values indicate that retained knowledge is preserved; and (iii) the held-out test set for generalization, where higher values indicate that unlearning does not harm out-of-sample performance.

\paragraph{Membership inference attack (MIA) \venue{S\&P'17}~\cite{Shokri2017SP}.}
We adopt the shadow-model MIA protocol of Shokri et al.: for each target configuration, we independently train multiple shadow models with the same architecture, training recipe, and data distribution as the target, on disjoint random splits with known member/non-member labels. Per-sample loss and logit statistics from shadow models are used to train a binary attack classifier, which is then applied to the target unlearned model to predict whether each forget-set sample was an original training member. The reported MIA score is the attack's true-positive rate on the forget set, i.e., the fraction of forget samples that the classifier labels as members. A well-unlearned model maps forget samples into the loss/logit regime of non-members, so its MIA score moves toward the retrain reference; larger values indicate residual memorization that an attacker can exploit. We therefore measure privacy leakage by the gap $|\mathrm{MIA}_{\text{method}} - \mathrm{MIA}_{\text{retrain}}|$, with a small gap meaning the unlearned model is close to a model that never saw the forget data under this attack. Absolute MIA values can fall below $50\%$ because the shadow-model classifier inherits the class prior and loss-distribution skew of the training pool; the operationally meaningful quantity is the gap to the retrain reference, not the distance to $50\%$.

\paragraph{LiRA true-positive rate at $1\%$ false-positive rate \venue{S\&P'22}~\cite{Carlini2022SP}.}
Aggregate MIA accuracy can hide worst-case leakage because an attacker usually cares about confidently identifying a small fraction of training members rather than performing well on average. We therefore also report the likelihood-ratio attack (LiRA) of Carlini et al.\ evaluated on the same shadow-model pool. For each candidate sample, LiRA computes a per-example likelihood-ratio statistic between two Gaussian distributions fitted over shadow-model logits (one for models that contain the sample, one for models that do not); the threshold is calibrated so that the false-positive rate on held-out non-members equals $1\%$, and we report the resulting true-positive rate ($\mathrm{TPR}@\mathrm{FPR}{=}1\%$) on the forget set. Values near the $1\%$ random-guess line indicate that forget samples are close to non-members under a targeted low-FPR attacker, while larger values indicate residual memorization that the aggregate MIA score can underestimate.

\paragraph{Communication cost~\cite{Kairouz2021FNT}.}
We measure the total payload transmitted between clients and the server during the unlearning procedure, in megabytes. This includes model broadcasts, client updates, and any auxiliary tensors such as unique-subspace bases $B_{u,v}$ and $B_{u,t}$. Lower values indicate lower bandwidth overhead for deployment.

\subsection{Implementation Details}
\label{sec:app_impl}

\paragraph{Model architecture.}
We use frozen pretrained multimodal encoders in bf16 as the backbone. Two types of trainable modules are added in fp32: (1) LoRA adapters inserted into the attention layers of each encoder, and (2) two-layer MLP projectors that map each encoder's output to a shared 256-dimensional embedding space (one projector per modality).

\paragraph{Gradient delta construction.}
The per-round update $\Delta w_k[t]$ is computed differently for the two module types. For the projectors, we directly compute the parameter difference between the local and global states. For LoRA layers, each adapter contributes a low-rank modification $BA$ to the frozen weight; we compute the effective weight delta as the difference in the low-rank product before and after local training, rather than differencing $B$ and $A$ separately, since only the product has a well-defined interpretation in the full parameter space. These deltas are then column-stacked to form the gradient matrices $G_f$ and $G_r$ (Section~\ref{sec:subspace}).

\paragraph{Unlearning scenario setup.}
\label{par:scenario_setup}
For client unlearning, a designated client $k^*$ withdraws from the federation; all its historical updates form $G_f$, and $k^*$ does not participate in subsequent rounds. For sample unlearning, the server specifies a global forget set $\mathcal{D}_f$ whose samples typically span multiple clients (we partition $\mathcal{D}_f$ uniformly across clients at data-splitting time); each client with overlapping samples performs one epoch of training on its local forget subset and uploads the resulting gradient, which is column-stacked into $G_f$, and every client subsequently continues training on $\mathcal{D}_k \setminus \mathcal{D}_f$. For class unlearning, since image-text retrieval datasets (Flickr30K, COCO) lack categorical labels, we pre-compute pseudo-class assignments by running KMeans on concatenated visual-textual embeddings of the pretrained backbone and distribute each pseudo-class uniformly across clients; the server specifies the target cluster via a textual descriptor matched to the nearest cluster centroid, every client that holds samples in that cluster contributes a per-client gradient to form $G_f$, and all clients exclude the target-cluster samples from subsequent training.

\subsection{Hyperparameter Sensitivity}
\label{sec:app_sensitivity}

We sweep the entanglement threshold $\delta$ and the Forget Lock strength $\alpha$ on Flickr30K with CLIP-B/32 across the three unlearning scenarios. For each configuration we run three independent trials with different random seeds and report mean and standard deviation; the resulting curves are shown in Figure~\ref{fig:sensitivity}.

\begin{figure}[h]
    \centering
    \includegraphics[width=\linewidth]{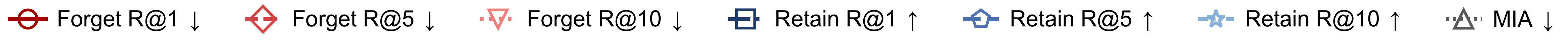}

    \subfloat[Client / $\delta$ sweep]{\includegraphics[width=0.32\linewidth]{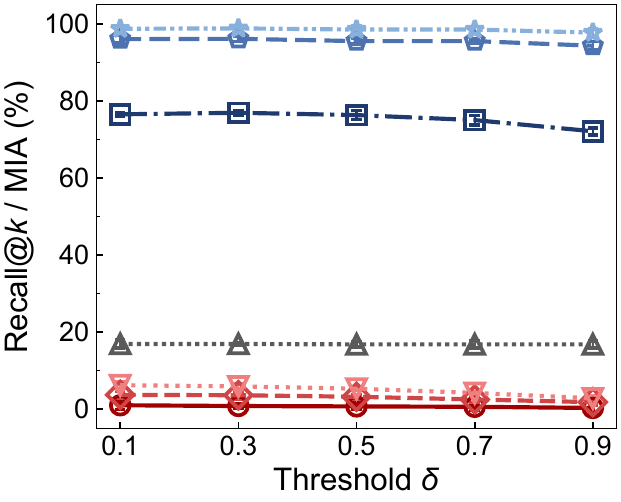}\label{fig:sens_delta_client}}
    \hfill
    \subfloat[Class / $\delta$ sweep]{\includegraphics[width=0.32\linewidth]{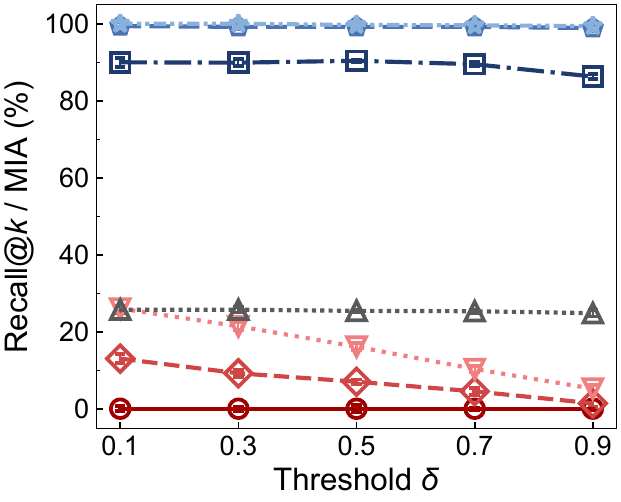}\label{fig:sens_delta_class}}
    \hfill
    \subfloat[Sample / $\delta$ sweep]{\includegraphics[width=0.32\linewidth]{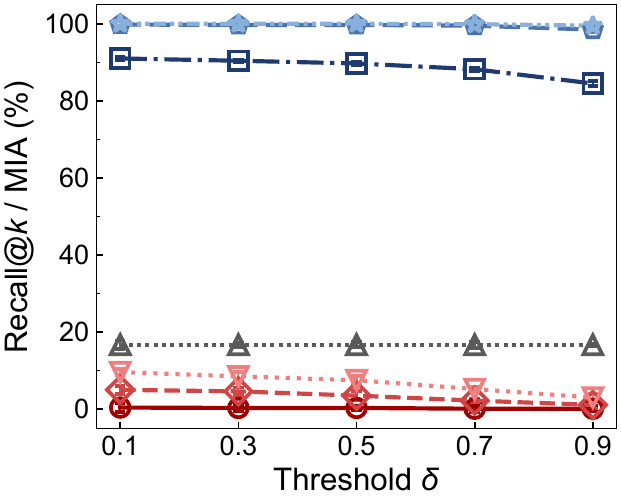}\label{fig:sens_delta_sample}}\\
    \subfloat[Client / $\alpha$ sweep]{\includegraphics[width=0.32\linewidth]{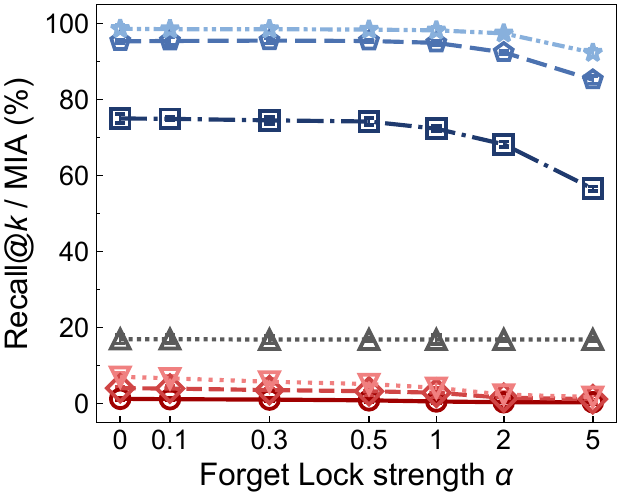}\label{fig:sens_alpha_client}}
    \hfill
    \subfloat[Class / $\alpha$ sweep]{\includegraphics[width=0.32\linewidth]{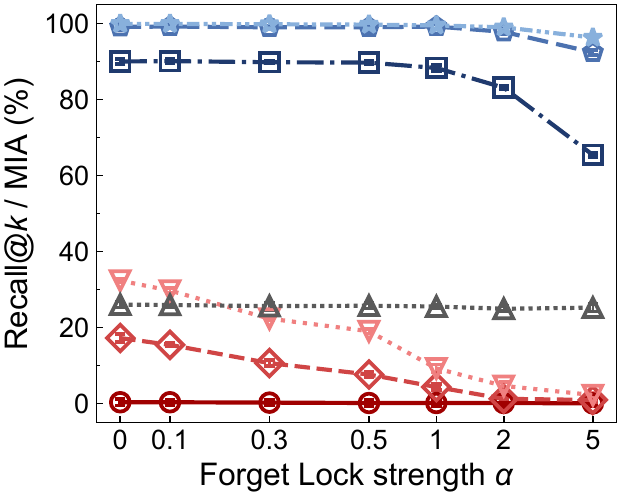}\label{fig:sens_alpha_class}}
    \hfill
    \subfloat[Sample / $\alpha$ sweep]{\includegraphics[width=0.32\linewidth]{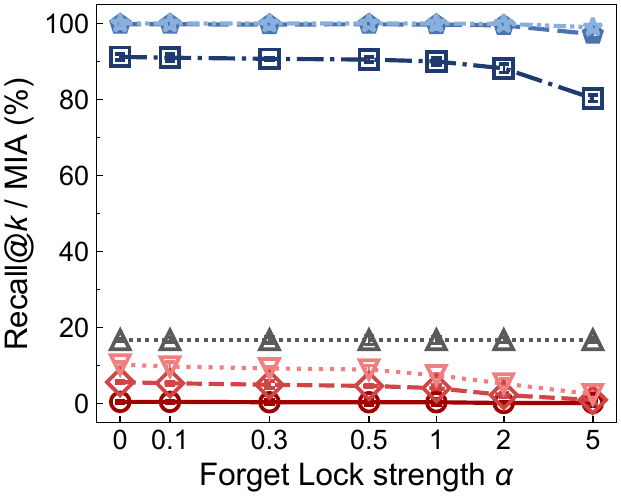}\label{fig:sens_alpha_sample}}
    \caption{Sensitivity of \underline{\texttt{EASE}} to the entanglement threshold $\delta$ (top row) and the Forget Lock strength $\alpha$ (bottom row) on Flickr30K with CLIP-B/32 across three unlearning scenarios. Each point is the mean over three independent runs with different random seeds; error bars indicate the standard deviation.}
    \label{fig:sensitivity}
\end{figure}

\paragraph{Findings on $\delta$ (top row).}
The $\delta$ sweep exhibits a wide plateau where Forget R@$k$ is already low while Retain R@$k$ stays close to the retrain reference, with noticeable Retain degradation only at the very high end of $\delta$. This plateau empirically matches the $\delta^2$-controlled collateral energy in Theorem~\ref{thm:retention} and indicates that the anchor/support boundary remains stable across a broad operating regime: small $\delta$ leaves Unique-Subspace Anchors partially open, while excessive $\delta$ misclassifies retain support as anchor and degrades retention.

\paragraph{Findings on $\alpha$ (bottom row).}
The Forget Lock strength exhibits a clearer trade-off. The projection alone already achieves substantial closure, so a moderate Lock suppresses Temporal Re-anchoring at negligible Retain cost; pushing $\alpha$ too high over-constrains legitimate updates in $\mathcal{U}_\mu^\perp$ and degrades Retain. This matches the durability bound of Theorem~\ref{thm:durability}, and shows that the Temporal Re-anchoring closure is delivered by a broad band of $\alpha$ values rather than a single operating point.

\paragraph{Findings on $K$ and $\beta$ (federated configuration).}
Forget R@$k$ remains low across the swept $K$ and $\beta$ ranges, consistent with bilateral projection suppressing the Modality Anchor under different federation sizes and heterogeneity levels. Retain R@$k$ varies modestly but remains stable across both sweeps, suggesting limited sensitivity to these federated configuration choices in our setting.

\begin{figure}[h]
    \centering
    \includegraphics[width=\linewidth]{figures/sensitivity/sensitivity_legend.pdf}

    \subfloat[Client / $K$ sweep]{\includegraphics[width=0.32\linewidth]{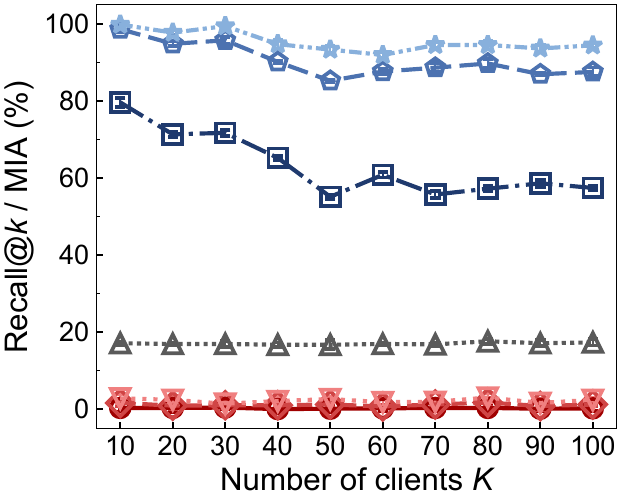}\label{fig:sens_K_client}}
    \hfill
    \subfloat[Class / $K$ sweep]{\includegraphics[width=0.32\linewidth]{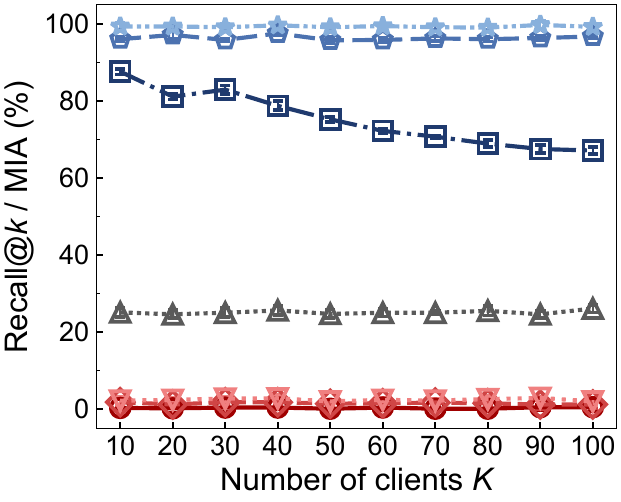}\label{fig:sens_K_class}}
    \hfill
    \subfloat[Sample / $K$ sweep]{\includegraphics[width=0.32\linewidth]{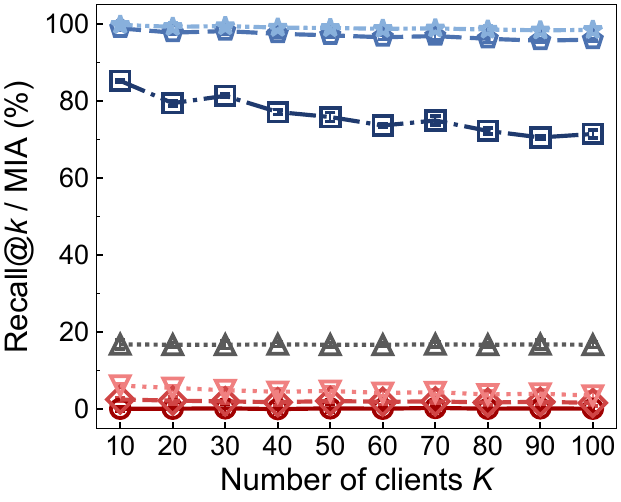}\label{fig:sens_K_sample}}\\
    \subfloat[Client / $\beta$ sweep]{\includegraphics[width=0.32\linewidth]{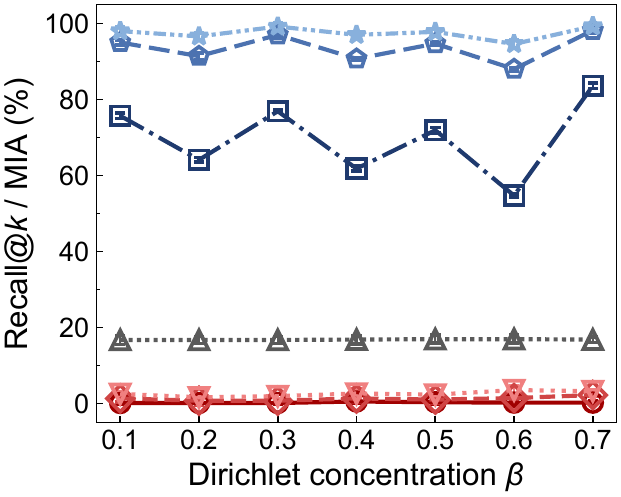}\label{fig:sens_beta_client}}
    \hfill
    \subfloat[Class / $\beta$ sweep]{\includegraphics[width=0.32\linewidth]{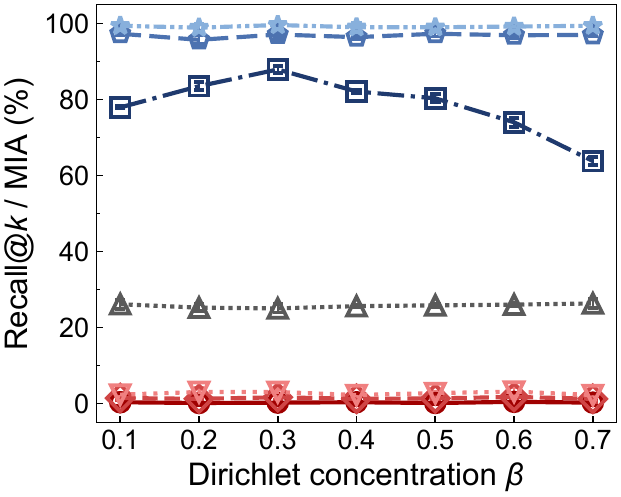}\label{fig:sens_beta_class}}
    \hfill
    \subfloat[Sample / $\beta$ sweep]{\includegraphics[width=0.32\linewidth]{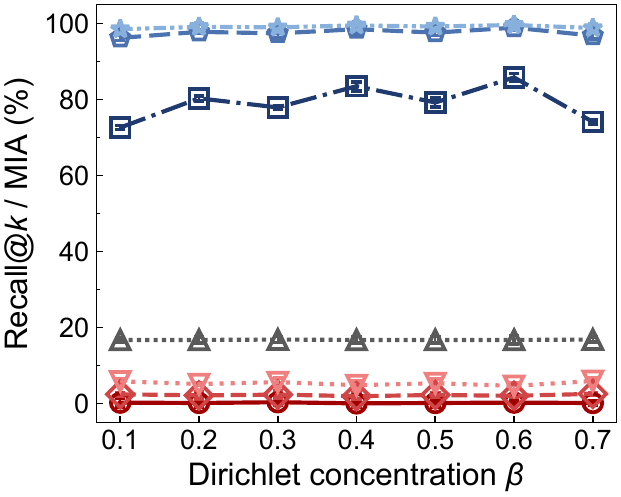}\label{fig:sens_beta_sample}}
    \caption{Sensitivity of \underline{\texttt{EASE}} to the federation size $K$ (top row) and the Dirichlet concentration $\beta$ (bottom row) on Flickr30K with CLIP-B/32, evaluated under three unlearning scenarios. Each point reports the mean over three independent runs with different random seeds; error bars indicate the standard deviation.}
    \label{fig:sensitivity_fed}
\end{figure}

\paragraph{Backbone transfer.}
To check whether the $\delta$ plateau and the $\alpha$ stable band are specific to a single backbone, we repeat both sweeps with CLIP-L/14 on Flickr30K. Figure~\ref{fig:sensitivity_l14} shows the same qualitative behaviour observed on CLIP-B/32: Forget R@$k$ remains low across the swept ranges while Retain R@$k$ stays close to the retrain reference, with degradation only at the high end of $\delta$ and $\alpha$.

\begin{figure}[h]
    \centering
    \includegraphics[width=\linewidth]{figures/sensitivity/sensitivity_legend.pdf}

    \subfloat[Client / $\delta$ sweep]{\includegraphics[width=0.32\linewidth]{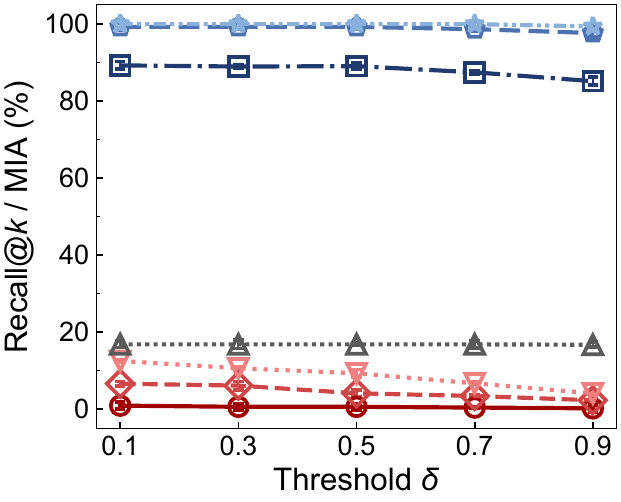}\label{fig:sens_l14_delta_client}}
    \hfill
    \subfloat[Class / $\delta$ sweep]{\includegraphics[width=0.32\linewidth]{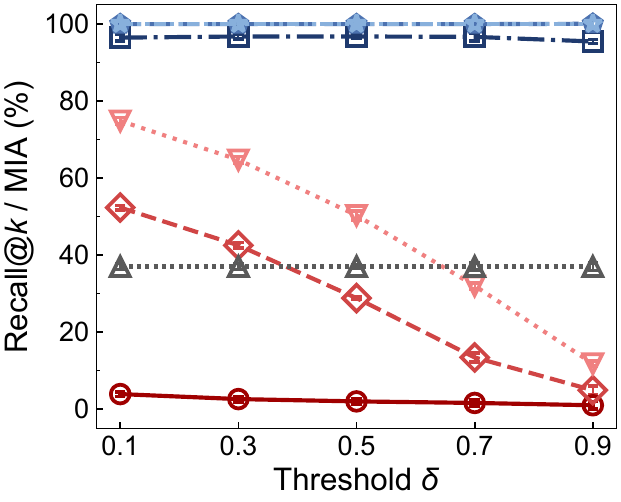}\label{fig:sens_l14_delta_class}}
    \hfill
    \subfloat[Sample / $\delta$ sweep]{\includegraphics[width=0.32\linewidth]{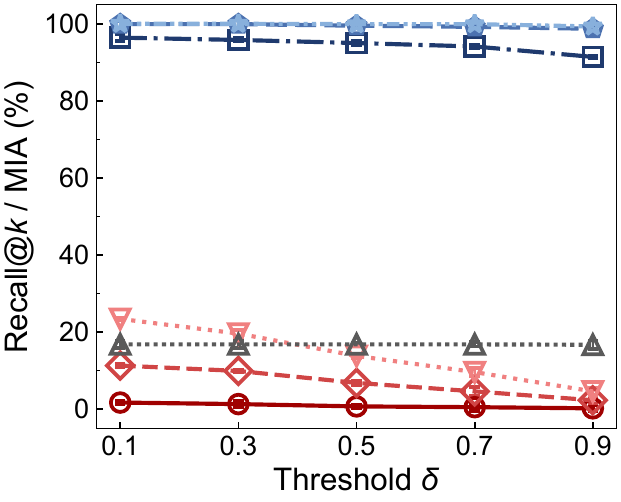}\label{fig:sens_l14_delta_sample}}\\
    \subfloat[Client / $\alpha$ sweep]{\includegraphics[width=0.32\linewidth]{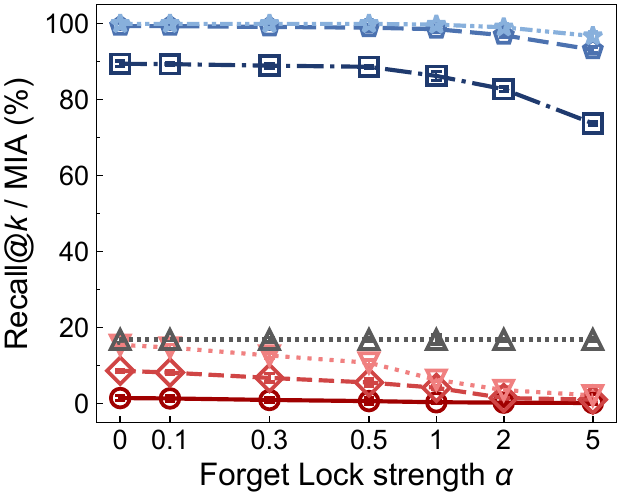}\label{fig:sens_l14_alpha_client}}
    \hfill
    \subfloat[Class / $\alpha$ sweep]{\includegraphics[width=0.32\linewidth]{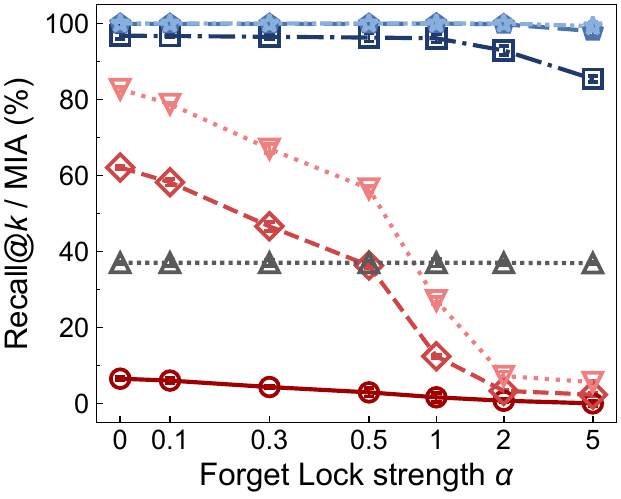}\label{fig:sens_l14_alpha_class}}
    \hfill
    \subfloat[Sample / $\alpha$ sweep]{\includegraphics[width=0.32\linewidth]{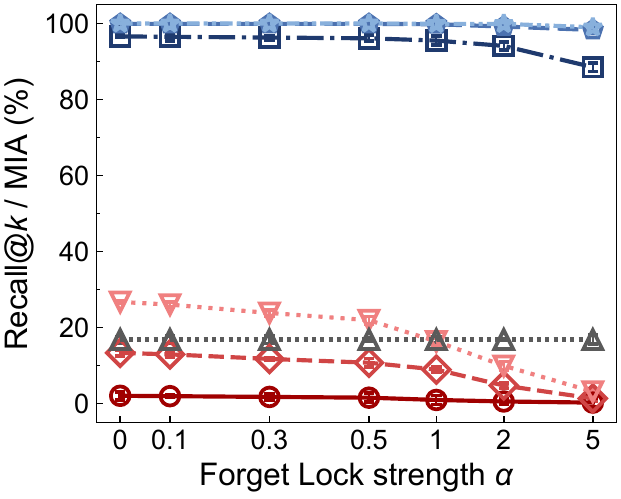}\label{fig:sens_l14_alpha_sample}}
    \caption{Sensitivity of \underline{\texttt{EASE}} to the entanglement threshold $\delta$ (top row) and the Forget Lock strength $\alpha$ (bottom row) on Flickr30K with CLIP-L/14 across three unlearning scenarios. Each point reports the mean over three independent runs with different random seeds; error bars indicate the standard deviation.}
    \label{fig:sensitivity_l14}
\end{figure}

\paragraph{Class scenario asymmetry on CLIP-L/14.}
The class panels in Figure~\ref{fig:sensitivity_l14} collapse only at $\delta\!\geq\!0.7$, while client and sample already reach the floor at small $\delta$. Two effects compound. First, class unlearning aggregates the gradient signal of an entire target cluster, so the empirical forget covariance has higher effective rank and a flatter energy spectrum than per-client or per-sample targets, and a larger $\delta$ is needed to cover the same fraction of cluster-correlated energy as removable anchor directions. Second, CLIP-L/14's wider feature space, 768 versus 512 channels, distributes cluster-correlated alignment across more redundant directions, so a small-$\delta$ truncation captures only the dominant pathway and leaves residual alignment that still supports R@5 and R@10 retrieval. Client and sample targets are intrinsically low-rank subspaces and therefore remain less sensitive to either factor across backbones.

\subsection{Visualization of Unlearning Effect}
\label{sec:app_raincloud}

We visualize per-pair cosine similarity distributions under the original model $w_n$, our method $w^*$, and the retrain reference $\tilde{w}$ in Figure~\ref{fig:raincloud_similarity}. For each image-text pair, we compute $\cos(z_v^i, z_t^i)$ between the visual and textual embeddings on the shared projector output; higher values indicate stronger cross-modal alignment, and an unlearned model is expected to produce similarities close to those of the retrain reference on forget pairs and close to the original on retain pairs. We report these per-pair similarities for three unlearning scenarios and two data splits. Each panel combines a half-violin, a strip plot, and a box plot, so the reader sees the overall shape, the individual pair values, and the median with quartiles in a single view.

\begin{figure}[h]
    \centering
    \subfloat[client forget]{\includegraphics[width=0.32\linewidth]{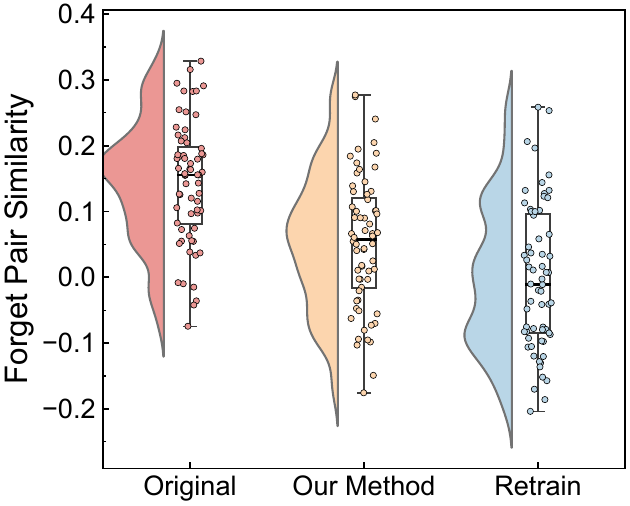}\label{fig:rc_client_forget}}
    \hfill
    \subfloat[class forget]{\includegraphics[width=0.32\linewidth]{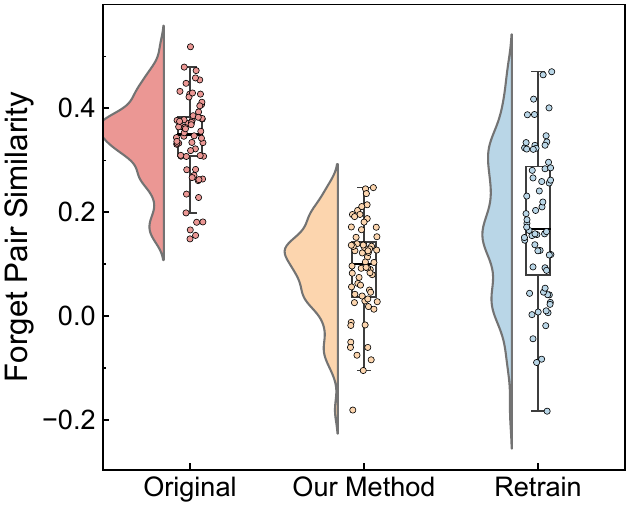}\label{fig:rc_class_forget}}
    \hfill
    \subfloat[sample forget]{\includegraphics[width=0.32\linewidth]{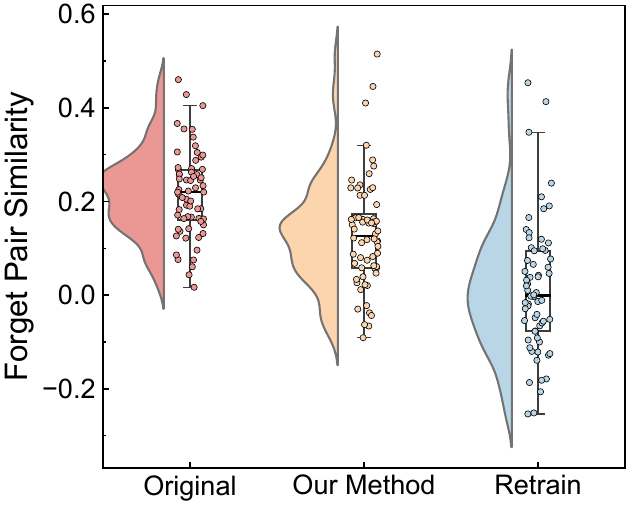}\label{fig:rc_sample_forget}}\\
    \subfloat[client retain]{\includegraphics[width=0.32\linewidth]{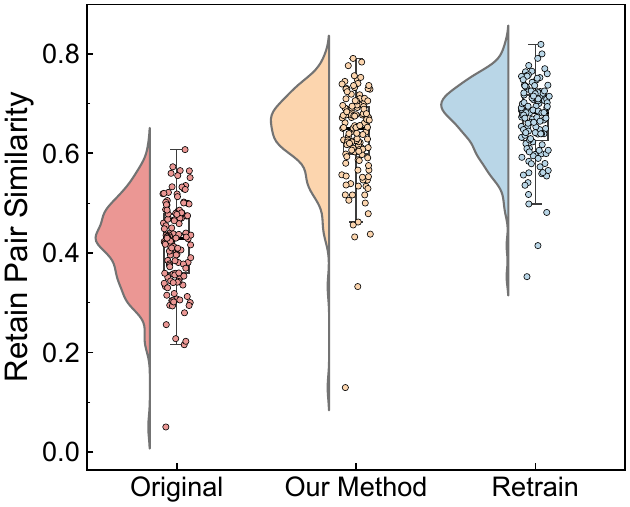}\label{fig:rc_client_retain}}
    \hfill
    \subfloat[class retain]{\includegraphics[width=0.32\linewidth]{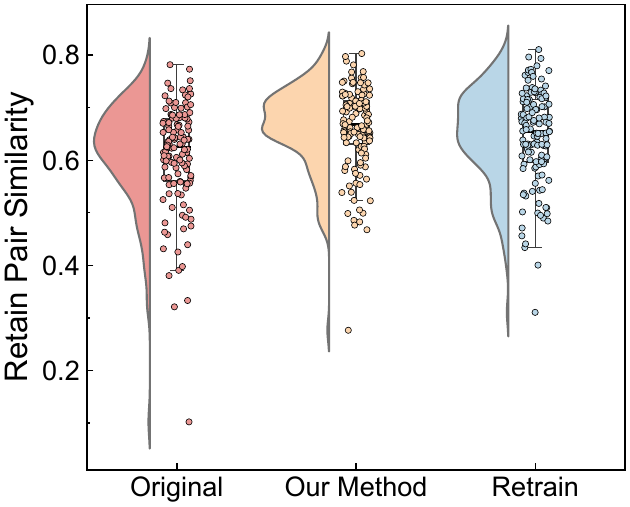}\label{fig:rc_class_retain}}
    \hfill
    \subfloat[sample retain]{\includegraphics[width=0.32\linewidth]{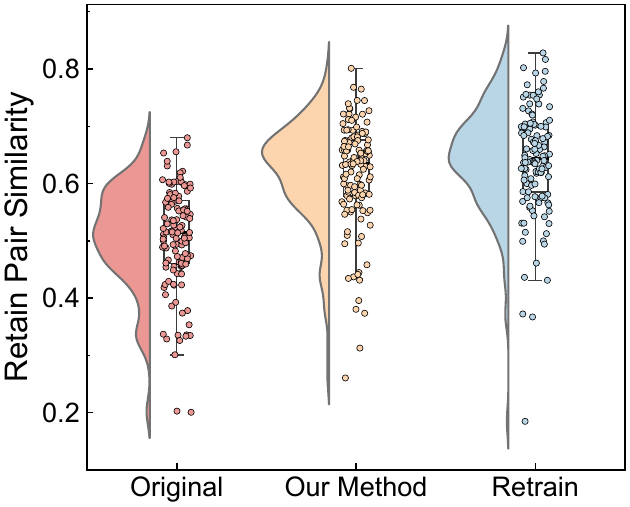}\label{fig:rc_sample_retain}}
    \caption{Per-pair image--text cosine similarity under the original, our method, and retrain. Top: forget pairs, where our method shifts toward the retrain reference. Bottom: retain pairs, where the three distributions overlap.}
    \label{fig:raincloud_similarity}
\end{figure}

Two patterns emerge. First, on forget pairs, the original model places most pairs at high similarity, reflecting the alignment learned during standard federated training; our method shifts this distribution leftward toward the retrain reference across client, class, and sample scenarios, meaning that the cross-modal alignment between forget image-text pairs is reduced toward the level of a model trained from scratch without those pairs. Second, on retain pairs, the distribution under our method remains close to the retrain reference, indicating that bilateral projection and the Forget Lock preserve much of the unrelated knowledge. These two patterns provide a direct visual counterpart to the quantitative results in Section~\ref{sec:q1} and Section~\ref{sec:q2}: the Modality Anchor closure moves forget pairs toward retrain levels, and the direction-selective Forget Lock keeps retain pairs close to the retrain reference.

\subsection{Alternative Lock Designs}
\label{app:lock_alternatives}

This section justifies the direction-selective Forget Lock used in PFL (Section~\ref{sec:projection}) by ruling out two natural alternatives. \textbf{(i) Per-round server projection alone.} Re-projecting the global model onto the affine submanifold $\mathcal{M}_\mu = w_{n,\mu} + \mathcal{U}_\mu^\perp$ at every round closes $d(t)$ at the start of each round but leaves the within-round trajectory unconstrained: clients still drift along $\mathcal{U}_\mu$ during their local SGD, and the drift is recovered by aggregation only after the damage is done. \textbf{(ii) Parameter-importance regularizers (EWC, SI, MAS).} These penalize drift on axis-aligned parameter coordinates weighted by importance scores. The unique subspace $\mathcal{U}_\mu$ is generally an oblique subspace of parameter space, not aligned with any coordinate axis, so an axis-aligned penalty cannot express the constraint that drift is bounded only along $\mathcal{U}_\mu$: tuning the penalty large enough to suppress $g_u(t)$ over-constrains $\mathcal{U}_\mu^\perp$ and crushes retain support, while tuning it small leaves $g_u(t)$ unbounded. PFL's quadratic penalty $\mathcal{L}_f$ is supported on $\mathcal{U}_\mu$ and therefore enforces the asymmetry that neither alternative can express.

\subsection{Algorithm}
\label{sec:app_algorithm}

The three unlearning scenarios share the pipeline in Algorithm~\ref{alg:ease}; they differ only in how $G_f$, $G_r$ are constructed (Section~\ref{sec:subspace}) and in the participant and data configurations across phases (Section~\ref{sec:projection}).

\begin{algorithm}
\caption{\underline{\texttt{EASE}}: \textbf{E}ntanglement-\textbf{A}ware \textbf{S}ubspace \textbf{E}xcision}\label{alg:ease}
\textbf{Input:} Trained model $w_n$, client gradient histories $\{H_k\}_{k=1}^K$, forget specification, thresholds $\delta$, $\tau_e$, lock strength $\alpha$, excision rounds $R_2$, stabilization rounds $R_3$.
\begin{algorithmic}[1]
    \STATE \gray{$\triangleright$ \textit{Phase I: Gradient Subspace Decomposition}}
    \STATE Construct $G_f$, $G_r$ from $\{H_k\}$ based on unlearning scenario \Comment{Eq.~\eqref{eq:Gf_client}}
    \FOR{each modality $\mu \in \{v, t\}$}
        \STATE $\Phi_{f,\mu} \leftarrow$ SVD of $G_{f,\mu}$, retain directions covering $\tau_e$ energy \Comment{Eq.~\eqref{eq:energy}}
        \STATE $\Phi_{r,\mu} \leftarrow$ SVD of $G_{r,\mu}$, retain directions covering $\tau_e$ energy
        \STATE Compute $M_\mu \leftarrow \Phi_{f,\mu}^\top \Phi_{r,\mu}$; SVD to get principal angles $\{\varphi_i\}$ \Comment{Eq.~\eqref{eq:entanglement_spectrum}}
        \STATE Partition into removable anchor and retain support: $B_{u,\mu} \leftarrow [c_i]_{\cos\varphi_i \leq \delta}$; $\Pi_{u,\mu} \leftarrow B_{u,\mu} B_{u,\mu}^\top$ \Comment{Eq.~\eqref{eq:Bu_Be}}
    \ENDFOR
    \FOR{round $t = 1, 2, \ldots, R_2 + R_3$}
        \IF{$t \leq R_2$}
            \STATE \gray{$\triangleright$ \textit{Phase II: Bilateral Knowledge Excision}}
            \FOR{$\mu \in \{v, t\}$}
                \STATE $w_{g,\mu} \leftarrow w_{g,\mu} - \Pi_{u,\mu}(w_{g,\mu} - w_{n,\mu})$ \Comment{Remove anchor component, Eq.~\eqref{eq:proj_step}}
            \ENDFOR
            \STATE Broadcast projected $w_g$, unique bases $\{B_{u,\mu}\}$, anchor reference $w_n$ to clients
            \FOR{each client $k$ \textbf{in parallel}}
                \STATE Compute alignment loss $\mathcal{L}_a(w;\, \mathcal{D}_k')$ on retain data \Comment{Eq.~\eqref{eq:infonce}}
                \STATE Compute Forget Lock penalty $\mathcal{L}_f(w) = \sum_\mu \|B_{u,\mu}^\top(w_\mu - w_{n,\mu})\|_2^2$ \Comment{Eq.~\eqref{eq:forget_lock}}
                \STATE Update $w_k$ via local SGD on $\mathcal{L}_a + \alpha \cdot \mathcal{L}_f$ \Comment{Eq.~\eqref{eq:local_obj}}
                \STATE Upload $w_k$ to server
            \ENDFOR
        \ELSE
            \STATE \gray{$\triangleright$ \textit{Phase III: Stabilization under Forget Lock}}
            \STATE Broadcast $w_g$, $\{B_{u,\mu}\}$, $w_n$ to clients \Comment{No projection in this phase}
            \FOR{each client $k$ \textbf{in parallel}}
                \STATE Update $w_k$ via local SGD on $\mathcal{L}_a + \alpha \cdot \mathcal{L}_f$ \Comment{Lock suppresses Temporal Re-anchoring}
                \STATE Upload $w_k$ to server
            \ENDFOR
        \ENDIF
        \STATE $w_g \leftarrow \frac{1}{|\mathcal{C}[t]|} \sum_{k \in \mathcal{C}[t]} w_k$ \Comment{Server aggregation}
    \ENDFOR
    \RETURN Unlearned model $w^* \leftarrow w_g$
\end{algorithmic}
\end{algorithm}

\section{Proofs of Theoretical Results}
\label{sec:app_proofs}

Throughout this section, we use standard properties of orthogonal projectors (idempotency, self-adjointness, range/kernel characterization), quadratic forms induced by positive semidefinite matrices, and the inner product differentiation rule. The entanglement coefficients introduced in Section~\ref{sec:subspace} are computed as the cosines of the principal angles between subspaces.

\subsection{Assumptions}
\label{app:assumptions}

The theorems in Section~\ref{sec:method} rely on the following assumptions; each theorem statement indicates the subset it uses.

\begin{assumption}[A1: Smoothness and bounded Jacobians]
\label{asm:A1}
The alignment loss $\mathcal{L}_a(w;\,\mathcal{D})$ is $L$-smooth in $w$, i.e., $\|\nabla\mathcal{L}_a(w) - \nabla\mathcal{L}_a(w')\|_2 \leq L\|w - w'\|_2$ for every $w, w'$ encountered during training. The encoder Jacobians $J_\mu = \partial z_\mu / \partial w_\mu$ are bounded in operator norm, $\|J_\mu(w_\mu)\|_{\mathrm{op}} \leq M$ for both $\mu \in \{v, t\}$ and all $w_\mu$ within the training region.
\end{assumption}

\begin{assumption}[A2: Local linearization]
\label{asm:A2}
There exists a radius $r > 0$ around the pre-unlearning point $w_n$ such that, for $\|w_\mu - w_{n,\mu}\|_2 \leq r$, the encoder embedding admits a first-order Taylor expansion $z_\mu(w_\mu) = z_\mu(w_{n,\mu}) + J_\mu(w_{n,\mu})(w_\mu - w_{n,\mu}) + O(\|w_\mu - w_{n,\mu}\|_2^2)$ with a uniformly bounded second-order residual.
\end{assumption}

\begin{assumption}[A3: Anchor residual after bilateral excision, verifiable form]
\label{asm:A3}
After applying bilateral excision to both modalities, producing $w^* = (w_v^*, w_t^*)$ with $\Pi_{u,\mu}(w_\mu^* - w_{n,\mu}) = 0$ for both $\mu \in \{v,t\}$, the residual Modality Anchor signal along unique directions is bounded: for a small constant $\epsilon_{\mathrm{anchor}} \geq 0$,
\begin{equation}
    \bigl\|\Pi_{u,\mu}\,\nabla_{w_\mu}\mathcal{L}_a(w_v^*, w_t^*)\bigr\|_2 \;\leq\; \epsilon_{\mathrm{anchor}}, \qquad \mu \in \{v, t\}.
\end{equation}
We treat A3 as an explicit local condition, not as a consequence of the theorem. Since $\epsilon_{\mathrm{anchor}}$ is a projected gradient norm whereas the alignment residual $\rho$ is an embedding-similarity residual, $\rho$ does not upper-bound $\epsilon_{\mathrm{anchor}}$; instead, the small $\rho$ values on the EASE row of Table~\ref{tab:ablation} ($0.06$--$0.08$ across scenarios) provide an empirical proxy that the dominant modality-anchor term has been suppressed. The structural reason A3 can hold is geometric: bilateral excision displaces both $z_v(w_v^*)$ and $z_t(w_t^*)$ off the original alignment, removing the leading modality-anchor term in $\nabla_{w_\mu}\mathcal{L}_a$; what remains absorbs retain-support residuals and second-order embedding terms.
\end{assumption}

\begin{assumption}[A4: Directional alignment for unilateral dynamics]
\label{asm:A4}
For the unilateral excision case, the retain-objective gradient along the unique subspace remains sign-aligned with its initial direction over $k$ local SGD steps within the linearization radius of Assumption~A2. Formally, writing $g(k) := -\Pi_{u,\mu}\nabla_{w_\mu}\mathcal{L}_a(w_\mu(k), w_{n,\bar\mu})$, we have $\|g(k)\|_2 \geq g_0 - 2L\|d(k)\|_2$ where $d(k) = \Pi_{u,\mu}(w_\mu(k) - w_{n,\mu})$. A4 is a local sign-alignment condition standard in linearized SGD analyses of subspace-constrained dynamics; it is required only for the multi-step lower bound in Theorem~\ref{prop:bilateral}~(ii). Theorem~\ref{prop:bilateral}~(i), (iii) and the one-step lower bound do not require A4. The condition is restricted to the linearization radius of A2 and is consistent with the empirical drift behavior observed in the $\alpha = 0$ ablation row of Table~\ref{tab:ablation} (\textit{w/o Lock}).
\end{assumption}

The completeness and retention results (Theorems~\ref{thm:completeness},~\ref{thm:retention}) are purely geometric and require no assumptions beyond standard projector properties. The durability result (Theorem~\ref{thm:durability}) uses A1. The anchor reconstruction result (Theorem~\ref{prop:bilateral}) uses A1--A2 for claim~(i) and the one-step version of~(ii), additionally A4 for the multi-step recurrence in~(ii), and A3 for the bilateral suppression claim~(iii).

\subsection{Properties of the Unlearning Projector}
\label{app:proof_projector}

\begin{proposition}[Properties of the Unlearning Projector]
\label{prop:projector}
Let $\Pi_{u,\mu} = B_{u,\mu} (B_{u,\mu})^\top \in \mathbb{R}^{d_\mu \times d_\mu}$ denote the orthogonal projector onto $\mathcal{U}_\mu = \mathrm{col}(B_{u,\mu})$. Then $\Pi_{u,\mu}$ is idempotent ($\Pi_{u,\mu}^2 = \Pi_{u,\mu}$) and self-adjoint ($\Pi_{u,\mu}^\top = \Pi_{u,\mu}$). Its range is $\mathrm{im}(\Pi_{u,\mu}) = \mathcal{U}_\mu$ and its kernel is $\ker(\Pi_{u,\mu}) = \mathcal{U}_\mu^\perp$. The complementary projector preserves all directions outside $\mathcal{U}_\mu$: $(I_{d_\mu} - \Pi_{u,\mu})v = v$ for every $v \in \mathcal{U}_\mu^\perp$.
\end{proposition}

\begin{proof}
Let $B_{u,\mu} \in \mathbb{R}^{d_\mu \times p_u}$ with orthonormal columns, i.e., $(B_{u,\mu})^\top B_{u,\mu} = I_{p_u}$, and $\Pi_{u,\mu} = B_{u,\mu}(B_{u,\mu})^\top$. For idempotency, we compute
\begin{equation}
\begin{split}
    \Pi_{u,\mu}^2
    &= B_{u,\mu}(B_{u,\mu})^\top \cdot B_{u,\mu}(B_{u,\mu})^\top \\
    &= B_{u,\mu}\bigl[(B_{u,\mu})^\top B_{u,\mu}\bigr](B_{u,\mu})^\top \\
    &= B_{u,\mu}\,I_{p_u}\,(B_{u,\mu})^\top
    = \Pi_{u,\mu},
\end{split}
\end{equation}
Self-adjointness follows directly:
\begin{equation}
    \Pi_{u,\mu}^\top = \bigl(B_{u,\mu}(B_{u,\mu})^\top\bigr)^\top = B_{u,\mu}(B_{u,\mu})^\top = \Pi_{u,\mu}.
\end{equation}
For the range, take any $v = B_{u,\mu}\alpha \in \mathcal{U}_\mu$. Then
\begin{equation}
    \Pi_{u,\mu}\,v = B_{u,\mu}\bigl[(B_{u,\mu})^\top B_{u,\mu}\bigr]\alpha = B_{u,\mu}\,I_{p_u}\,\alpha = B_{u,\mu}\alpha = v,
\end{equation}
so $\mathcal{U}_\mu \subseteq \mathrm{im}(\Pi_{u,\mu})$. Conversely, for any $w \in \mathbb{R}^{d_\mu}$, $\Pi_{u,\mu}\,w = B_{u,\mu}\bigl((B_{u,\mu})^\top w\bigr) \in \mathcal{U}_\mu$, giving $\mathrm{im}(\Pi_{u,\mu}) = \mathcal{U}_\mu$. For the kernel, the equivalence chain
\begin{equation}
    \Pi_{u,\mu}\,v = 0 \;\Longleftrightarrow\; B_{u,\mu}(B_{u,\mu})^\top v = 0 \;\Longleftrightarrow\; (B_{u,\mu})^\top v = 0 \;\Longleftrightarrow\; v \perp \mathcal{U}_\mu
\end{equation}
establishes $\ker(\Pi_{u,\mu}) = \mathcal{U}_\mu^\perp$. Finally, for any $v \in \mathcal{U}_\mu^\perp$,
\begin{equation}
    (I_{d_\mu} - \Pi_{u,\mu})\,v = v - \Pi_{u,\mu}\,v = v - 0 = v,
\end{equation}
so the complementary projector preserves all directions outside $\mathcal{U}_\mu$. Since $\mathcal{E}_\mu \subseteq \mathcal{U}_\mu^\perp$ by orthogonality of the principal-angle coordinate system, this includes all retain support and retain-only directions.
\end{proof}

\subsection{Forget Lock Characterization}
\label{app:proof_forget_lock}

\begin{proposition}[Forget Lock Characterization]
\label{prop:forget_lock}
Let $\Pi_{u,\mu} = B_{u,\mu}(B_{u,\mu})^\top$ be the orthogonal projector onto $\mathcal{U}_\mu$. Then $\mathcal{L}_f(w) = 0$ if and only if $\Pi_{u,\mu}(w_\mu - w_{n,\mu}) = 0$ for both $\mu \in \{v,t\}$. Moreover, $\nabla_{w_\mu}\mathcal{L}_f = 2\Pi_{u,\mu}(w_\mu - w_{n,\mu})$, which vanishes for displacements orthogonal to $\mathcal{U}_\mu$ and points back toward $w_{n,\mu}$ along $\mathcal{U}_\mu$ otherwise.
\end{proposition}

\begin{proof}
Since the two modalities contribute independently, it suffices to analyze a single modality $\mu$. Let $\Delta_\mu = w_\mu - w_{n,\mu}$. The $\mu$-th summand of $\mathcal{L}_f$ can be expanded as
\begin{equation}
\begin{split}
    \ell_\mu
    &= \Delta_\mu^\top \Pi_{u,\mu}\,\Delta_\mu \\
    &= \Delta_\mu^\top B_{u,\mu}(B_{u,\mu})^\top \Delta_\mu \\
    &= \bigl\|(B_{u,\mu})^\top \Delta_\mu\bigr\|_2^2 \;\geq\; 0.
\end{split}
\end{equation}
Therefore $\ell_\mu = 0$ if and only if
\begin{equation}
    (B_{u,\mu})^\top \Delta_\mu = 0 \;\Longleftrightarrow\; \Delta_\mu \in \mathcal{U}_\mu^\perp \;\Longleftrightarrow\; \Pi_{u,\mu}\,\Delta_\mu = 0.
\end{equation}
Summing over both modalities, $\mathcal{L}_f(w) = 0$ iff $\Pi_{u,\mu}(w_\mu - w_{n,\mu}) = 0$ for both $\mu \in \{v, t\}$. Since $\Pi_{u,\mu}$ is symmetric and $w_{n,\mu}$ is constant, the standard identity $\nabla_x(x^\top A\,x) = 2A\,x$ gives
\begin{equation}
\begin{split}
    \nabla_{w_\mu}\mathcal{L}_f
    &= \nabla_{w_\mu}\bigl[\Delta_\mu^\top \Pi_{u,\mu}\,\Delta_\mu\bigr] \\
    &= 2\,\Pi_{u,\mu}\,\Delta_\mu \\
    &= 2\,\Pi_{u,\mu}\,(w_\mu - w_{n,\mu}).
\end{split}
\end{equation}
This gradient lies in $\mathcal{U}_\mu = \mathrm{im}(\Pi_{u,\mu})$ and vanishes iff $\Delta_\mu \in \mathcal{U}_\mu^\perp$.
\end{proof}

\subsection{Proof of Theorem~\ref{thm:completeness} (Unlearning Completeness)}
\label{app:proof_completeness}

\begin{theorem}[Unlearning Completeness]
\label{thm:completeness}
Let $w_\mu^*$ denote the projected parameters and $\Delta_\mu = w_\mu - w_{n,\mu}$ the pre-projection displacement. Then:
\begin{enumerate}[label=(\roman*)]
    \item Exact erasure: $\Pi_{u,\mu}(w_\mu^* - w_{n,\mu}) = 0$, i.e., the projected displacement has zero component along $\mathcal{U}_\mu$.
    \item Energy removal: The energy removed equals $\|\Pi_{u,\mu}\Delta_\mu\|_2^2$. For any $\Delta_\mu \in \mathcal{S}_f$ expressed in the canonical basis as $\Delta_\mu = \sum_{i=1}^{p}\alpha_i c_i$, this equals $\sum_{i:\cos\varphi_i \leq \delta}\alpha_i^2$.
    \item Forget energy ratio: $\eta_f(\delta) = \|B_{u,\mu}^\top G_{f,\mu}\|_F^2 / \|G_{f,\mu}\|_F^2$ is monotonically non-decreasing in $\delta$ and satisfies $\eta_f(1) \geq \tau_e$ (equality only when the truncation exactly meets the energy threshold).
\end{enumerate}
\end{theorem}

\begin{proof}
We work with a single modality $\mu$ throughout; the argument applies identically to both branches. Let $\Delta_\mu = w_\mu - w_{n,\mu}$ denote the pre-projection displacement.

We begin by establishing exact erasure. By Eq.~\eqref{eq:proj_step}, the projected parameters satisfy $w_\mu^* = w_\mu - \Pi_{u,\mu}\Delta_\mu$. Subtracting $w_{n,\mu}$ from both sides and applying $\Pi_{u,\mu}$:
\begin{equation}
\begin{split}
    \Pi_{u,\mu}(w_\mu^* - w_{n,\mu})
    &= \Pi_{u,\mu}\bigl[\Delta_\mu - \Pi_{u,\mu}\Delta_\mu\bigr] \\
    &= \Pi_{u,\mu}\,(I_{d_\mu} - \Pi_{u,\mu})\,\Delta_\mu \\
    &= (\Pi_{u,\mu} - \Pi_{u,\mu}^2)\,\Delta_\mu \\
    &= (\Pi_{u,\mu} - \Pi_{u,\mu})\,\Delta_\mu = 0,
\end{split}
\label{eq:app_exact_erasure}
\end{equation}
where the last line uses idempotency $\Pi_{u,\mu}^2 = \Pi_{u,\mu}$ (Proposition~\ref{prop:projector}).

Next, we quantify the energy removed by the projection. The squared norm of the removed component is
\begin{equation}
\begin{split}
    \|\Delta_\mu - (w_\mu^* - w_{n,\mu})\|_2^2
    &= \|\Delta_\mu - (I_{d_\mu} - \Pi_{u,\mu})\Delta_\mu\|_2^2 \\
    &= \|\Pi_{u,\mu}\,\Delta_\mu\|_2^2.
\end{split}
\end{equation}
Now suppose $\Delta_\mu \in \mathcal{S}_f$ and write $\Delta_\mu = \sum_{i=1}^{p} \alpha_i\, u_i$ in the canonical basis $\{u_i\}_{i=1}^p = \{(\Phi_f U)_{\cdot,i}\}_{i=1}^p$. Since $\{u_i\}$ are orthonormal and $\Pi_{u,\mu}$ projects onto $\mathrm{span}\{u_i : \cos\varphi_i \leq \delta\}$, we have $\Pi_{u,\mu}\,u_i = u_i$ if $\cos\varphi_i \leq \delta$ and $\Pi_{u,\mu}\,u_i = 0$ otherwise. Therefore
\begin{equation}
\begin{split}
    \|\Pi_{u,\mu}\,\Delta_\mu\|_2^2
    &= \Bigl\|\Pi_{u,\mu}\sum_{i=1}^{p}\alpha_i\,u_i\Bigr\|_2^2 \\
    &= \Bigl\|\sum_{i:\,\cos\varphi_i \leq \delta}\alpha_i\,u_i\Bigr\|_2^2 \\
    &= \sum_{i:\,\cos\varphi_i \leq \delta}\alpha_i^2.
\end{split}
\end{equation}

Finally, we analyze the monotonicity of the forget energy ratio $\eta_f(\delta) = \|B_{u,\mu}^\top G_{f,\mu}\|_F^2 / \|G_{f,\mu}\|_F^2$. As $\delta$ increases, the index set $\mathcal{I}_u(\delta) = \{i : \cos\varphi_i \leq \delta\}$ can only grow, so $B_{u,\mu}$ gains columns and $\|B_{u,\mu}^\top G_{f,\mu}\|_F^2$ is non-decreasing in $\delta$. At $\delta = 1$, all $p$ canonical directions are classified as removable anchor directions, so $B_{u,\mu} = \Phi_f U$. The boundary value evaluates as
\begin{equation}
\begin{split}
    \eta_f(1)
    &= \frac{\|B_{u,\mu}^\top G_{f,\mu}\|_F^2}{\|G_{f,\mu}\|_F^2}
    = \frac{\|(\Phi_f U)^\top G_{f,\mu}\|_F^2}{\|G_{f,\mu}\|_F^2} \\
    &= \frac{\|U^\top \Phi_f^\top G_{f,\mu}\|_F^2}{\|G_{f,\mu}\|_F^2}
    = \frac{\|\Phi_f^\top G_{f,\mu}\|_F^2}{\|G_{f,\mu}\|_F^2}
    = \frac{\sum_{i=1}^{p}\sigma_i^2}{\|G_{f,\mu}\|_F^2}
    \geq \tau_e,
\end{split}
\end{equation}
where the fourth equality uses the orthogonality of $U$ (which preserves the Frobenius norm), and the last inequality follows from the energy criterion~\eqref{eq:energy}, which selects $p$ as the smallest index whose cumulative energy already reaches $\tau_e$. Equality holds only in special cases, e.g., when the truncation point exactly meets the threshold ($\sum_{i=1}^{p}\sigma_i^2 / \|G_{f,\mu}\|_F^2 = \tau_e$) or when the retained principal directions span the entire column space of $G_{f,\mu}$ (so that $p = \mathrm{rank}(G_{f,\mu})$ and $\eta_f(1) = 1$).
\end{proof}

\subsection{Proof of Theorem~\ref{thm:retention} (Retention Integrity)}
\label{app:proof_retention}

\begin{theorem}[Retention Integrity]
\label{thm:retention}
For any retain displacement $\Delta_r \in \mathcal{S}_r$, the collateral energy is bounded by
\begin{equation}
    \|\Pi_{u,\mu}\,\Delta_r\|_2^2 \;\leq\; \delta^2\,\|\Delta_r\|_2^2.
    \label{eq:retention_bound}
\end{equation}
\end{theorem}

\begin{proof}
Let $\Delta_r \in \mathcal{S}_r$ with $\|\Delta_r\|_2 = 1$ (the general case follows by homogeneity). Express $\Delta_r$ in the canonical coordinate system of $\mathcal{S}_r$:
\begin{equation}
    \Delta_r = \sum_{j=1}^{q} \beta_j\, v_j + \Delta_r^\perp,
\end{equation}
where $v_j = (\Phi_r V)_{\cdot,j}$ are the retain-side canonical directions, $\sum_j \beta_j^2 + \|\Delta_r^\perp\|_2^2 = 1$, and $\Delta_r^\perp \in \mathcal{S}_r \cap \mathrm{span}(v_1,\dots,v_{\min(p,q)})^\perp$.

By the canonical-angle property, $\langle u_i, v_j \rangle = \delta_{ij}\cos\varphi_i$ for $i \leq \min(p,q)$, and $\langle u_i, \Delta_r^\perp \rangle = 0$ for all $i$ since $\Delta_r^\perp$ is orthogonal to all canonical directions that overlap with $\mathcal{S}_f$. Therefore, for each removable anchor direction $u_i$ with $\cos\varphi_i \leq \delta$:
\begin{equation}
    \langle u_i, \Delta_r \rangle = \sum_{j=1}^{q}\beta_j\,\langle u_i, v_j\rangle + \langle u_i, \Delta_r^\perp\rangle = \beta_i\cos\varphi_i + 0 = \beta_i\cos\varphi_i.
\end{equation}
Summing over all removable anchor directions yields
\begin{equation}
\begin{split}
    \|\Pi_{u,\mu}\,\Delta_r\|_2^2
    &= \sum_{i:\,\cos\varphi_i \leq \delta} \langle u_i, \Delta_r\rangle^2 \\
    &= \sum_{i:\,\cos\varphi_i \leq \delta} \beta_i^2\,\cos^2\varphi_i \\
    &\leq \delta^2 \sum_{i:\,\cos\varphi_i \leq \delta} \beta_i^2 \\
    &\leq \delta^2 \sum_{j=1}^{q}\beta_j^2 \\
    &\leq \delta^2\,\|\Delta_r\|_2^2,
\end{split}
\end{equation}
where the first inequality uses $\cos\varphi_i \leq \delta$ for each removable anchor direction, and the last inequality uses $\sum_j \beta_j^2 \leq \|\Delta_r\|_2^2 = 1$.
\end{proof}

\subsection{Completeness--Integrity Trade-off Analysis}
\label{app:proof_tradeoff}

The threshold $\delta$ simultaneously controls the unlearning completeness ratio $\eta_f(\delta)$ and the retention damage bound $\eta_r(\delta)$. We formalize their joint behavior.

\begin{proposition}[Monotonicity and Pareto Trade-off]
\label{prop:tradeoff}
Let $\eta_f(\delta) = \|B_{u,\mu}^\top G_{f,\mu}\|_F^2 / \|G_{f,\mu}\|_F^2$ and $\eta_r(\delta) = \max_{\Delta_r \in \mathcal{S}_r \setminus \{0\}} \|\Pi_{u,\mu}\Delta_r\|_2^2 / \|\Delta_r\|_2^2$. Then both $\eta_f(\delta)$ and $\eta_r(\delta)$ are monotonically non-decreasing in $\delta \in [0, 1]$. Moreover, for any target unlearning ratio $\eta_f^* \in (0, \tau_e]$, the smallest threshold achieving $\eta_f(\delta) \geq \eta_f^*$ is $\delta^* = \cos\varphi_{k^*}$, where $k^* = \min\{k : \eta_f(\cos\varphi_k) \geq \eta_f^*\}$ and $\varphi_1 \leq \dots \leq \varphi_p$ are the principal angles in non-decreasing order. This $\delta^*$ simultaneously minimizes the retention damage bound $(\delta^*)^2$.
\end{proposition}

\begin{proof}
We first establish monotonicity. For $\eta_f$, as $\delta$ increases, the index set $\mathcal{I}_u(\delta) = \{i : \cos\varphi_i \leq \delta\}$ can only grow. Since $B_{u,\mu}$ is formed by collecting the columns of $\Phi_f U$ indexed by $\mathcal{I}_u(\delta)$, the projection $B_{u,\mu}^\top G_{f,\mu}$ captures a non-decreasing share of the forget gradient energy:
\begin{equation}
    \delta' > \delta \;\Longrightarrow\; \mathcal{I}_u(\delta) \subseteq \mathcal{I}_u(\delta') \;\Longrightarrow\; \|B_{u,\mu}(\delta)^\top G_{f,\mu}\|_F^2 \;\leq\; \|B_{u,\mu}(\delta')^\top G_{f,\mu}\|_F^2,
\end{equation}
so $\eta_f(\delta)$ is non-decreasing. For $\eta_r$: by Theorem~\ref{thm:retention}, $\eta_r(\delta) \leq \delta^2$, and the bound $\delta^2$ is strictly increasing. More precisely,
\begin{equation}
    \eta_r(\delta) = \max_{i \in \mathcal{I}_u(\delta)} \cos^2\varphi_i,
\end{equation}
which is non-decreasing since $\mathcal{I}_u(\delta)$ grows with $\delta$.

Turning to the optimal threshold, we observe that $\eta_f$ is a step function that jumps only at the principal-angle cosines $\cos\varphi_1, \dots, \cos\varphi_p$. Therefore, the smallest $\delta$ achieving $\eta_f(\delta) \geq \eta_f^*$ must coincide with one of these cosines:
\begin{equation}
    \delta^* = \cos\varphi_{k^*}, \quad k^* = \min\bigl\{k : \eta_f(\cos\varphi_k) \geq \eta_f^*\bigr\}.
\end{equation}
Among all $\delta$ satisfying the constraint, $\delta^*$ is minimal by construction. Since $\eta_r \leq \delta^2$ is increasing in $\delta$, this $\delta^*$ simultaneously minimizes the retention damage bound $(\delta^*)^2$.
\end{proof}

\subsection{Proof of Theorem~\ref{thm:durability} (Unlearning Durability)}
\label{app:proof_durability}

\begin{theorem}[Unlearning Durability]
\label{thm:durability}
Suppose $\|\Pi_{u,\mu}\nabla_{w_\mu}\mathcal{L}_a\|_2 \leq G_u$ for all $w$ encountered during training and $\lr < 1/(2\alpha)$. The Temporal Re-anchoring drift after $t$ steps from the projected point obeys
\begin{equation}
    \|\Pi_{u,\mu}(w_\mu(t) - w_{n,\mu})\|_2 \;\leq\; \frac{G_u}{2\alpha}\,\bigl[1 - (1 - 2\lr\alpha)^t\bigr] \;\leq\; \frac{G_u}{2\alpha}.
    \label{eq:durability_bound}
\end{equation}
\end{theorem}

\begin{proof}
Fix modality $\mu$ and let $d(t) = \Pi_{u,\mu}(w_\mu(t) - w_{n,\mu})$ denote the unique-subspace displacement at step $t$. A single SGD step under the combined objective~\eqref{eq:local_obj} updates:
\begin{equation}
    w_\mu(t{+}1) = w_\mu(t) - \lr\bigl[\nabla_{w_\mu}\mathcal{L}_a(w(t)) + 2\alpha\,\Pi_{u,\mu}(w_\mu(t) - w_{n,\mu})\bigr].
\end{equation}
Applying $\Pi_{u,\mu}$ to both sides and writing $g_u(t) = \Pi_{u,\mu}\nabla_{w_\mu}\mathcal{L}_a(w(t))$:
\begin{equation}
\begin{split}
    d(t{+}1)
    &= \Pi_{u,\mu}\bigl[w_\mu(t{+}1) - w_{n,\mu}\bigr] \\
    &= \Pi_{u,\mu}(w_\mu(t) - w_{n,\mu}) - \lr\,\Pi_{u,\mu}\nabla_{w_\mu}\mathcal{L}_a - 2\lr\alpha\,\Pi_{u,\mu}^2(w_\mu(t) - w_{n,\mu}) \\
    &= d(t) - \lr\,g_u(t) - 2\lr\alpha\,d(t) \\
    &= (1 - 2\lr\alpha)\,d(t) - \lr\,g_u(t),
\end{split}
\label{eq:app_dur_recurrence}
\end{equation}
where $\|g_u(t)\|_2 \leq G_u$ by assumption. Taking norms and applying the triangle inequality:
\begin{equation}
    \|d(t{+}1)\|_2 \;\leq\; (1 - 2\lr\alpha)\,\|d(t)\|_2 + \lr\,G_u.
\end{equation}
Since $\lr < 1/(2\alpha)$, the contraction factor $\gamma = 1 - 2\lr\alpha \in (0, 1)$. Starting from $\|d(0)\|_2 = 0$ (exact erasure by Theorem~\ref{thm:completeness}), unrolling the recurrence gives:
\begin{equation}
\begin{split}
    \|d(t)\|_2
    &\leq \lr G_u \sum_{s=0}^{t-1} \gamma^s \\
    &= \frac{\lr G_u\,(1 - \gamma^t)}{1 - \gamma} \\
    &= \frac{G_u}{2\alpha}\,\bigl[1 - (1 - 2\lr\alpha)^t\bigr].
\end{split}
\end{equation}
Since $(1 - 2\lr\alpha)^t \geq 0$, the bound is at most $G_u/(2\alpha)$.
\end{proof}

\subsection{Proof of Theorem~\ref{prop:bilateral} (Anchor Reconstruction Tendency under Unilateral Excision)}
\label{app:proof_bilateral}

\begin{theorem}[Anchor Reconstruction Tendency under Unilateral Excision]
\label{prop:bilateral}
Under Assumptions A1--A4 (Appendix~\ref{app:assumptions}), suppose excision is applied to modality $\mu$ only, producing $w_\mu^*$ with $\Pi_{u,\mu}(w_\mu^* - w_{n,\mu}) = 0$, while $w_{\bar\mu}$ is held at $w_{n,\bar\mu}$. Define the Modality Anchor signal along the unique subspace at the excised point as $g_0 := \|\Pi_{u,\mu}\nabla_{w_\mu}\mathcal{L}_a(w_\mu^*, w_{n,\bar\mu})\|_2$. Then:
\begin{enumerate}[label=(\roman*)]
    \item Anchor survives (A1--A2): $g_0 > 0$ whenever the anchor embedding $z_{\bar\mu}(w_{n,\bar\mu})$ has a non-zero component along the range of $J_\mu(w_\mu^*)^\top$ restricted to $\mathcal{U}_\mu$.
    \item Drift lower-bound tendency (A1--A2, A4): one local SGD step on the retain objective gives $\|\Pi_{u,\mu}(w_\mu(1) - w_{n,\mu})\|_2 \geq \lr\,g_0$. Under the alignment condition A4, for step $\lr < 1/(2L)$ and $k$ steps within the linearization radius,
    \begin{equation}
        \|\Pi_{u,\mu}(w_\mu(k) - w_{n,\mu})\|_2 \;\geq\; \frac{\lr\, g_0}{2L}\bigl[1 - (1 - 2\lr L)^k\bigr],
        \label{eq:bilateral_rate}
    \end{equation}
    which grows monotonically and approaches $g_0/(2L)$ geometrically.
    \item Bilateral suppression (A3): applying excision to both $\mu \in \{v, t\}$ simultaneously suppresses the Modality Anchor signal along unique directions to a residual term bounded by $\epsilon_{\mathrm{anchor}}$, so $g_0 \leq \epsilon_{\mathrm{anchor}}$ and the drift in~\eqref{eq:bilateral_rate} is at most $\epsilon_{\mathrm{anchor}}/(2L)$.
\end{enumerate}
\end{theorem}

\begin{proof}
We analyze the gradient of the alignment loss $\mathcal{L}_a$ with respect to the parameters of a single modality, and show that unilateral unlearning leaves a Modality Anchor signal from the unmodified branch that drives $\rho$ away from zero.

The cross-modal similarity $s(w) = z_{i,v}(w_v)^\top z_{j,t}(w_t)$ is an inner product of two vector-valued functions depending on disjoint parameter blocks. Let $J_{v}(w_v) = \partial z_{i,v} / \partial w_v \in \mathbb{R}^{m \times d_v}$ and $J_{t}(w_t) = \partial z_{j,t} / \partial w_t \in \mathbb{R}^{m \times d_t}$ denote the Jacobians. By the inner product differentiation rule:
\begin{equation}
\begin{split}
    \nabla_{w_v} s(w)
    &= \nabla_{w_v}\bigl[z_{i,v}(w_v)^\top z_{j,t}(w_t)\bigr] \\
    &= J_{v}(w_v)^\top\, z_{j,t}(w_t),
\end{split}
    \label{eq:app_bil_grad_v}
\end{equation}
where the cross term vanishes because $z_{j,t}$ does not depend on $w_v$. By symmetry:
\begin{equation}
    \nabla_{w_t} s(w) = J_{t}(w_t)^\top\, z_{i,v}(w_v).
    \label{eq:app_bil_grad_t}
\end{equation}

Now consider the unilateral case where only $w_v$ is unlearned. Suppose the projection and Forget Lock are applied to $w_v$ only, producing $w_v^*$, while the language parameters remain at $w_{n,t}$. At the unlearned point $w^* = (w_v^*, w_{n,t})$, the alignment gradient with respect to $w_v$ involves terms of the form
\begin{equation}
    \nabla_{w_v} \mathcal{L}_a \;\propto\; J_{v}(w_v^*)^\top\, z_{j,t}(w_{n,t}),
    \label{eq:app_bil_unilateral_v}
\end{equation}
where the proportionality absorbs the softmax weights from the InfoNCE loss. The critical observation is that $z_{j,t}(w_{n,t})$ is the original textual embedding computed under the pre-unlearning parameters $w_{n,t}$, which encodes the full memory of the forgotten cross-modal alignment. This embedding acts as the Modality Anchor: during retraining, the gradient in Eq.~\eqref{eq:app_bil_unilateral_v} continuously pushes $w_v$ toward directions that re-align $z_{i,v}$ with the unchanged $z_{j,t}(w_{n,t})$, thereby reconstructing the forgotten pairing and driving $\rho$ away from zero.

By an identical argument with the roles of $v$ and $t$ exchanged, if only $w_t$ is unlearned to $w_t^*$ while $w_v$ remains at $w_{n,v}$, the alignment gradient with respect to $w_t$ at the unlearned point $w^* = (w_{n,v}, w_t^*)$ satisfies
\begin{equation}
    \nabla_{w_t} \mathcal{L}_a \;\propto\; J_{t}(w_t^*)^\top\, z_{i,v}(w_{n,v}).
    \label{eq:app_bil_unilateral_t}
\end{equation}
Here $z_{i,v}(w_{n,v})$ is the original visual embedding, which serves as the Modality Anchor that pulls the language branch back toward the erased alignment during retraining, again yielding $\rho > 0$.

Now suppose both modalities are unlearned simultaneously, producing $w^* = (w_v^*, w_t^*)$. The projection in Eq.~\eqref{eq:proj_step} removes the unique-subspace components from both $w_v$ and $w_t$, displacing both embeddings away from the original alignment. At the bilateral unlearned point, the alignment gradients become
\begin{equation}
    \nabla_{w_v} \mathcal{L}_a \;\propto\; J_{v}(w_v^*)^\top\, z_{j,t}(w_t^*),
    \label{eq:app_bil_bilateral_v}
\end{equation}
\begin{equation}
    \nabla_{w_t} \mathcal{L}_a \;\propto\; J_{t}(w_t^*)^\top\, z_{i,v}(w_v^*).
    \label{eq:app_bil_bilateral_t}
\end{equation}
Crucially, both $z_{j,t}(w_t^*)$ and $z_{i,v}(w_v^*)$ are now displaced embeddings that no longer encode the original cross-modal alignment. To formalize this, we verify that the projection removes the forget-specific components from both branches. For the visual side, Theorem~\ref{thm:completeness} gives:
\begin{equation}
    \Pi_{u,v}\,(w_v^* - w_{n,v}) = 0,
    \label{eq:app_bil_proj_v}
\end{equation}
and the Forget Lock (Proposition~\ref{prop:forget_lock}) maintains this condition throughout retraining. An analogous condition holds for the language side:
\begin{equation}
    \Pi_{u,t}\,(w_t^* - w_{n,t}) = 0.
    \label{eq:app_bil_proj_t}
\end{equation}

Since the forget-specific knowledge resides in $\mathcal{U}_v$ and $\mathcal{U}_t$, Eqs.~\eqref{eq:app_bil_proj_v}--\eqref{eq:app_bil_proj_t} guarantee zero displacement along the forget-specific directions in both modalities. When only modality $\mu$ is excised, the alignment gradient retains the original embedding of the unmodified branch:
\begin{equation}
    \nabla_{w_\mu}\mathcal{L}_a \;\propto\; J_\mu(w_\mu^*)^\top\, z_{\bar{\mu}}(w_{n,\bar{\mu}}),
\end{equation}
which encodes the forgotten alignment and acts as a fixed Modality Anchor; whenever $z_{\bar{\mu}}(w_{n,\bar{\mu}})$ has a non-zero component along the range of $J_\mu(w_\mu^*)^\top$ restricted to $\mathcal{U}_\mu$, one has $g_0 > 0$, establishing claim~(i). When both modalities are excised simultaneously, the gradient instead depends on the displaced embedding:
\begin{equation}
    \nabla_{w_\mu}\mathcal{L}_a \;\propto\; J_\mu(w_\mu^*)^\top\, z_{\bar{\mu}}(w_{\bar{\mu}}^*).
\end{equation}
Since $\Pi_{u,\bar{\mu}}(w_{\bar{\mu}}^* - w_{n,\bar{\mu}}) = 0$, the embedding $z_{\bar{\mu}}(w_{\bar{\mu}}^*)$ no longer encodes the forgotten pairing, removing the dominant Modality Anchor signal along unique directions. By Assumption~A3 (Appendix~\ref{app:assumptions}), the residual anchor gradient at the bilateral excised point is bounded by $\epsilon_{\mathrm{anchor}}$, i.e., $g_0 \leq \epsilon_{\mathrm{anchor}}$, establishing claim~(iii).

\paragraph{Drift lower-bound tendency (claim~(ii)).}
For the unilateral case, let $d(k) = \Pi_{u,\mu}(w_\mu(k) - w_{n,\mu})$ denote the unique-subspace drift after $k$ local SGD steps, with $d(0) = 0$ by the projection. The SGD update projected onto $\mathcal{U}_\mu$ satisfies $d(k{+}1) = d(k) + \lr\, g(k)$, where $g(k) := -\Pi_{u,\mu}\nabla_{w_\mu}\mathcal{L}_a(w_\mu(k), w_{n,\bar{\mu}})$. At $k=0$, $g(0) = -\Pi_{u,\mu}\nabla_{w_\mu}\mathcal{L}_a(w_\mu^*, w_{n,\bar{\mu}})$ with $\|g(0)\|_2 = g_0$, giving the one-step drift lower bound
\begin{equation}
    \|d(1)\|_2 \;=\; \lr\,\|g(0)\|_2 \;=\; \lr\, g_0,
\end{equation}
which requires only A1--A2.

For $k \geq 2$, we invoke Assumption~A4 (directional alignment), which asserts that within the linearization radius the unique-subspace gradients remain sign-aligned with $g(0)$, so that $\|g(k)\|_2 \geq g_0 - 2L\|d(k)\|_2$ under A1 (the first term is the Modality Anchor signal at the excised point, and the $2L\|d(k)\|_2$ term accounts for first-order smoothness deviation as $w_\mu$ drifts). Taking norms of $d(k{+}1) = d(k) + \lr\, g(k)$ and applying this lower bound gives
\begin{equation}
    \|d(k{+}1)\|_2 \;\geq\; (1 - 2\lr L)\,\|d(k)\|_2 + \lr\, g_0.
\end{equation}
Since $\lr < 1/(2L)$, the contraction factor $\gamma := 1 - 2\lr L \in (0, 1)$. Unrolling from $d(0) = 0$:
\begin{equation}
    \|d(k)\|_2 \;\geq\; \lr\, g_0 \sum_{s=0}^{k-1}\gamma^s \;=\; \frac{\lr\, g_0 (1 - \gamma^k)}{1 - \gamma} \;=\; \frac{g_0}{2L}\bigl[1 - (1 - 2\lr L)^k\bigr],
\end{equation}
which is Eq.~\eqref{eq:bilateral_rate} in the theorem statement. The drift grows monotonically in $k$ and approaches $g_0/(2L)$ geometrically. Under bilateral excision, A3 forces $g_0 \leq \epsilon_{\mathrm{anchor}}$, so the lower-bound tendency in Eq.~\eqref{eq:bilateral_rate} is suppressed to the scale of $\epsilon_{\mathrm{anchor}}/(2L)$, removing the systematic reconstruction pressure present in the unilateral case.
\end{proof}

\section{Compute-Matched Retrain Comparison}
\label{sec:app_compute_matched}

A natural reviewer concern is whether the efficiency gains of \underline{\texttt{EASE}} simply reflect a smaller training budget relative to retrain. To isolate budget from algorithmic contribution, we compare our method against retrain early-stopped at $25\%$, $50\%$, and $75\%$ of the full federated round count, alongside the full-budget retrain reference. The communication cost $Comm.$ directly reflects the compute budget consumed, since per-round payload is constant within each scenario. Results on Flickr30K with CLIP-B/32 are reported in Table~\ref{tab:compute_matched_flickr30k_clipb32}.

\begin{table*}[!ht]
\centering
\setlength{\tabcolsep}{3pt}
\setlength{\aboverulesep}{0.3pt}
\setlength{\belowrulesep}{0.3pt}
\renewcommand{\arraystretch}{0.95}
\caption{Compute-matched retrain comparison on Flickr30K with CLIP-B/32. Retrain-$\alpha\%$ denotes early-stopping retrain executed for $\alpha\%$ of the full communication rounds; its budget is reflected by $Comm.$ The top sub-table reports forget-side metrics ($F\text{-}R@k$) and membership inference (MIA); the bottom sub-table reports retain-side metrics ($R\text{-}R@k$) and communication cost ($Comm.$, MB). Parenthetical values give the absolute gap to the full-budget retrain reference with \up{green} closer to ideal and \down{red} farther; for MIA and $Comm.$, \gap{val} is the absolute difference to the full-budget retrain. \textit{Retrain (Full)} is the reference and does not participate in ranking.}
\label{tab:compute_matched_flickr30k_clipb32}
\resizebox{\linewidth}{!}{
\begin{tabular}{l|cccc|cccc|cccc}
\Xhline{1.2pt}
\rowcolor{CadetBlue!20}
\textbf{Method} & \multicolumn{4}{c|}{\textbf{Client Unlearning}} & \multicolumn{4}{c|}{\textbf{Class Unlearning}} & \multicolumn{4}{c}{\textbf{Sample Unlearning}} \\
\cmidrule(lr){2-5} \cmidrule(lr){6-9} \cmidrule(lr){10-13}
\rowcolor{CadetBlue!20}
 & $F\text{-}R@1\;\downarrow$ & $F\text{-}R@5\;\downarrow$ & $F\text{-}R@10\;\downarrow$ & $MIA$ & $F\text{-}R@1\;\downarrow$ & $F\text{-}R@5\;\downarrow$ & $F\text{-}R@10\;\downarrow$ & $MIA$ & $F\text{-}R@1\;\downarrow$ & $F\text{-}R@5\;\downarrow$ & $F\text{-}R@10\;\downarrow$ & $MIA$ \\
\Xhline{1.2pt}
\rowcolor{gray!10}
Retrain-25\% & 0.0 \up{0.1} & 0.4 \up{1.1} & 0.9 \up{1.4} & 12.3 \gap{4.4} & 0.1 \up{0.1} & 0.4 \up{0.3} & 0.9 \up{0.5} & 19.5 \gap{4.7} & 0.0 \up{0.1} & 0.3 \up{0.8} & 0.8 \up{1.6} & 14.8 \gap{1.9} \\
Retrain-50\% & 0.0 \up{0.1} & 0.8 \up{0.7} & 1.6 \up{0.7} & 15.2 \gap{1.5} & 0.2 & 0.6 \up{0.1} & 1.1 \up{0.3} & 22.8 \gap{1.4} & 0.0 \up{0.1} & 0.6 \up{0.5} & 1.6 \up{0.8} & 16.0 \gap{0.7} \\
\rowcolor{gray!10}
Retrain-75\% & 0.1 & 1.2 \up{0.3} & 1.9 \up{0.4} & 16.4 \gap{0.3} & 0.2 & 0.7 & 1.3 \up{0.1} & 23.9 \gap{0.3} & 0.1 & 0.9 \up{0.2} & 2.1 \up{0.3} & 16.5 \gap{0.2} \\
\textbf{EASE} & 0.3 \down{0.2} & 1.0 \up{0.5} & 2.2 \up{0.1} & 16.9 \gap{0.2} & 0.3 \down{0.1} & 1.5 \down{0.8} & 2.6 \down{1.2} & 25.4 \gap{1.2} & 0.1 & 2.2 \down{1.1} & 5.5 \down{3.1} & 16.7 \\
\textit{Retrain (Full)} & 0.1 & 1.5 & 2.3 & 16.7 & 0.2 & 0.7 & 1.4 & 24.2 & 0.1 & 1.1 & 2.4 & 16.7 \\
\Xhline{1.2pt}
\end{tabular}}
\vspace{-2pt}
\resizebox{\linewidth}{!}{
\begin{tabular}{l|cccc|cccc|cccc}
\Xhline{1.2pt}
\rowcolor{CadetBlue!20}
\textbf{Method} & $R\text{-}R@1\;\uparrow$ & $R\text{-}R@5\;\uparrow$ & $R\text{-}R@10\;\uparrow$ & $Comm.\;\downarrow$ & $R\text{-}R@1\;\uparrow$ & $R\text{-}R@5\;\uparrow$ & $R\text{-}R@10\;\uparrow$ & $Comm.\;\downarrow$ & $R\text{-}R@1\;\uparrow$ & $R\text{-}R@5\;\uparrow$ & $R\text{-}R@10\;\uparrow$ & $Comm.\;\downarrow$ \\
\Xhline{1.2pt}
\rowcolor{gray!10}
Retrain-25\% & 62.3 \down{28.4} & 88.5 \down{11.4} & 94.7 \down{5.3} & 65.6 \gap{196.9} & 61.7 \down{27.7} & 87.5 \down{12.4} & 93.9 \down{6.1} & 49.2 \gap{147.7} & 54.2 \down{26.4} & 83.1 \down{14.6} & 90.5 \down{8.9} & 49.2 \gap{147.7} \\
Retrain-50\% & 79.8 \down{10.9} & 97.1 \down{2.8} & 99.0 \down{1.0} & 131.3 \gap{131.2} & 78.0 \down{11.4} & 95.6 \down{4.3} & 98.3 \down{1.7} & 98.5 \gap{98.4} & 69.2 \down{11.4} & 92.6 \down{5.1} & 96.8 \down{2.6} & 98.5 \gap{98.4} \\
\rowcolor{gray!10}
Retrain-75\% & 87.1 \down{3.6} & 99.3 \down{0.6} & 99.8 \down{0.2} & 196.9 \gap{65.6} & 85.1 \down{4.3} & 98.5 \down{1.4} & 99.5 \down{0.5} & 147.7 \gap{49.2} & 76.8 \down{3.8} & 96.4 \down{1.3} & 98.5 \down{0.9} & 147.7 \gap{49.2} \\
\textbf{EASE} & 86.5 \down{4.2} & 99.2 \down{0.7} & 99.8 \down{0.2} & 52.5 \gap{210.0} & 80.8 \down{8.6} & 96.6 \down{3.3} & 99.3 \down{0.7} & 47.2 \gap{149.7} & 79.4 \down{1.2} & 97.7 & 99.2 \down{0.2} & 39.4 \gap{157.5} \\
\textit{Retrain (Full)} & 90.7 & 99.9 & 100.0 & 262.5 & 89.4 & 99.9 & 100.0 & 196.9 & 80.6 & 97.7 & 99.4 & 196.9 \\
\Xhline{1.2pt}
\end{tabular}}
\end{table*}

\underline{\texttt{EASE}} uses less communication than Retrain-$25\%$ in every scenario (Table~\ref{tab:compute_matched_flickr30k_clipb32}) yet improves $R\text{-}R@1$ by $+24.2$, $+19.1$, and $+25.2$ points on the client, class, and sample scenarios, respectively, while keeping forget-side metrics near the full-budget retrain reference. Retrain-$25\%$ can also produce low forget recall, but its retain collapse indicates undertraining rather than targeted unlearning. EASE reaches retain quality close to Retrain-$75\%$ at only $20$--$25\%$ of the full communication budget.

\section{Likelihood-Ratio Attack at Low False-Positive Rate}
\label{sec:app_lira}

We report $\mathrm{LiRA}~\mathrm{TPR}@\mathrm{FPR}{=}1\%$, defined in Appendix~\ref{sec:app_metrics}, across all nine (dataset, backbone) combinations and three unlearning scenarios in Figure~\ref{fig:lira_all}, averaged over three seeds with standard-deviation error bars.

\begin{figure}[h]
    \centering
    \includegraphics[width=\linewidth]{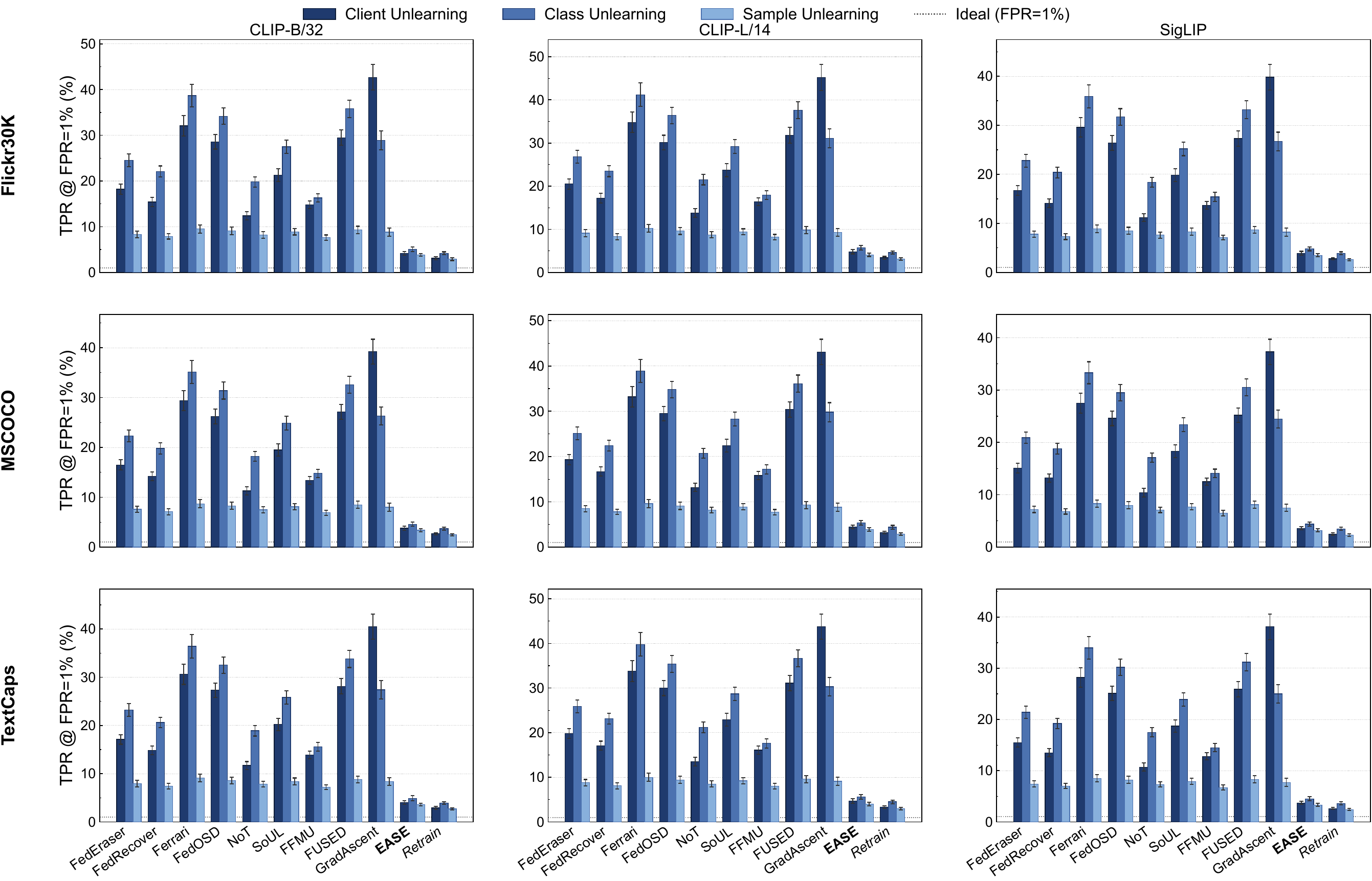}
    \caption{$\mathrm{TPR}@\mathrm{FPR}=1\%$ across three datasets (rows: Flickr30K, MSCOCO, TextCaps) and three backbones (columns: CLIP-B/32, CLIP-L/14, SigLIP). Lower is better; the dotted line marks the $1\%$ random-guess reference for non-member-like behavior. \textit{Retrain} is the non-participating reference.}
    \label{fig:lira_all}
\end{figure}

\paragraph{Cross-dataset pattern.}
Across all nine combinations, \underline{\texttt{EASE}} stays within $1.5$ points of \textit{Retrain} on every scenario, placing it consistently at the bottom of each subplot. In contrast, forget-reconstruction baselines such as Ferrari, FedOSD, and FUSED sit in the $25$--$42$\% range on the client and class scenarios, suggesting that their aggregate MIA values in Section~\ref{sec:q1} underestimate how much forget-side signal a targeted attacker can still recover. GradAscent produces the largest low-FPR leakage in the client scenario, consistent with the known instability of gradient-reversal updates on tightly coupled modalities.

\paragraph{Backbone sensitivity.}
The relative ordering is stable across CLIP-B/32, CLIP-L/14, and SigLIP; absolute values shift upward for larger encoders (CLIP-L/14) due to greater capacity for memorization, but the gap between \underline{\texttt{EASE}} and baselines widens rather than shrinks. This is consistent with bilateral excision targeting the Modality Anchor that scales with encoder capacity, while per-modality baselines do not close this anchor regardless of backbone size.

\paragraph{Scenario dependence.}
The sample scenario produces uniformly lower $\mathrm{TPR}$ than client and class, because sample forget sets are sparsely distributed across clients and carry a weaker attacker signature in aggregate. Even so, our method approaches \textit{Retrain} most closely in this regime, indicating that gains from bilateral excision persist even where the residual leakage is already small.

\section{Full Comparison Results}
\label{sec:app_baselines_full}

We report the remaining eight (dataset, backbone) combinations here, extending Table~\ref{tab:q1_flickr30k_clipb32} in the main body. Each combination follows the same two-part layout and ranking convention: the top sub-table covers forget-side metrics ($F\text{-}R@k$) and MIA, and the bottom sub-table covers retain-side metrics ($R\text{-}R@k$) and communication cost.

\begin{table*}[!ht]
\centering
\setlength{\tabcolsep}{3pt}
\setlength{\aboverulesep}{0.3pt}
\setlength{\belowrulesep}{0.3pt}
\renewcommand{\arraystretch}{0.95}
\caption{Main comparison on Flickr30K with CLIP-L/14 across three unlearning scenarios. Convention follows Table~\ref{tab:q1_flickr30k_clipb32}.}
\label{tab:q1_flickr30k_clipl14}
\resizebox{\linewidth}{!}{
\begin{tabular}{l|cccc|cccc|cccc}
\Xhline{1.2pt}
\rowcolor{CadetBlue!20}
\textbf{Method} & \multicolumn{4}{c|}{\textbf{Client Unlearning}} & \multicolumn{4}{c|}{\textbf{Class Unlearning}} & \multicolumn{4}{c}{\textbf{Sample Unlearning}} \\
\cmidrule(lr){2-5} \cmidrule(lr){6-9} \cmidrule(lr){10-13}
\rowcolor{CadetBlue!20}
 & $F\text{-}R@1\;\downarrow$ & $F\text{-}R@5\;\downarrow$ & $F\text{-}R@10\;\downarrow$ & $MIA$ & $F\text{-}R@1\;\downarrow$ & $F\text{-}R@5\;\downarrow$ & $F\text{-}R@10\;\downarrow$ & $MIA$ & $F\text{-}R@1\;\downarrow$ & $F\text{-}R@5\;\downarrow$ & $F\text{-}R@10\;\downarrow$ & $MIA$ \\
\Xhline{1.2pt}
\rowcolor{gray!10}
FedEraser & 88.5 \down{88.4} & 97.3 \down{96.2} & 99.2 \down{97.2} & 52.8 \gap{36.1} & 79.8 \down{79.6} & 95.5 \down{95.0} & 97.7 \down{96.2} & 57.6 \gap{40.8} & \best{0.0} & 0.9 \up{0.5} & \best{1.7} \up{0.3} & \second{16.7} \\
FedRecover & 75.3 \down{75.2} & 91.5 \down{90.4} & 94.5 \down{92.5} & 54.5 \gap{37.8} & 69.7 \down{69.5} & 87.0 \down{86.5} & 91.8 \down{90.3} & 59.6 \gap{42.8} & \second{0.0} & 1.0 \up{0.4} & 2.1 \down{0.1} & 16.7 \\
\rowcolor{gray!10}
Ferrari & 95.9 \down{95.8} & 99.9 \down{98.8} & 100.0 \down{98.0} & \best{59.5} \gap{42.8} & 94.5 \down{94.3} & 99.9 \down{99.4} & 100.0 \down{98.5} & \best{63.7} \gap{46.9} & 0.0 & 1.0 \up{0.4} & 2.4 \down{0.4} & 16.7 \\
FedOSD & 95.7 \down{95.6} & 100.0 \down{98.9} & 100.0 \down{98.0} & 57.6 \gap{40.9} & 93.2 \down{93.0} & 100.0 \down{99.5} & 100.0 \down{98.5} & 61.5 \gap{44.7} & 0.3 \down{0.3} & 1.8 \down{0.4} & 2.8 \down{0.8} & 16.7 \\
\rowcolor{gray!10}
NoT & 85.4 \down{85.3} & 98.8 \down{97.7} & 99.4 \down{97.4} & 51.2 \gap{34.5} & 76.4 \down{76.2} & 96.9 \down{96.4} & 98.8 \down{97.3} & 55.9 \gap{39.1} & 0.1 \down{0.1} & 1.3 \up{0.1} & 3.3 \down{1.3} & 16.7 \\
SoUL & 93.6 \down{93.5} & 99.8 \down{98.7} & 100.0 \down{98.0} & 53.6 \gap{36.9} & 91.2 \down{91.0} & 99.4 \down{98.9} & 99.9 \down{98.4} & 59.7 \gap{42.9} & 0.1 \down{0.1} & 0.9 \up{0.5} & 2.0 & 16.7 \\
\end{tabular}}
\vspace{-2pt}
\resizebox{\linewidth}{!}{
\begin{tabular}{l|cccc|cccc|cccc}
\Xhline{1.2pt}
\rowcolor{CadetBlue!20}
\textbf{Method} & $R\text{-}R@1\;\uparrow$ & $R\text{-}R@5\;\uparrow$ & $R\text{-}R@10\;\uparrow$ & $Comm.\;\downarrow$ & $R\text{-}R@1\;\uparrow$ & $R\text{-}R@5\;\uparrow$ & $R\text{-}R@10\;\uparrow$ & $Comm.\;\downarrow$ & $R\text{-}R@1\;\uparrow$ & $R\text{-}R@5\;\uparrow$ & $R\text{-}R@10\;\uparrow$ & $Comm.\;\downarrow$ \\
\Xhline{1.2pt}
\rowcolor{gray!10}
FedEraser & \second{94.4} \down{5.2} & 99.3 \down{0.7} & 99.5 \down{0.5} & 537.5 & 90.4 \down{6.5} & 99.1 \down{0.9} & 99.7 \down{0.3} & 403.1 & 91.3 \down{7.4} & 99.2 \down{0.8} & 99.9 \down{0.1} & 403.1 \\
FedRecover & 91.9 \down{7.7} & 98.3 \down{1.7} & 99.4 \down{0.6} & \best{96.8} \gap{440.7} & 89.8 \down{7.1} & 98.5 \down{1.5} & 99.3 \down{0.7} & \best{72.6} \gap{330.5} & 72.2 \down{26.5} & 87.4 \down{12.6} & 91.4 \down{8.6} & \best{72.6} \gap{330.5} \\
\rowcolor{gray!10}
Ferrari & 94.1 \down{5.5} & \second{99.2} \down{0.8} & \second{99.6} \down{0.4} & 147.8 \gap{389.7} & \second{91.2} \down{5.7} & 99.2 \down{0.8} & 99.8 \down{0.2} & 120.9 \gap{282.2} & \second{92.1} \down{6.6} & 99.1 \down{0.9} & 99.8 \down{0.2} & 120.9 \gap{282.2} \\
FedOSD & 93.7 \down{5.9} & 99.0 \down{1.0} & 99.6 \down{0.4} & 215.0 \gap{322.5} & 91.5 \down{5.4} & 99.4 \down{0.6} & 99.8 \down{0.2} & 161.2 \gap{241.9} & 92.7 \down{6.0} & 99.5 \down{0.5} & 99.7 \down{0.3} & 161.2 \gap{241.9} \\
\rowcolor{gray!10}
NoT & 92.5 \down{7.1} & 98.7 \down{1.3} & 99.4 \down{0.6} & \second{107.5} \gap{430.0} & 85.3 \down{11.6} & \second{98.7} \down{1.3} & 99.5 \down{0.5} & \second{80.6} \gap{322.5} & 90.5 \down{8.2} & \second{98.8} \down{1.2} & 99.6 \down{0.4} & \second{80.6} \gap{322.5} \\
SoUL & 93.3 \down{6.3} & 99.1 \down{0.9} & 99.6 \down{0.4} & 107.5 \gap{430.0} & 90.8 \down{6.1} & 99.1 \down{0.9} & \second{99.7} \down{0.3} & 80.6 \gap{322.5} & 91.7 \down{7.0} & 99.2 \down{0.8} & \second{99.8} \down{0.2} & 80.6 \gap{322.5} \\
\end{tabular}}
\end{table*}

\begin{table*}[!ht]
\centering
\setlength{\tabcolsep}{3pt}
\setlength{\aboverulesep}{0.3pt}
\setlength{\belowrulesep}{0.3pt}
\renewcommand{\arraystretch}{0.95}
\caption{Main comparison on Flickr30K with SigLIP across three unlearning scenarios. Convention follows Table~\ref{tab:q1_flickr30k_clipb32}.}
\label{tab:q1_flickr30k_siglip}
\resizebox{\linewidth}{!}{
\begin{tabular}{l|cccc|cccc|cccc}
\Xhline{1.2pt}
\rowcolor{CadetBlue!20}
\textbf{Method} & \multicolumn{4}{c|}{\textbf{Client Unlearning}} & \multicolumn{4}{c|}{\textbf{Class Unlearning}} & \multicolumn{4}{c}{\textbf{Sample Unlearning}} \\
\cmidrule(lr){2-5} \cmidrule(lr){6-9} \cmidrule(lr){10-13}
\rowcolor{CadetBlue!20}
 & $F\text{-}R@1\;\downarrow$ & $F\text{-}R@5\;\downarrow$ & $F\text{-}R@10\;\downarrow$ & $MIA$ & $F\text{-}R@1\;\downarrow$ & $F\text{-}R@5\;\downarrow$ & $F\text{-}R@10\;\downarrow$ & $MIA$ & $F\text{-}R@1\;\downarrow$ & $F\text{-}R@5\;\downarrow$ & $F\text{-}R@10\;\downarrow$ & $MIA$ \\
\Xhline{1.2pt}
\rowcolor{gray!10}
FedEraser & 82.1 \down{81.7} & 93.9 \down{92.6} & 97.0 \down{94.9} & 51.2 \gap{34.5} & 57.4 \down{57.3} & 82.6 \down{81.4} & 91.2 \down{89.1} & 64.0 \gap{45.7} & \second{0.1} \down{0.1} & 0.7 \up{0.2} & 1.6 \up{0.1} & 16.7 \\
FedRecover & 73.1 \down{72.7} & 88.6 \down{87.3} & 92.6 \down{90.5} & 48.1 \gap{31.4} & 51.2 \down{51.1} & 78.6 \down{77.4} & 85.6 \down{83.5} & 47.3 \gap{29.0} & 0.3 \down{0.3} & \best{0.4} \up{0.5} & 1.4 \up{0.3} & \best{16.8} \gap{0.1} \\
\rowcolor{gray!10}
Ferrari & 93.1 \down{92.7} & 99.8 \down{98.5} & 100.0 \down{97.9} & \best{54.9} \gap{38.2} & 77.4 \down{77.3} & 97.9 \down{96.7} & 99.6 \down{97.5} & \best{67.8} \gap{49.5} & 0.1 \down{0.1} & 0.9 & \second{1.3} \up{0.4} & 16.7 \\
FedOSD & 90.4 \down{90.0} & 99.3 \down{98.0} & 99.9 \down{97.8} & 52.4 \gap{35.7} & 75.9 \down{75.8} & 96.6 \down{95.4} & 99.5 \down{97.4} & 66.5 \gap{48.2} & 0.2 \down{0.2} & 0.8 \up{0.1} & 1.8 \down{0.1} & 16.7 \\
\rowcolor{gray!10}
NoT & 77.6 \down{77.2} & 95.7 \down{94.4} & 98.1 \down{96.0} & 49.9 \gap{33.2} & 53.5 \down{53.4} & 84.7 \down{83.5} & 93.1 \down{91.0} & 58.5 \gap{40.2} & 0.3 \down{0.3} & 0.7 \up{0.2} & 1.6 \up{0.1} & 16.7 \\
\end{tabular}}
\vspace{-2pt}
\resizebox{\linewidth}{!}{
\begin{tabular}{l|cccc|cccc|cccc}
\Xhline{1.2pt}
\rowcolor{CadetBlue!20}
\textbf{Method} & $R\text{-}R@1\;\uparrow$ & $R\text{-}R@5\;\uparrow$ & $R\text{-}R@10\;\uparrow$ & $Comm.\;\downarrow$ & $R\text{-}R@1\;\uparrow$ & $R\text{-}R@5\;\uparrow$ & $R\text{-}R@10\;\uparrow$ & $Comm.\;\downarrow$ & $R\text{-}R@1\;\uparrow$ & $R\text{-}R@5\;\uparrow$ & $R\text{-}R@10\;\uparrow$ & $Comm.\;\downarrow$ \\
\Xhline{1.2pt}
\rowcolor{gray!10}
FedEraser & 90.1 \down{9.1} & 99.1 \down{0.9} & 99.7 \down{0.3} & 475.0 & 89.4 \down{8.4} & 98.7 \down{1.3} & 99.2 \down{0.8} & 356.2 & 92.4 \down{5.9} & 99.3 \down{0.7} & 99.8 \down{0.2} & 356.2 \\
FedRecover & 87.9 \down{11.3} & 98.2 \down{1.8} & 99.3 \down{0.7} & \best{85.5} \gap{389.5} & 86.1 \down{11.7} & 97.9 \down{2.1} & 98.8 \down{1.2} & \best{64.1} \gap{292.1} & 68.0 \down{30.3} & 86.2 \down{13.8} & 90.5 \down{9.5} & \best{64.1} \gap{292.1} \\
\rowcolor{gray!10}
Ferrari & 90.9 \down{8.3} & 99.1 \down{0.9} & \second{99.8} \down{0.2} & 130.6 \gap{344.4} & 89.1 \down{8.7} & 98.6 \down{1.4} & 99.2 \down{0.8} & 106.9 \gap{249.3} & 91.6 \down{6.7} & 99.5 \down{0.5} & \second{99.8} \down{0.2} & 106.9 \gap{249.3} \\
FedOSD & 91.2 \down{8.0} & \second{99.3} \down{0.7} & 99.8 \down{0.2} & 190.0 \gap{285.0} & 88.3 \down{9.5} & 98.7 \down{1.3} & 99.1 \down{0.9} & 142.5 \gap{213.7} & 89.7 \down{8.6} & 99.5 \down{0.5} & 99.6 \down{0.4} & 142.5 \gap{213.7} \\
\rowcolor{gray!10}
NoT & \second{82.7} \down{16.5} & 98.6 \down{1.4} & 99.5 \down{0.5} & \second{95.0} \gap{380.0} & 87.5 \down{10.3} & 98.1 \down{1.9} & 99.0 \down{1.0} & \second{71.2} \gap{285.0} & 82.0 \down{16.3} & 99.0 \down{1.0} & 99.5 \down{0.5} & \second{71.2} \gap{285.0} \\
SoUL & 90.5 \down{8.7} & 99.0 \down{1.0} & 99.7 \down{0.3} & 95.0 \gap{380.0} & \second{88.7} \down{9.1} & \second{98.5} \down{1.5} & 99.2 \down{0.8} & 71.2 \gap{285.0} & 92.0 \down{6.3} & 99.4 \down{0.6} & 99.7 \down{0.3} & 71.2 \gap{285.0} \\
\end{tabular}}
\end{table*}

\begin{table*}[!ht]
\centering
\setlength{\tabcolsep}{3pt}
\setlength{\aboverulesep}{0.3pt}
\setlength{\belowrulesep}{0.3pt}
\renewcommand{\arraystretch}{0.95}
\caption{Main comparison on MSCOCO with CLIP-B/32 across three unlearning scenarios. Convention follows Table~\ref{tab:q1_flickr30k_clipb32}.}
\label{tab:q1_mscoco_clipb32}
\resizebox{\linewidth}{!}{
\begin{tabular}{l|cccc|cccc|cccc}
\Xhline{1.2pt}
\rowcolor{CadetBlue!20}
\textbf{Method} & \multicolumn{4}{c|}{\textbf{Client Unlearning}} & \multicolumn{4}{c|}{\textbf{Class Unlearning}} & \multicolumn{4}{c}{\textbf{Sample Unlearning}} \\
\cmidrule(lr){2-5} \cmidrule(lr){6-9} \cmidrule(lr){10-13}
\rowcolor{CadetBlue!20}
 & $F\text{-}R@1\;\downarrow$ & $F\text{-}R@5\;\downarrow$ & $F\text{-}R@10\;\downarrow$ & $MIA$ & $F\text{-}R@1\;\downarrow$ & $F\text{-}R@5\;\downarrow$ & $F\text{-}R@10\;\downarrow$ & $MIA$ & $F\text{-}R@1\;\downarrow$ & $F\text{-}R@5\;\downarrow$ & $F\text{-}R@10\;\downarrow$ & $MIA$ \\
\Xhline{1.2pt}
\rowcolor{gray!10}
FedEraser & 39.5 \down{39.5} & 69.9 \down{69.2} & 81.4 \down{79.9} & 50.1 \gap{9.2} & 38.5 \down{38.4} & 68.7 \down{68.0} & 82.5 \down{80.6} & 53.2 \gap{12.1} & \best{0.1} \down{0.1} & \best{0.7} \down{0.2} & 1.8 \down{1.0} & 40.9 \\
FedRecover & 32.6 \down{32.6} & \second{58.4} \down{57.7} & \second{71.4} \down{69.9} & 49.9 \gap{9.0} & 42.5 \down{42.4} & 74.0 \down{73.3} & 83.2 \down{81.3} & 49.7 \gap{8.6} & 0.6 \down{0.6} & 1.0 \down{0.5} & 1.5 \down{0.7} & \best{41.1} \gap{0.2} \\
\rowcolor{gray!10}
Ferrari & 54.2 \down{54.2} & 85.0 \down{84.3} & 92.5 \down{91.0} & 50.8 \gap{9.9} & 56.4 \down{56.3} & 86.4 \down{85.7} & 94.6 \down{92.7} & \best{55.7} \gap{14.6} & 0.3 \down{0.3} & 1.0 \down{0.5} & 1.6 \down{0.8} & 40.9 \\
FedOSD & 49.4 \down{49.4} & 80.9 \down{80.2} & 89.9 \down{88.4} & \best{51.0} \gap{10.1} & 55.6 \down{55.5} & 85.4 \down{84.7} & 93.5 \down{91.6} & 55.1 \gap{14.0} & 0.2 \down{0.2} & 0.9 \down{0.4} & 1.8 \down{1.0} & 40.9 \\
\rowcolor{gray!10}
NoT & 39.6 \down{39.6} & 70.4 \down{69.7} & 82.5 \down{81.0} & 50.2 \gap{9.3} & 35.1 \down{35.0} & 68.4 \down{67.7} & 78.8 \down{76.9} & 47.9 \gap{6.8} & 0.4 \down{0.4} & 1.1 \down{0.6} & 1.5 \down{0.7} & \second{41.0} \gap{0.1} \\
\end{tabular}}
\vspace{-2pt}
\resizebox{\linewidth}{!}{
\begin{tabular}{l|cccc|cccc|cccc}
\Xhline{1.2pt}
\rowcolor{CadetBlue!20}
\textbf{Method} & $R\text{-}R@1\;\uparrow$ & $R\text{-}R@5\;\uparrow$ & $R\text{-}R@10\;\uparrow$ & $Comm.\;\downarrow$ & $R\text{-}R@1\;\uparrow$ & $R\text{-}R@5\;\uparrow$ & $R\text{-}R@10\;\uparrow$ & $Comm.\;\downarrow$ & $R\text{-}R@1\;\uparrow$ & $R\text{-}R@5\;\uparrow$ & $R\text{-}R@10\;\uparrow$ & $Comm.\;\downarrow$ \\
\Xhline{1.2pt}
\rowcolor{gray!10}
FedEraser & 59.8 \down{10.7} & 90.1 \down{5.1} & 95.5 \down{3.6} & 262.5 & 65.1 \down{10.7} & 95.1 \down{3.3} & 98.4 \down{1.6} & 196.9 & 49.9 \down{1.7} & 84.6 \down{1.9} & 93.2 \down{1.4} & 196.9 \\
FedRecover & 56.8 \down{13.7} & 89.7 \down{5.5} & 95.3 \down{3.8} & \best{47.2} \gap{215.3} & 62.1 \down{13.7} & 94.4 \down{4.0} & 98.2 \down{1.8} & \best{35.4} \gap{161.5} & 27.4 \down{24.2} & 53.3 \down{33.2} & 64.4 \down{30.2} & \best{35.4} \gap{161.5} \\
\rowcolor{gray!10}
Ferrari & \second{59.8} \down{10.7} & \second{90.1} \down{5.1} & 95.5 \down{3.6} & 72.2 \gap{190.3} & 64.3 \down{11.5} & 95.1 \down{3.3} & 98.4 \down{1.6} & 59.1 \gap{137.8} & 48.8 \down{2.8} & 84.8 \down{1.7} & 93.2 \down{1.4} & 59.1 \gap{137.8} \\
FedOSD & 60.1 \down{10.4} & 90.2 \down{5.0} & \second{95.6} \down{3.5} & 105.0 \gap{157.5} & 65.7 \down{10.1} & 95.0 \down{3.4} & 98.5 \down{1.5} & 78.8 \gap{118.1} & \second{43.5} \down{8.1} & 76.8 \down{9.7} & 87.1 \down{7.5} & 78.8 \gap{118.1} \\
\rowcolor{gray!10}
NoT & 47.8 \down{22.7} & 81.5 \down{13.7} & 91.9 \down{7.2} & \second{52.5} \gap{210.0} & \second{63.5} \down{12.3} & 94.6 \down{3.8} & \second{98.2} \down{1.8} & \second{39.4} \gap{157.5} & 38.0 \down{13.6} & 68.6 \down{17.9} & 81.3 \down{13.3} & \second{39.4} \gap{157.5} \\
SoUL & 59.4 \down{11.1} & 90.0 \down{5.2} & 95.5 \down{3.6} & 52.5 \gap{210.0} & 64.7 \down{11.1} & \second{95.0} \down{3.4} & 98.3 \down{1.7} & 39.4 \gap{157.5} & 44.9 \down{6.7} & \second{77.1} \down{9.4} & 88.7 \down{5.9} & 39.4 \gap{157.5} \\
\end{tabular}}
\end{table*}

\begin{table*}[!ht]
\centering
\setlength{\tabcolsep}{3pt}
\setlength{\aboverulesep}{0.3pt}
\setlength{\belowrulesep}{0.3pt}
\renewcommand{\arraystretch}{0.95}
\caption{Main comparison on MSCOCO with CLIP-L/14 across three unlearning scenarios. Convention follows Table~\ref{tab:q1_flickr30k_clipb32}.}
\label{tab:q1_mscoco_clipl14}
\resizebox{\linewidth}{!}{
\begin{tabular}{l|cccc|cccc|cccc}
\Xhline{1.2pt}
\rowcolor{CadetBlue!20}
\textbf{Method} & \multicolumn{4}{c|}{\textbf{Client Unlearning}} & \multicolumn{4}{c|}{\textbf{Class Unlearning}} & \multicolumn{4}{c}{\textbf{Sample Unlearning}} \\
\cmidrule(lr){2-5} \cmidrule(lr){6-9} \cmidrule(lr){10-13}
\rowcolor{CadetBlue!20}
 & $F\text{-}R@1\;\downarrow$ & $F\text{-}R@5\;\downarrow$ & $F\text{-}R@10\;\downarrow$ & $MIA$ & $F\text{-}R@1\;\downarrow$ & $F\text{-}R@5\;\downarrow$ & $F\text{-}R@10\;\downarrow$ & $MIA$ & $F\text{-}R@1\;\downarrow$ & $F\text{-}R@5\;\downarrow$ & $F\text{-}R@10\;\downarrow$ & $MIA$ \\
\Xhline{1.2pt}
\rowcolor{gray!10}
FedEraser & 30.9 \down{30.6} & 55.6 \down{54.7} & 67.3 \down{65.7} & 49.1 \gap{8.2} & 36.0 \down{35.9} & 61.8 \down{61.2} & 74.7 \down{72.5} & 51.2 \gap{10.1} & \best{0.1} & \best{0.6} \up{0.5} & \second{1.5} \up{0.5} & \second{40.9} \\
FedRecover & \second{24.8} \down{24.5} & \second{49.0} \down{48.1} & \second{59.3} \down{57.7} & \second{51.3} \gap{10.4} & 40.9 \down{40.8} & 69.5 \down{68.9} & 77.4 \down{75.2} & 46.7 \gap{5.6} & \second{0.1} & 0.9 \up{0.2} & 1.9 \up{0.1} & 40.9 \\
\rowcolor{gray!10}
Ferrari & 49.9 \down{49.6} & 82.0 \down{81.1} & 91.4 \down{89.8} & \best{51.9} \gap{11.0} & 59.9 \down{59.8} & 86.2 \down{85.6} & 91.8 \down{89.6} & \best{52.3} \gap{11.2} & 0.1 & 0.9 \up{0.2} & 2.0 & 40.9 \\
FedOSD & 44.3 \down{44.0} & 74.8 \down{73.9} & 86.4 \down{84.8} & 50.7 \gap{9.8} & 57.9 \down{57.8} & 85.4 \down{84.8} & 91.8 \down{89.6} & \second{52.1} \gap{11.0} & 0.2 \down{0.1} & 0.9 \up{0.2} & 2.3 \down{0.3} & \best{41.0} \gap{0.1} \\
\rowcolor{gray!10}
NoT & 34.6 \down{34.3} & 64.8 \down{63.9} & 78.7 \down{77.1} & 49.9 \gap{9.0} & 35.2 \down{35.1} & 62.4 \down{61.8} & 73.3 \down{71.1} & 48.0 \gap{6.9} & 0.3 \down{0.2} & \second{0.7} \up{0.4} & \best{1.4} \up{0.6} & 40.9 \\
\end{tabular}}
\vspace{-2pt}
\resizebox{\linewidth}{!}{
\begin{tabular}{l|cccc|cccc|cccc}
\Xhline{1.2pt}
\rowcolor{CadetBlue!20}
\textbf{Method} & $R\text{-}R@1\;\uparrow$ & $R\text{-}R@5\;\uparrow$ & $R\text{-}R@10\;\uparrow$ & $Comm.\;\downarrow$ & $R\text{-}R@1\;\uparrow$ & $R\text{-}R@5\;\uparrow$ & $R\text{-}R@10\;\uparrow$ & $Comm.\;\downarrow$ & $R\text{-}R@1\;\uparrow$ & $R\text{-}R@5\;\uparrow$ & $R\text{-}R@10\;\uparrow$ & $Comm.\;\downarrow$ \\
\Xhline{1.2pt}
\rowcolor{gray!10}
FedEraser & 57.8 \down{14.9} & \second{88.9} \down{7.8} & 93.8 \down{5.5} & 537.5 & 77.2 \down{10.3} & 98.4 \down{1.3} & 99.7 \down{0.3} & 403.1 & \second{53.1} \down{7.9} & 82.5 \down{7.7} & 91.9 \down{4.7} & 403.1 \\
FedRecover & 55.6 \down{17.1} & 88.4 \down{8.3} & 93.6 \down{5.7} & \best{96.8} \gap{440.7} & 74.7 \down{12.8} & 97.6 \down{2.1} & 99.2 \down{0.8} & \best{72.6} \gap{330.5} & 22.5 \down{38.5} & 42.8 \down{47.4} & 55.2 \down{41.4} & \best{72.6} \gap{330.5} \\
\rowcolor{gray!10}
Ferrari & 57.4 \down{15.3} & 88.6 \down{8.1} & 93.8 \down{5.5} & 147.8 \gap{389.7} & 76.9 \down{10.6} & 98.3 \down{1.4} & 99.5 \down{0.5} & 120.9 \gap{282.2} & 52.8 \down{8.2} & 82.3 \down{7.9} & 91.8 \down{4.8} & 120.9 \gap{282.2} \\
FedOSD & 58.4 \down{14.3} & 88.8 \down{7.9} & \second{93.9} \down{5.4} & 215.0 \gap{322.5} & \second{77.5} \down{10.0} & 98.4 \down{1.3} & 99.6 \down{0.4} & 161.2 \gap{241.9} & 40.8 \down{20.2} & 73.1 \down{17.1} & \second{85.1} \down{11.5} & 161.2 \gap{241.9} \\
\rowcolor{gray!10}
NoT & \second{50.8} \down{21.9} & 81.6 \down{15.1} & 90.4 \down{8.9} & \second{107.5} \gap{430.0} & 75.3 \down{12.2} & \second{97.8} \down{1.9} & 99.4 \down{0.6} & \second{80.6} \gap{322.5} & 33.3 \down{27.7} & 67.4 \down{22.8} & 80.7 \down{15.9} & \second{80.6} \gap{322.5} \\
SoUL & 57.4 \down{15.3} & 88.7 \down{8.0} & 93.8 \down{5.5} & 107.5 \gap{430.0} & 76.5 \down{11.0} & 98.2 \down{1.5} & \second{99.6} \down{0.4} & 80.6 \gap{322.5} & 44.2 \down{16.8} & 76.5 \down{13.7} & 88.2 \down{8.4} & 80.6 \gap{322.5} \\
\end{tabular}}
\end{table*}

\begin{table*}[!ht]
\centering
\setlength{\tabcolsep}{3pt}
\setlength{\aboverulesep}{0.3pt}
\setlength{\belowrulesep}{0.3pt}
\renewcommand{\arraystretch}{0.95}
\caption{Main comparison on MSCOCO with SigLIP across three unlearning scenarios. Convention follows Table~\ref{tab:q1_flickr30k_clipb32}.}
\label{tab:q1_mscoco_siglip}
\resizebox{\linewidth}{!}{
\begin{tabular}{l|cccc|cccc|cccc}
\Xhline{1.2pt}
\rowcolor{CadetBlue!20}
\textbf{Method} & \multicolumn{4}{c|}{\textbf{Client Unlearning}} & \multicolumn{4}{c|}{\textbf{Class Unlearning}} & \multicolumn{4}{c}{\textbf{Sample Unlearning}} \\
\cmidrule(lr){2-5} \cmidrule(lr){6-9} \cmidrule(lr){10-13}
\rowcolor{CadetBlue!20}
 & $F\text{-}R@1\;\downarrow$ & $F\text{-}R@5\;\downarrow$ & $F\text{-}R@10\;\downarrow$ & $MIA$ & $F\text{-}R@1\;\downarrow$ & $F\text{-}R@5\;\downarrow$ & $F\text{-}R@10\;\downarrow$ & $MIA$ & $F\text{-}R@1\;\downarrow$ & $F\text{-}R@5\;\downarrow$ & $F\text{-}R@10\;\downarrow$ & $MIA$ \\
\Xhline{1.2pt}
\rowcolor{gray!10}
FedEraser & 56.7 \down{56.5} & 87.9 \down{87.3} & 92.3 \down{90.7} & 50.7 \gap{9.8} & 39.2 \down{39.2} & 71.2 \down{70.7} & 81.4 \down{80.4} & 55.8 \gap{14.8} & 0.3 \down{0.1} & 0.8 \up{0.4} & \best{1.1} \up{1.3} & 40.9 \\
FedRecover & 52.0 \down{51.8} & 79.0 \down{78.4} & 86.9 \down{85.3} & 49.6 \gap{8.7} & 48.1 \down{48.1} & 75.1 \down{74.6} & 84.2 \down{83.2} & 53.5 \gap{12.5} & 0.2 & 1.1 \up{0.1} & 2.3 \up{0.1} & 41.2 \gap{0.3} \\
\rowcolor{gray!10}
Ferrari & 69.7 \down{69.5} & 93.5 \down{92.9} & 97.3 \down{95.7} & \best{52.6} \gap{11.7} & 66.7 \down{66.7} & 93.2 \down{92.7} & 97.7 \down{96.7} & \best{56.4} \gap{15.4} & \best{0.1} \up{0.1} & 0.8 \up{0.4} & 1.9 \up{0.5} & 40.5 \gap{0.4} \\
FedOSD & 66.8 \down{66.6} & 92.8 \down{92.2} & 96.9 \down{95.3} & 52.1 \gap{11.2} & 62.1 \down{62.1} & 90.7 \down{90.2} & 96.8 \down{95.8} & \second{56.2} \gap{15.2} & \second{0.1} \up{0.1} & 0.9 \up{0.3} & 2.1 \up{0.3} & 41.0 \gap{0.1} \\
\rowcolor{gray!10}
NoT & 44.2 \down{44.0} & 74.4 \down{73.8} & 83.8 \down{82.2} & 51.4 \gap{10.5} & 49.7 \down{49.7} & 80.9 \down{80.4} & 89.3 \down{88.3} & 51.6 \gap{10.6} & 0.2 & 1.0 \up{0.2} & 2.0 \up{0.4} & \second{42.1} \gap{1.2} \\
\end{tabular}}
\vspace{-2pt}
\resizebox{\linewidth}{!}{
\begin{tabular}{l|cccc|cccc|cccc}
\Xhline{1.2pt}
\rowcolor{CadetBlue!20}
\textbf{Method} & $R\text{-}R@1\;\uparrow$ & $R\text{-}R@5\;\uparrow$ & $R\text{-}R@10\;\uparrow$ & $Comm.\;\downarrow$ & $R\text{-}R@1\;\uparrow$ & $R\text{-}R@5\;\uparrow$ & $R\text{-}R@10\;\uparrow$ & $Comm.\;\downarrow$ & $R\text{-}R@1\;\uparrow$ & $R\text{-}R@5\;\uparrow$ & $R\text{-}R@10\;\uparrow$ & $Comm.\;\downarrow$ \\
\Xhline{1.2pt}
\rowcolor{gray!10}
FedEraser & \second{67.7} \down{4.8} & 96.3 \down{1.9} & 98.9 \down{0.6} & 475.0 & 74.7 \down{11.4} & 97.5 \down{1.9} & \second{99.6} \down{0.4} & 356.2 & 66.7 \down{1.5} & 92.5 & 96.9 \down{1.2} & 356.2 \\
FedRecover & 66.0 \down{6.5} & 95.8 \down{2.4} & 98.7 \down{0.8} & 142.5 \gap{332.5} & 72.5 \down{13.6} & 97.2 \down{2.2} & 99.2 \down{0.8} & \best{64.1} \gap{292.1} & 64.1 \down{4.1} & 91.7 \down{0.8} & 96.7 \down{1.4} & \best{53.8} \gap{302.4} \\
\rowcolor{gray!10}
Ferrari & 67.4 \down{5.1} & 96.1 \down{2.1} & 98.9 \down{0.6} & 130.6 \gap{344.4} & 74.3 \down{11.8} & \second{97.6} \down{1.8} & 99.5 \down{0.5} & 106.9 \gap{249.3} & 66.3 \down{1.9} & \second{92.4} \down{0.1} & \second{96.9} \down{1.2} & 71.3 \gap{284.9} \\
FedOSD & 58.6 \down{13.9} & 91.1 \down{7.1} & 99.0 \down{0.5} & 190.0 \gap{285.0} & \second{75.3} \down{10.8} & 97.4 \down{2.0} & 99.5 \down{0.5} & 142.5 \gap{213.7} & 67.4 \down{0.8} & 92.2 \down{0.3} & 97.0 \down{1.1} & 142.5 \gap{213.7} \\
\rowcolor{gray!10}
NoT & 51.7 \down{20.8} & \second{79.5} \down{18.7} & 88.1 \down{11.4} & \best{71.3} \gap{403.7} & 73.1 \down{13.0} & 97.1 \down{2.3} & 99.3 \down{0.7} & \second{71.2} \gap{285.0} & 55.1 \down{13.1} & 82.3 \down{10.2} & 90.7 \down{7.4} & 71.3 \gap{284.9} \\
SoUL & 67.0 \down{5.5} & 91.6 \down{6.6} & 98.9 \down{0.6} & \second{71.3} \gap{403.7} & 74.3 \down{11.8} & 97.5 \down{1.9} & 99.4 \down{0.6} & 71.2 \gap{285.0} & \second{66.7} \down{1.5} & 92.3 \down{0.2} & 96.8 \down{1.3} & 71.3 \gap{284.9} \\
\end{tabular}}
\end{table*}

\begin{table*}[!ht]
\centering
\setlength{\tabcolsep}{3pt}
\setlength{\aboverulesep}{0.3pt}
\setlength{\belowrulesep}{0.3pt}
\renewcommand{\arraystretch}{0.95}
\caption{Main comparison on TextCaps with CLIP-B/32 across three unlearning scenarios. Convention follows Table~\ref{tab:q1_flickr30k_clipb32}.}
\label{tab:q1_textcaps_clipb32}
\resizebox{\linewidth}{!}{
\begin{tabular}{l|cccc|cccc|cccc}
\Xhline{1.2pt}
\rowcolor{CadetBlue!20}
\textbf{Method} & \multicolumn{4}{c|}{\textbf{Client Unlearning}} & \multicolumn{4}{c|}{\textbf{Class Unlearning}} & \multicolumn{4}{c}{\textbf{Sample Unlearning}} \\
\cmidrule(lr){2-5} \cmidrule(lr){6-9} \cmidrule(lr){10-13}
\rowcolor{CadetBlue!20}
 & $F\text{-}R@1\;\downarrow$ & $F\text{-}R@5\;\downarrow$ & $F\text{-}R@10\;\downarrow$ & $MIA$ & $F\text{-}R@1\;\downarrow$ & $F\text{-}R@5\;\downarrow$ & $F\text{-}R@10\;\downarrow$ & $MIA$ & $F\text{-}R@1\;\downarrow$ & $F\text{-}R@5\;\downarrow$ & $F\text{-}R@10\;\downarrow$ & $MIA$ \\
\Xhline{1.2pt}
\rowcolor{gray!10}
FedEraser & 60.3 \down{60.1} & 85.2 \down{84.7} & 89.6 \down{88.1} & 63.0 \gap{13.7} & 22.0 \down{22.0} & 44.3 \down{42.4} & 53.5 \down{50.3} & 55.3 \gap{3.6} & \second{0.1} \down{0.1} & \second{0.7} \up{0.5} & 2.4 \down{0.4} & \second{51.3} \gap{2.6} \\
FedRecover & 43.7 \down{43.5} & 70.6 \down{70.1} & 78.7 \down{77.2} & 62.5 \gap{13.2} & 18.8 \down{18.8} & 43.5 \down{41.6} & 53.8 \down{50.6} & 54.9 \gap{3.2} & 0.1 \down{0.1} & 0.8 \up{0.4} & 1.7 \up{0.3} & 50.0 \gap{1.3} \\
\rowcolor{gray!10}
Ferrari & 83.3 \down{83.1} & 98.2 \down{97.7} & 98.8 \down{97.3} & \best{63.2} \gap{13.9} & 76.6 \down{76.6} & 95.5 \down{93.6} & 99.1 \down{95.9} & \best{56.6} \gap{4.9} & 0.1 \down{0.1} & 0.7 \up{0.5} & 1.8 \up{0.2} & 49.8 \gap{1.1} \\
FedOSD & 80.4 \down{80.2} & 97.4 \down{96.9} & 98.7 \down{97.2} & 63.1 \gap{13.8} & 43.2 \down{43.2} & 72.8 \down{70.9} & 85.6 \down{82.4} & 55.3 \gap{3.6} & 0.1 \down{0.1} & 1.0 \up{0.2} & 1.7 \up{0.3} & 50.8 \gap{2.1} \\
\rowcolor{gray!10}
NoT & 43.0 \down{42.8} & 77.3 \down{76.8} & 86.5 \down{85.0} & \second{63.2} \gap{13.9} & 11.1 \down{11.1} & 25.5 \down{23.6} & 38.6 \down{35.4} & 55.3 \gap{3.6} & 0.1 \down{0.1} & 1.0 \up{0.2} & 2.5 \down{0.5} & 47.9 \gap{0.8} \\
\end{tabular}}
\vspace{-2pt}
\resizebox{\linewidth}{!}{
\begin{tabular}{l|cccc|cccc|cccc}
\Xhline{1.2pt}
\rowcolor{CadetBlue!20}
\textbf{Method} & $R\text{-}R@1\;\uparrow$ & $R\text{-}R@5\;\uparrow$ & $R\text{-}R@10\;\uparrow$ & $Comm.\;\downarrow$ & $R\text{-}R@1\;\uparrow$ & $R\text{-}R@5\;\uparrow$ & $R\text{-}R@10\;\uparrow$ & $Comm.\;\downarrow$ & $R\text{-}R@1\;\uparrow$ & $R\text{-}R@5\;\uparrow$ & $R\text{-}R@10\;\uparrow$ & $Comm.\;\downarrow$ \\
\Xhline{1.2pt}
\rowcolor{gray!10}
FedEraser & \second{91.6} \down{4.9} & 98.9 \down{1.1} & 99.4 \down{0.6} & 262.5 & 82.3 \down{10.0} & \second{94.6} \down{5.4} & 96.5 \down{3.5} & 196.9 & 84.4 \down{2.4} & 97.8 \down{1.2} & 99.4 \down{0.4} & 196.9 \\
FedRecover & 66.3 \down{30.2} & 91.9 \down{8.1} & 99.1 \down{0.9} & \best{47.2} \gap{215.3} & 71.1 \down{21.2} & 92.5 \down{7.5} & 96.5 \down{3.5} & \best{35.4} \gap{161.5} & 42.9 \down{43.9} & 69.5 \down{29.5} & 77.0 \down{22.8} & \best{35.4} \gap{161.5} \\
\rowcolor{gray!10}
Ferrari & 79.8 \down{16.7} & 98.9 \down{1.1} & 99.4 \down{0.6} & 72.2 \gap{190.3} & 85.7 \down{6.6} & 99.1 \down{0.9} & \second{99.6} \down{0.4} & 59.1 \gap{137.8} & 83.7 \down{3.1} & 98.0 \down{1.0} & 99.2 \down{0.6} & 59.1 \gap{137.8} \\
FedOSD & 64.9 \down{31.6} & 90.7 \down{9.3} & 99.5 \down{0.5} & 105.0 \gap{157.5} & 86.5 \down{5.8} & 98.9 \down{1.1} & 99.5 \down{0.5} & 78.8 \gap{118.1} & 83.3 \down{3.5} & 97.9 \down{1.1} & 99.3 \down{0.5} & 78.8 \gap{118.1} \\
\rowcolor{gray!10}
NoT & 31.1 \down{65.4} & 63.2 \down{36.8} & 78.0 \down{22.0} & \second{52.5} \gap{210.0} & 51.7 \down{40.6} & 86.9 \down{13.1} & 96.2 \down{3.8} & \second{39.4} \gap{157.5} & \second{44.3} \down{42.5} & 80.2 \down{18.8} & 89.5 \down{10.3} & \second{39.4} \gap{157.5} \\
SoUL & 66.5 \down{30.0} & \second{91.9} \down{8.1} & 96.4 \down{3.6} & 52.5 \gap{210.0} & \second{86.1} \down{6.2} & 99.0 \down{1.0} & 99.6 \down{0.4} & 39.4 \gap{157.5} & 83.7 \down{3.1} & 97.9 \down{1.1} & \second{99.3} \down{0.5} & 39.4 \gap{157.5} \\
\end{tabular}}
\end{table*}

\begin{table*}[!ht]
\centering
\setlength{\tabcolsep}{3pt}
\setlength{\aboverulesep}{0.3pt}
\setlength{\belowrulesep}{0.3pt}
\renewcommand{\arraystretch}{0.95}
\caption{Main comparison on TextCaps with CLIP-L/14 across three unlearning scenarios. Convention follows Table~\ref{tab:q1_flickr30k_clipb32}.}
\label{tab:q1_textcaps_clipl14}
\resizebox{\linewidth}{!}{
\begin{tabular}{l|cccc|cccc|cccc}
\Xhline{1.2pt}
\rowcolor{CadetBlue!20}
\textbf{Method} & \multicolumn{4}{c|}{\textbf{Client Unlearning}} & \multicolumn{4}{c|}{\textbf{Class Unlearning}} & \multicolumn{4}{c}{\textbf{Sample Unlearning}} \\
\cmidrule(lr){2-5} \cmidrule(lr){6-9} \cmidrule(lr){10-13}
\rowcolor{CadetBlue!20}
 & $F\text{-}R@1\;\downarrow$ & $F\text{-}R@5\;\downarrow$ & $F\text{-}R@10\;\downarrow$ & $MIA$ & $F\text{-}R@1\;\downarrow$ & $F\text{-}R@5\;\downarrow$ & $F\text{-}R@10\;\downarrow$ & $MIA$ & $F\text{-}R@1\;\downarrow$ & $F\text{-}R@5\;\downarrow$ & $F\text{-}R@10\;\downarrow$ & $MIA$ \\
\Xhline{1.2pt}
\rowcolor{gray!10}
FedEraser & 66.4 \down{66.4} & 85.0 \down{84.3} & 89.5 \down{87.8} & 63.1 \gap{12.8} & 54.9 \down{54.6} & 75.6 \down{74.0} & 82.2 \down{79.9} & \best{63.2} \gap{1.9} & \best{0.0} \up{0.3} & \best{0.7} \up{0.9} & \best{1.4} \up{0.8} & 49.1 \gap{3.1} \\
FedRecover & 60.7 \down{60.7} & 80.4 \down{79.7} & 86.0 \down{84.3} & 62.5 \gap{12.2} & 36.0 \down{35.7} & 59.8 \down{58.2} & 67.0 \down{64.7} & 62.8 \gap{1.5} & 0.3 & \second{0.7} \up{0.9} & \second{1.5} \up{0.7} & 47.6 \gap{4.6} \\
\rowcolor{gray!10}
Ferrari & 92.7 \down{92.7} & 99.7 \down{99.0} & 100.0 \down{98.3} & \best{63.2} \gap{12.9} & 76.8 \down{76.5} & 95.1 \down{93.5} & 98.3 \down{96.0} & \second{63.2} \gap{1.9} & 0.4 \down{0.1} & 1.1 \up{0.5} & 2.1 \up{0.1} & 49.7 \gap{2.5} \\
FedOSD & 91.2 \down{91.2} & 99.5 \down{98.8} & 100.0 \down{98.3} & \second{63.2} \gap{12.9} & 73.0 \down{72.7} & 94.2 \down{92.6} & 97.6 \down{95.3} & 63.2 \gap{1.9} & 0.4 \down{0.1} & 1.3 \up{0.3} & 2.9 \down{0.7} & 50.2 \gap{2.0} \\
\rowcolor{gray!10}
NoT & 51.5 \down{51.5} & 80.7 \down{80.0} & 89.9 \down{88.2} & 63.2 \gap{12.9} & 28.4 \down{28.1} & 53.4 \down{51.8} & 65.7 \down{63.4} & 63.2 \gap{1.9} & \second{0.0} \up{0.3} & 0.8 \up{0.8} & 1.7 \up{0.5} & 48.5 \gap{3.7} \\
\end{tabular}}
\vspace{-2pt}
\resizebox{\linewidth}{!}{
\begin{tabular}{l|cccc|cccc|cccc}
\Xhline{1.2pt}
\rowcolor{CadetBlue!20}
\textbf{Method} & $R\text{-}R@1\;\uparrow$ & $R\text{-}R@5\;\uparrow$ & $R\text{-}R@10\;\uparrow$ & $Comm.\;\downarrow$ & $R\text{-}R@1\;\uparrow$ & $R\text{-}R@5\;\uparrow$ & $R\text{-}R@10\;\uparrow$ & $Comm.\;\downarrow$ & $R\text{-}R@1\;\uparrow$ & $R\text{-}R@5\;\uparrow$ & $R\text{-}R@10\;\uparrow$ & $Comm.\;\downarrow$ \\
\Xhline{1.2pt}
\rowcolor{gray!10}
FedEraser & \second{90.9} \down{5.1} & 99.1 \down{0.9} & 99.6 \down{0.4} & 537.5 & 87.6 \down{10.9} & 98.4 \down{1.6} & 99.3 \down{0.7} & 403.1 & 91.1 \down{2.8} & 98.8 \down{0.7} & 99.7 \down{0.2} & 403.1 \\
FedRecover & 87.6 \down{8.4} & 98.3 \down{1.7} & 99.3 \down{0.7} & \best{96.8} \gap{440.7} & 84.6 \down{13.9} & 97.6 \down{2.4} & 99.1 \down{0.9} & \best{72.6} \gap{330.5} & 63.1 \down{30.8} & 82.4 \down{17.1} & 87.6 \down{12.3} & \best{72.6} \gap{330.5} \\
\rowcolor{gray!10}
Ferrari & 90.2 \down{5.8} & 99.0 \down{1.0} & 99.6 \down{0.4} & 147.8 \gap{389.7} & 87.6 \down{10.9} & \second{98.5} \down{1.5} & 99.2 \down{0.8} & 120.9 \gap{282.2} & 91.1 \down{2.8} & 98.9 \down{0.6} & 99.6 \down{0.3} & 120.9 \gap{282.2} \\
FedOSD & 90.9 \down{5.1} & 98.8 \down{1.2} & 99.5 \down{0.5} & 215.0 \gap{322.5} & 86.8 \down{11.7} & 98.3 \down{1.7} & 99.3 \down{0.7} & 161.2 \gap{241.9} & 91.4 \down{2.5} & 99.1 \down{0.4} & 99.7 \down{0.2} & 161.2 \gap{241.9} \\
\rowcolor{gray!10}
NoT & 63.6 \down{32.4} & 92.2 \down{7.8} & 99.4 \down{0.6} & \second{107.5} \gap{430.0} & 75.6 \down{22.9} & 98.0 \down{2.0} & \second{99.1} \down{0.9} & \second{80.6} \gap{322.5} & 62.5 \down{31.4} & \second{89.7} \down{9.8} & \second{96.3} \down{3.6} & \second{80.6} \gap{322.5} \\
SoUL & 90.2 \down{5.8} & 98.9 \down{1.1} & \second{99.6} \down{0.4} & 107.5 \gap{430.0} & \second{87.2} \down{11.3} & 98.4 \down{1.6} & 99.3 \down{0.7} & 80.6 \gap{322.5} & \second{90.7} \down{3.2} & 98.7 \down{0.8} & 99.6 \down{0.3} & 80.6 \gap{322.5} \\
\end{tabular}}
\end{table*}

\begin{table*}[!ht]
\centering
\setlength{\tabcolsep}{3pt}
\setlength{\aboverulesep}{0.3pt}
\setlength{\belowrulesep}{0.3pt}
\renewcommand{\arraystretch}{0.95}
\caption{Main comparison on TextCaps with SigLIP across three unlearning scenarios. Convention follows Table~\ref{tab:q1_flickr30k_clipb32}.}
\label{tab:q1_textcaps_siglip}
\resizebox{\linewidth}{!}{
\begin{tabular}{l|cccc|cccc|cccc}
\Xhline{1.2pt}
\rowcolor{CadetBlue!20}
\textbf{Method} & \multicolumn{4}{c|}{\textbf{Client Unlearning}} & \multicolumn{4}{c|}{\textbf{Class Unlearning}} & \multicolumn{4}{c}{\textbf{Sample Unlearning}} \\
\cmidrule(lr){2-5} \cmidrule(lr){6-9} \cmidrule(lr){10-13}
\rowcolor{CadetBlue!20}
 & $F\text{-}R@1\;\downarrow$ & $F\text{-}R@5\;\downarrow$ & $F\text{-}R@10\;\downarrow$ & $MIA$ & $F\text{-}R@1\;\downarrow$ & $F\text{-}R@5\;\downarrow$ & $F\text{-}R@10\;\downarrow$ & $MIA$ & $F\text{-}R@1\;\downarrow$ & $F\text{-}R@5\;\downarrow$ & $F\text{-}R@10\;\downarrow$ & $MIA$ \\
\Xhline{1.2pt}
\rowcolor{gray!10}
FedEraser & 67.5 \down{67.5} & 83.8 \down{82.9} & 90.0 \down{88.1} & 63.1 \gap{13.2} & 66.1 \down{66.0} & 83.4 \down{82.0} & 87.4 \down{84.8} & \best{61.4} \gap{5.7} & 0.3 & 1.0 \down{0.2} & 1.8 \up{0.2} & \second{50.4} \gap{0.6} \\
FedRecover & 56.3 \down{56.3} & 72.0 \down{71.1} & 78.7 \down{76.8} & 62.9 \gap{13.0} & 57.1 \down{57.0} & 74.6 \down{73.2} & 79.5 \down{76.9} & 60.1 \gap{4.4} & 0.2 \up{0.1} & 1.6 \down{0.8} & 2.3 \down{0.3} & 48.2 \gap{2.8} \\
\rowcolor{gray!10}
Ferrari & 94.5 \down{94.5} & 99.7 \down{98.8} & 100.0 \down{98.1} & \best{63.2} \gap{13.3} & 96.8 \down{96.7} & 99.9 \down{98.5} & 100.0 \down{97.4} & \second{61.4} \gap{5.7} & 0.3 & 1.1 \down{0.3} & 1.8 \up{0.2} & 50.0 \gap{1.0} \\
FedOSD & 92.5 \down{92.5} & 99.7 \down{98.8} & 100.0 \down{98.1} & \second{63.2} \gap{13.3} & 95.3 \down{95.2} & 99.6 \down{98.2} & 100.0 \down{97.4} & 61.4 \gap{5.7} & \best{0.0} \up{0.3} & 0.9 \down{0.1} & 2.1 \down{0.1} & 50.0 \gap{1.0} \\
\rowcolor{gray!10}
NoT & 54.7 \down{54.7} & 81.3 \down{80.4} & 88.7 \down{86.8} & 63.2 \gap{13.3} & 49.6 \down{49.5} & 79.7 \down{78.3} & 87.1 \down{84.5} & 61.4 \gap{5.7} & \second{0.0} \up{0.3} & 0.9 \down{0.1} & 1.7 \up{0.3} & 49.6 \gap{1.4} \\
\end{tabular}}
\vspace{-2pt}
\resizebox{\linewidth}{!}{
\begin{tabular}{l|cccc|cccc|cccc}
\Xhline{1.2pt}
\rowcolor{CadetBlue!20}
\textbf{Method} & $R\text{-}R@1\;\uparrow$ & $R\text{-}R@5\;\uparrow$ & $R\text{-}R@10\;\uparrow$ & $Comm.\;\downarrow$ & $R\text{-}R@1\;\uparrow$ & $R\text{-}R@5\;\uparrow$ & $R\text{-}R@10\;\uparrow$ & $Comm.\;\downarrow$ & $R\text{-}R@1\;\uparrow$ & $R\text{-}R@5\;\uparrow$ & $R\text{-}R@10\;\uparrow$ & $Comm.\;\downarrow$ \\
\Xhline{1.2pt}
\rowcolor{gray!10}
FedEraser & \second{94.5} \down{5.4} & 99.4 \down{0.6} & 99.8 \down{0.2} & 475.0 & 93.1 \down{5.4} & 99.0 \down{1.0} & 99.7 \down{0.3} & 356.2 & 97.2 \down{1.5} & 99.7 \down{0.3} & 99.9 \down{0.1} & 356.2 \\
FedRecover & 92.0 \down{7.9} & 98.4 \down{1.6} & 99.3 \down{0.7} & \best{85.5} \gap{389.5} & 92.1 \down{6.4} & 98.5 \down{1.5} & 99.3 \down{0.7} & \best{64.1} \gap{292.1} & 57.9 \down{40.8} & 76.7 \down{23.3} & 83.3 \down{16.7} & \best{64.1} \gap{292.1} \\
\rowcolor{gray!10}
Ferrari & 93.4 \down{6.5} & 99.3 \down{0.7} & 99.6 \down{0.4} & 130.6 \gap{344.4} & 92.7 \down{5.8} & \second{99.2} \down{0.8} & 99.5 \down{0.5} & 106.9 \gap{249.3} & 96.9 \down{1.8} & \second{99.6} \down{0.4} & 99.7 \down{0.3} & 106.9 \gap{249.3} \\
FedOSD & 93.8 \down{6.1} & 99.1 \down{0.9} & 99.8 \down{0.2} & 190.0 \gap{285.0} & \second{94.1} \down{4.4} & 99.4 \down{0.6} & \second{99.7} \down{0.3} & 142.5 \gap{213.7} & 96.5 \down{2.2} & 99.5 \down{0.5} & 99.8 \down{0.2} & 142.5 \gap{213.7} \\
\rowcolor{gray!10}
NoT & 64.2 \down{35.7} & 90.5 \down{9.5} & 94.3 \down{5.7} & \second{95.0} \gap{380.0} & 73.8 \down{24.7} & 98.7 \down{1.3} & 99.4 \down{0.6} & \second{71.2} \gap{285.0} & \second{68.6} \down{30.1} & 93.5 \down{6.5} & 99.6 \down{0.4} & \second{71.2} \gap{285.0} \\
SoUL & 93.8 \down{6.1} & 99.2 \down{0.8} & 99.7 \down{0.3} & 95.0 \gap{380.0} & 93.1 \down{5.4} & 99.1 \down{0.9} & 99.6 \down{0.4} & 71.2 \gap{285.0} & 96.5 \down{2.2} & 99.5 \down{0.5} & 99.8 \down{0.2} & 71.2 \gap{285.0} \\
\end{tabular}}
\end{table*}

\clearpage


\newpage

\end{document}